\documentclass[12pt]{article}
\usepackage{amsfonts}
\usepackage{amsmath}
\usepackage{amsthm}
\usepackage{amssymb}
\usepackage{braket}
\usepackage{graphicx}
\usepackage{hyperref}
\usepackage{color}
\usepackage{dsfont}

\usepackage{array}

\usepackage{dsfont}
\usepackage{authblk}
\usepackage[nottoc,notlot,notlof]{tocbibind}

\newcommand{\ftj}{\ensuremath{\mathsf{T}}} 
\newcommand{\Z}{\ensuremath{\mathbb{Z}}}
\newcommand{\tft}[1]{\ensuremath{Z[{#1}]}}
\newcommand{\vft}[1]{\ensuremath{V[{#1}]}}
\newcommand{\simp}{\ensuremath{\Delta}}
\newcommand{\man}{\ensuremath{X}}
\newcommand{\nan}{\ensuremath{Y}}
\newcommand{\N}{\ensuremath{\mathbb{N}}}
\newcommand{\U}{\ensuremath{\mathsf{U}}}
\newcommand{\triang}{\ensuremath{\mathcal{K}}}
\newcommand{\R}{\ensuremath{\mathbb{R}}}
\newcommand{\C}{\ensuremath{\mathbb{C}}}
\newcommand{\CP}{\ensuremath{\mathbb{CP}}}
\newcommand{\cellu}{\ensuremath{C}}
\newcommand{\hilb}{\ensuremath{\mathcal{H}}}
\newcommand{\tr}[1]{\ensuremath{\text{Tr}[{#1}]}}
\newcommand{\vect}[1]{\ensuremath{{\underline{\mathbf{#1}}}}}
\newcommand{\cat}{\ensuremath{\mathcal{C}}}
\newcommand{\config}{\ensuremath{S}}
\newcommand{\D}{\ensuremath{\mathcal{D}}}

\DeclareMathOperator{\sgn}{sgn}
\DeclareMathOperator{\cl}{cl}
\DeclareMathOperator{\str}{st}
\DeclareMathOperator{\link}{lk}
\DeclareMathOperator{\im}{im}
\DeclareMathOperator{\interior}{int}

\newcommand{\orient}[1]{\ensuremath{\sigma}(#1)}
\newcommand{\complex}{\ensuremath{K}}
\newcommand{\subc}{\ensuremath{J}}
\newcommand{\wo}{\ensuremath{\backslash}}
\newcommand{\comp}[1]{\ensuremath{\sigma_{#1}}}
\newcommand{\st}[1]{\ensuremath{{\str}_{#1}}}
\newcommand{\lk}[1]{\ensuremath{{\link}_{#1}}}
\newcommand{\A}{\ensuremath{\mathcal{A}}}
\newcommand{\G}{\ensuremath{\mathbb{G}}}

\newcommand{\openone}{\ensuremath{\mathds{1}}}

\newtheorem{lemma}{Lemma}
\newtheorem{proposition}{Proposition}

\newcommand{\mcK}{\mathcal{K}}

\title{Hamiltonian models for topological phases of matter in three spatial dimensions}
\author[a]{Dominic J. Williamson}
\author[b]{Zhenghan Wang}
\affil[a]{\small{\emph{Vienna Center for Quantum Technology, University of Vienna, Boltzmanngasse 5, 1090 Vienna, Austria}}}
\affil[b]{\small{\emph{Microsoft Research Station Q, CNSI Bldg. 2237, Santa Barbara, California 93106-6105, USA}}}
\affil[b]{\small{\emph{Department of Mathematics, University of California, Santa Barbara, California 93106-6105, USA}}}

\makeatletter

\begin{document}

\maketitle

\begin{abstract}
We present commuting projector Hamiltonian realizations of a large class of (3+1)D topological models based on mathematical objects called unitary G-crossed braided fusion categories. This construction comes with a wealth of examples from the literature of symmetry-enriched topological phases. The spacetime counterparts to our Hamiltonians are unitary state sum topological quantum fields theories (TQFTs) that appear to capture all known constructions in the literature, including the Crane-Yetter-Walker-Wang and 2-Group gauge theory models. We also present Hamiltonian realizations of a state sum TQFT recently constructed by Kashaev whose relation to existing models was previously unknown. We argue that this TQFT is captured as  a special case of the Crane-Yetter-Walker-Wang model, with a premodular input category in some instances.
\end{abstract}

\tableofcontents

\section{Introduction}
\label{intro}

Theoretically, a topological phase of matter (TPM)\cite{wen1990ground,einarsson} without any symmetry protection is an equivalence class of local Hamiltonians~\cite{hastings2005quasiadiabatic,bravyi2010topological,chen2010local,kitaev2003fault} whose low energy physics is modeled by a stable\footnote{Stable can be understood as no spontaneous symmetry breaking. The technical definition is $\tft{S^3\times S^1}=1$ which implies local operators act trivially within the ground space.}
 unitary topological quantum field theory (TQFT)~\cite{witten1988topological,segal1988definition,atiyah1988topological,moore1988polynomial,moore1989classical,kitaev2006anyons,nayak2008non}.   Given a realistic Hamiltonian it is generally difficult to determine which TPM it is in.  A fruitful approach is to reverse engineer Hamiltonians from known TQFTs.  Famous examples include Kitaev's toric code~\cite{kitaev2003fault,dijkgraaf1990topological} and Levin-Wen models~\cite{levin2005string,turaev1992state}.
 
  Physical TQFTs are local and this is usually formulated by a set of axioms known as the gluing formulas~\cite{walker1991witten,freedman2008picture,baez1995higher}.  A more explicit form of locality is a state sum construction~\cite{turaev1992state}.  
 It is generally believed that state sum TQFTs are in 1-1 correspondence with fully extended TQFTs and both admit local commuting projector Hamiltonian realizations. 
 However this conjecture has not been rigorously proven in full generality largely due to an inability to drop restrictive symmetry assumptions on the input data and higher ``j-symbols''. 
 While it is difficult to algebraically formalize the fully extended TQFT framework without these assumptions some progress has been made for the state sum case in Ref.\cite{higherdto}.
An interesting example that clearly violates the symmetry assumptions is Kashaev's state sum $(3+1)$-TQFT~\cite{kashaev2014simple,kashaev2015realizations}, whose \mbox{j-symbols} strongly depend on the linear ordering of the vertices of a $4$-simplex.

 A basic principle in the study of state sum TQFTs is that the behavior of a local $(n+1)$-TQFT restricted to a disk is encoded by some higher $n$-category $\cat$~\cite{walker1991witten,freedman2008picture,baez1995higher}. 
Furthermore the partition functions and a local commuting projector Hamiltonian can be constructed from $\cat$ as illustrated by the Turaev-Viro and Levin-Wen models~\cite{levin2005string,turaev1992state} (generalized Kitaev models~\cite{kitaev2003fault,dijkgraaf1990topological}) in two spatial dimensions.  
The physical excitations in this general picture should be described by a special $(n+1)$-category that is constructed by taking a generalized Drinfeld double of the $n$-category $\cat$~\cite{muger2003subfactors,muger2003subfactors2}.
The major deficiencies of this general approach are the lack of a good algebraic definition for the appropriate weak n-categories and an absence of examples. This is in stark contrast to the well developed theory of fusion categories relevant to the $(2+1)$D case~\cite{bakalov2001lectures,etingof2015tensor}.

In this paper our focus will be on three dimensional topological phases of matter and the associated $(3+1)$-TQFTs. 
 Many concrete constructions in three spatial dimensions have been proposed~\cite{crane1993categorical,crane1997state,kashaev2014simple,kashaev2015realizations,carter1999structures,crane1994four,yetter1993tqft,mackaay2000finite,mackaay1999spherical,kapustin2014topological,kapustin2013higher}, but all seemingly fall short of capturing the full intricacies of $(3+1)$D topological order. 
Looking back to $(2+1)$D we note that state sum TQFTs constructed from unitary fusion categories (UFCs) are sufficiently general to achieve all non chiral topological orders~\cite{levin2005string,turaev1992state} (from the higher categorical point of view these should be understood as $2$-categories that contain a single object). By analogy, we will refer to the most general input to $(3+1)$D state sum TQFTs as unitary fusion $2$-categories (which should correspond to $3$-categories that contain a single object).

 As a step towards a fully general unitary fusion $2$-category construction, we focus on a case that is populated by a rich class of examples originating from the algebraic theory of defects in $(2+1)$D symmetry enriched topological (SET) orders. These defects are described by mathematical objects known as unitary G-crossed braided fusion categories (UGxBFCs).
We build on the work of Ref.\cite{shawnthesis} in which a large class of state sum $(3+1)$-TQFTs were rigorously constructed from UGxBFCs that (with a small extension) seem to include almost all known examples of unitary state sum $(3+1)$-TQFTs.  In this paper we propose their Hamiltonian realization, generalizing the construction of Ref.\cite{walker2012}.  
 We note that a related construction of $(3+1)$-TQFTs based on a proposed definition of spherical 2-category was given in Ref.\cite{mackaay1999spherical} but it was shown in Ref.\cite{shawnthesis} that this definition was too restrictive to even capture the unitary $G$-crossed braided fusion categories (UGxBFC).

A family of models possibly outside this class of TQFTs are Kashaev's examples~\cite{kashaev2014simple}.  We also establish a Hamiltonian formulation of these $(3+1)$-TQFTs. 
Moreover we pose, and provide evidence for, the following conjecture: Kashaev's TQFTs are equivalent to a subset of 
Crane-Yetter TQFTs with input categories that may be premodular, in particular  $sVec$ for some instances (i.e. Walker-Wang models with fermionic string types).

\section{Background}
\label{background}

In this section we recount the definition of state sum TQFTs, and their associated Hamiltonians and tensor network ground states, before moving on to discuss two recently constructed classes of state sums; the UGxBFC and Kashaev's $\Z_N$ models. To facilitate the discussion  of these topics we first set up some basic conventions and terminologies that are used throughout the manuscript. 
From this point forward we also make free use of notation and constructions from combinatoric topology, for those unfamiliar with this topic the necessary points are summarized in Appendix.\ref{comb topology}.

We will define topological partition functions $\tft{X}$ on space-time manifolds $X$ of dimension $n+1$, and Hilbert spaces $\vft{Y}$ on spatial $n$-manifolds $Y$ that are equipped with triangulations~\cite{atiyah1988topological}.  But triangulation can mean many different things.  For highly non-symmetric state sum TQFTs,   we usually  need a simplicial triangulation of the spatial manifold $Y$, but only a $\Delta$-complex triangulation of the space-time manifold $X$.

A manifold $M$ has a simplicial triangulation if $M$ is homeomorphic to the realization or underlying space $|\mathcal{K}|$ of an abstract simplicial complex $\mathcal{K}$~\cite{hatcher606algebraic} (there are in fact extra technicalities, see the appendix). 
  A simplicial complex $\mathcal{K}$ is a collection of subsets of a finite set $V$, called the vertices of $\mathcal{K}$, with the property that if a subset $\sigma$ of $V$ is in $\mathcal{K}$ then all subsets of $\sigma$ are also in $\mathcal{K}$.  The subset $\sigma$ is called an $i$-simplex if $\sigma$ has $(i+1)$ vertices.  A geometric realization $|\mathcal{K}|$ of $\mathcal{K}$ can be built by associating each vertex $v\in V$ to a basis vector of the Euclidean space ${\mathbb{R}}^{|V|}$.  An important technical point for our construction is that we assume $V$ has a linear order.  Therefore our simplicial triangulations always have induced branching structures by drawing an arrow on each edge from the lower numbered vertex to the higher one.

A manifold $M$ of dimension $m$ has a $\Delta$-complex triangulation if the manifold $M$ is constructed from a finite collection of $m$-simplices, which are glued together along the $(m-1)$-dimensional faces by simplicial maps.  In particular a $\Delta$-complex triangulation of a manifold can have a single vertex, for example the torus $T^2$ with two triangles.

\subsection{State Sum TQFTs}

An oriented unitary $(n+1)$-TQFT $(V,Z)$ is technically a symmetric monoidal functor from $(n+1)Cob$ to $Vec_\C$~\cite{atiyah1988topological,segal1988definition}. 
This is nothing more than a very compact way of axiomatizing topological invariance of a field theory and can be broken down into a series of more elementary statements.
The TQFT assigns a topologically invariant partition function $\tft{X^{n+1}}\in \mathbb{C}$ to each oriented closed $(n+1)$-manifold
 and a finite dimensional Hilbert space $\vft{Y^n}$ to each oriented closed $n$-manifold. 
It furthermore assigns a linear map $\tft{X^{n+1}}:\vft{Y_0^n}\rightarrow\vft{Y_1^n}$ to an oriented $(n+1)$-manifold with boundary $\partial X^{n+1}=\overline{Y}_0\sqcup Y_1$. Unpacking the definition leads to gluing formulas that ensure topological invariance amongst other technical axioms. Additionally for a unitary TQFT orientation reversal is mapped to complex conjugation. We do not delve any further into the general definition here, instead we move on to the more specific notion of a state sum TQFT.

The most general possible construction of state sum TQFTs is not yet rigorously formalized, due to technicalities in proving independence from the choice of branching structure, we will present an overview here. 
 A state sum TQFT comes with a finite set of input labels $\{L^{(i)}\}_{i=0}^{n}$.  For any triangulation $\mcK$  of an $n+1$-manifold $X$, we first choose a linear ordering of the vertices $\mcK^{(0)}$ (in fact a local ordering or branching structure will suffice).  
 Then a configuration on the triangulated manifold is specified by a set of maps $s^{(i)}: \mcK^{(i)} \rightarrow L^{(i)}$, which color each $i$-simplex in $ \mcK^{(i)}$ with a label from $L^{(i)}$.  
Moreover, to capture the most general solutions we allow each label $l$ to have an associated \lq\lq quantum dimension" $d_l\in\C^\times$.  
Finally  in a configuration $s$ each labeled n+1-simplex $\simp$ is evaluated to a \lq\lq j-symbol'' $\ftj^{\orient{\simp}}_{s(\simp)}$, where $\orient{\simp}=\pm$ is the orientation of the n+1-simplex.  The partition function is then 
$$ \tft{X}=\sum_{s} \prod_{\simp_{n+1}}  \ftj_{s(\simp_{n+1})}^{\orient{\simp_{n+1}}} \frac{ \prod\limits_{\simp_{n-1}}  d_{s(\simp_{n-1})}   \prod\limits_{\simp_{n-3}}  d_{s(\simp_{n-3})} }{\prod\limits_{\simp_n}  d_{s(\simp_{n})} \prod\limits_{\simp_{n-2}}  d_{s(\simp_{n-2})}} \cdots $$
note our quantum dimensions may be rescaled compared to the usual definition from a unitary fusion category (UFC).

This prescription extends to triangulated manifolds with boundary $\partial X=\overline{Y}_0\sqcup Y_1$~\cite{higherdto}
\begin{align*}
\tft{X}=\sum_{s} \prod_{\simp_{n+1}}  \ftj_{s(\simp_{n+1})}^{\orient{\simp_{n+1}}} \frac{ \prod\limits_{\simp_{n-1}}  d^{c(\simp_{n-1})}_{s(\simp_{n-1})}   \prod\limits_{\simp_{n-3}}  d^{c(\simp_{n-3})}_{s(\simp_{n-3})} }{\prod\limits_{\simp_n}  d^{c(\simp_{n})}_{s(\simp_{n})} \prod\limits_{\simp_{n-2}}  d^{c(\simp_{n-2})}_{s(\simp_{n-2})}} \cdots  
\bigotimes_{\simp_j \in Y_1}\ket{s(\simp_j)} \bigotimes_{\simp_i \in Y_0} \bra{s(\simp_i)}
\end{align*}
where $c(\simp_i)=\frac{1}{2}$ if $\simp_i\in \partial X$ and 1 if it is in the interior.
Hilbert spaces $\vft{Y}$ are then defined to be the support subspace of the linear operator $\tft{Y\times I}$ for a triangulation of $Y\times I$ that matches $Y$ on both boundaries. 
Topological invariance of the state sum is more precisely an invariance of $Z$ under piecewise linear (PL) homeomorphisms on the $(n+1)$-manifold. 
PL homeomorphic manifolds are related by a sequence of local bistellar flips of the triangulation, drawn from a finite set known as the Pachner moves~\cite{pachner1991pl}. This recasts topological invariance of the state sum into a finite set of equations that the j-symbol must satisfy~\cite{higherdto}, corresponding to retriangulations of the $(n+1)$-ball. 
This guarantees the partition function is independent of the choice of triangulation, moreover one must show the partition function is independent of the choice of vertex ordering.

\subsection{Hamiltonians, Tensor Network Ground States and PEPO Symmetries}

State sum TQFTs have a natural interpretation as tensor networks~\cite{turaev1992state,koenig2010quantum,higherdto} (see Ref.\cite{bridgeman2016hand} for an introduction to tensor networks). Copies of a single tensor are associated to each simplex and contracted according to how the simplices are glued together. Topological invariance of the discrete partition functions is ensured if and only if the simplex tensor satisfies tensor equations corresponding to the Pachner moves~\cite{pachner1991pl}.
There is a standard (although not widely known) construction to obtain a local real-space renormalization group (RSRG) fixed point commuting projector Hamiltonian that stabilizes the vector space of a state sum TQFT on a triangulated surface, generalizing that of Levin \& Wen~\cite{levin2005string}.
The Hamiltonian is given by
\begin{align}
\label{tftham}
H=\sum_v \openone - \tft{v'* {\cl\,\st{v}}}
\end{align}
where $v'$ is a copy of vertex $v$ after one imaginary time step. For a definition of the operations $\{*,\cl,\st{} \}$ see Appendix.\ref{comb topology}.
All the aforementioned properties of the Hamiltonian follow from the Pachner move invariance of the simplex tensor~\cite{higherdto}.
A projected entangled pair state (PEPS)~\cite{Fannes92,VerstraeteMurgCirac08,VerstraeteCirac06,GarciaVerstraeteWolfCirac08} representation of a ground state wave function (GSWF) on a triangulated manifold $(\nan,\triang)$ is given by $ \tft{v_0* \triang}$. 
Provided $\tft{\{v_0,v_1\} *Y}>0$ (which implies $\dim \vft{\nan}=\tft{\nan\times S^1}>0$) the resulting state is nonzero. 
This PEPS is a frustration free ground state as it satisfies $\tft{v'* {\cl\,\st{v}}} \tft{v_0* \triang} = \tft{v_0* \triang}$ following from the Pachner moves. 
This PEPS has a projected entangled pair operator (PEPO) symmetry which can be ``pulled through'' the virtual level~\cite{schuch2010peps,buerschaper2014twisted,burak2014characterizing,williamson2014matrix,williamson2016fermionic}, this is also ensured by the Pachner moves. 
The symmetry is indicative of topological order in the model via a bulk boundary correspondence given by taking the double of the algebra of tensor network operators on the boundary to construct the emergent physical excitations~\cite{bultinck2015anyons,lan2014topological,fidkowski2009string}.
The framework also yields a multiscale entanglement renormalization ansatz (MERA)~\cite{vidal2007entanglement} description of the ground space constructed by taking a triangulated identity bordism $(\nan\times I, \triang')$ such that the triangulation at the space manifold $(\nan,0)$ reduces to the physical lattice $\triang$ and we pick a minimal triangulation $\triang''$ of $(\nan,1)$ at the `top' of the MERA corresponding to the far IR scale. Then upon fixing a vector containing the fully coarse grained topological information $\ket{t}$ the MERA is given by 
\begin{align}
\tft{\triang'}\ket{t}
\end{align}
For physical lattice models it is important that the Hamiltonians output by our construction are Hermitian. This is ensured by a sufficient condition on the underlying tensor, namely that it is symmetric under simultaneous complex conjugation and orientation reversal. 
We note in the framework of Ref.\cite{higherdto} there is some technicality involved when dealing with weight functions associated to objects on lower dimensional strata of the triangulation.
\begin{figure}[th]
\center
{{
 \includegraphics[height=0.45\linewidth]{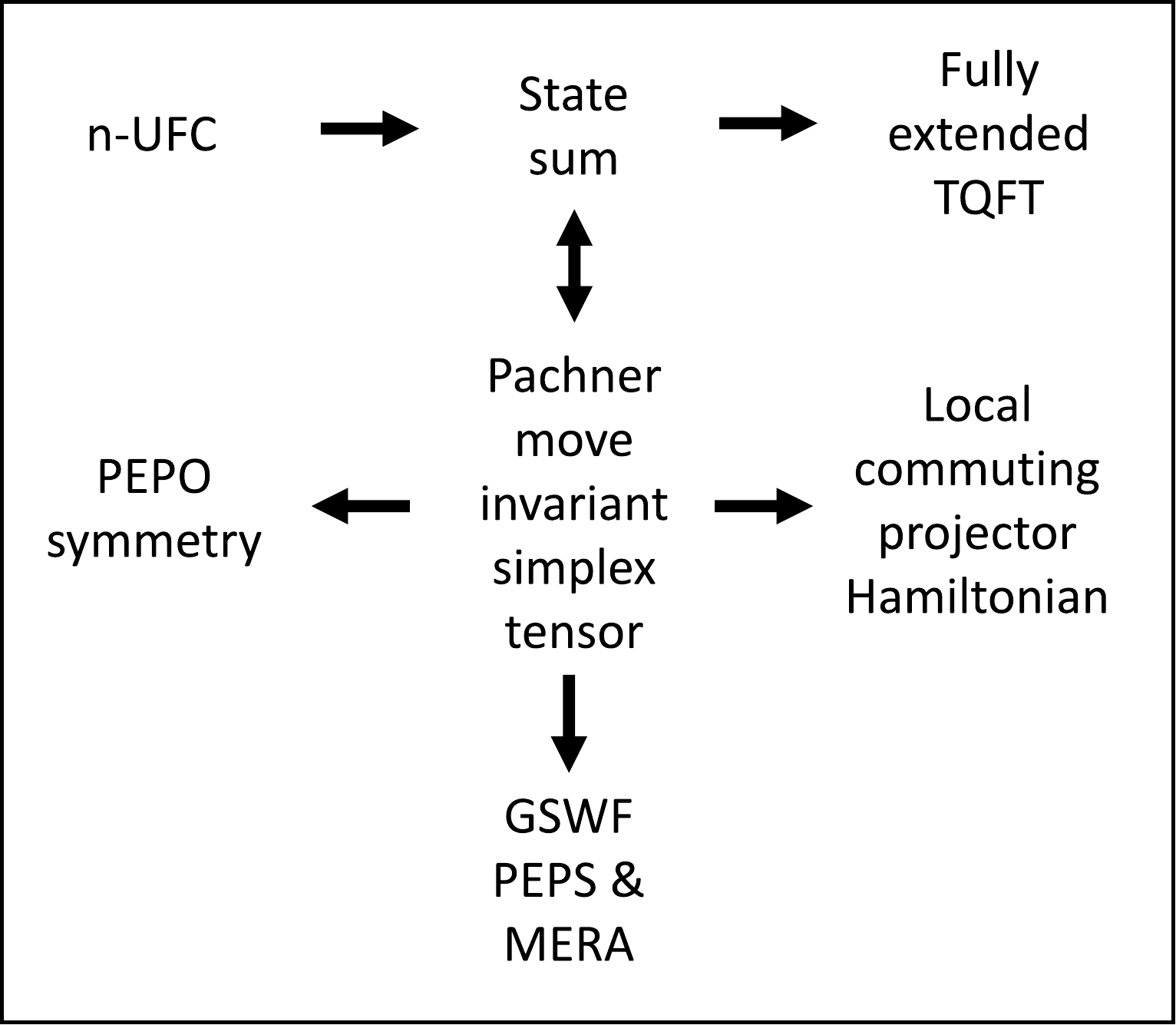} }} 
 \caption{Summary of the results in Ref.\cite{higherdto}.}
\label{pmistsum}
\end{figure}

\subsection{Review of the UGxBFC TQFT}

A new class of $(3+1)$-TQFTs was constructed in Ref.\cite{shawnthesis} from unitary $G$-crossed braided fusion categories (UGxBFC) $\mathcal{C}_G^{\times}$~\cite{turaev2000homotopy,kirillov2004g,etingof2009fusion,barkeshli2014symmetry}, where $G$ is a finite group. These UGxBFCs can be though of as special unitary fusion $2$-categories.  When $G$ is trivial, a UGxBFC $\mathcal{C}_G^{\times}$ reduces to a premodular category and the resulting TQFT is the Crane-Yetter TQFT~\cite{crane1997state,crane1993categorical} whose Hamiltonian realization is described in Ref.\cite{walker2012}.  In general, a UGxBFC has $|G|$ sectors and the trivial sector is always a premodular category.  

\subsubsection{UGxBFC}

A UGxBFC~\cite{turaev2000homotopy,kirillov2004g,etingof2009fusion,barkeshli2014symmetry} can roughly be thought of as a spherical fusion $1.5$-category---it is a spherical fusion category with a $G$-crossed braiding, hence it does not seem to be a totally general spherical fusion $2$-category. That being said there is no satisfactory agreed upon definition of a spherical $2$-category in the literature. While a definition was proposed in Ref.\cite{mackaay1999spherical} it was shown in Ref.\cite{shawnthesis} that this was too restrictive to even capture the UGxBFCs. Moving forward the UGxBFCs constitute a very important class of unitary fusion 2-categories as they provide a huge family of nontrivial examples, which have proved otherwise hard to come by.

 For our purpose, the most convenient way to define a $G$-crossed braided fusion category is through a collection of symbols $\{N_{ab}^c,F^{abc}_{d}, R^{ab}_c, \varkappa_a, U_g(a, b ; c), \eta_x(g,h)\}$~\cite{barkeshli2014symmetry}. 
 This extends the description of a unitary premodular category through a collection of symbols $\{N_{ab}^c,F^{abc}_{d}, R^{ab}_c, \varkappa_a\}$~\cite{kitaev2006anyons}.

A UGxBFC is an abstract description of point like defects of a symmetry group $G$ in a gapped phase of matter in $(2+1)$D. Each defect carries a flux $g\in G$ but there may be multiple topologically distinct defects carrying the same $G$-flux, this is described by a $G$-graded category 
$$\cat_G=\bigoplus_{g\in G} \cat_g $$
where each simple object is contained in some sector $a\in \cat_g$. We follow the notation of Ref.\cite{barkeshli2014symmetry} and use $a_g$ as shorthand for $a\in\cat_g$. Defects can be fused by physically bringing them together, this is described by a set of multiplicities $N_{ab}^c$ counting the number of ways $a$ and $b$ can fuse to $c$. In particular the fusion $a \times b \rightarrow c$ is admissible iff $N_{ab}^c\neq 0$. 
The fusion should respect the grading, i.e. 
$$ a_g\times b_h = \sum_{c\in \cat_G} N_{ab}^c c=  \sum_{c\in \cat_{gh}} N_{ab}^c c_{gh} .$$
The $\cat_1$ sector is closed under fusion and contains the unique vacuum object that fuses trivially with everything else, thus forming a fusion subcategory.

The fusion of three defects is not strictly associative, two different fusion paths with result $d$ are related by an $F$-symbol associator 
$(a\times b) \times c \xrightarrow{F^{abc}_d} a\times (b \times c)$
more precisely 
\begin{align}\label{Fsymb}
\vcenter{\hbox{
\includegraphics[scale=1]{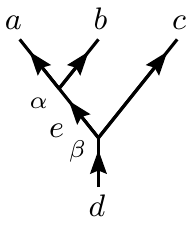}}}
\ = \sum_{f\mu\nu} [F^{abc}_{d}]_{e\alpha\beta}^{f\mu\nu} 
\vcenter{\hbox{
\includegraphics[scale=1]{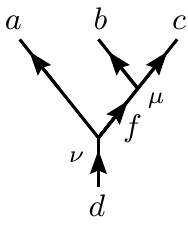}}}
\end{align}
Unitarity of the fusion category requires $[({F^{abc}_{d}})^{-1}]_{f\mu\nu}^{e\alpha\beta}= ([F^{abc}_{d}]_{e\alpha\beta}^{f\mu\nu})^*$. 
For this associator to be consistent all paths between a pair of diagrams must agree, this is guaranteed by the well known pentagon equation
\begin{align}
\sum_{\delta} [F^{fcd}_{e}]_{g\beta\gamma}^{l\delta\nu}  [F^{abl}_{e}]_{f\alpha\delta}^{k\lambda\mu} 
= \sum_{h\sigma\psi\rho}  [F^{abc}_{g}]_{f\alpha\beta}^{h\sigma\psi}  [F^{ahd}_{e}]_{g\sigma\gamma}^{k\lambda\rho}  [F^{bcd}_{k}]_{h\psi\rho}^{l\mu\nu}
\end{align}
this is depicted diagrammatically in Fig.\ref{pentagon}.
 \begin{figure}[ht]
\center
{{
 \includegraphics[scale=1]{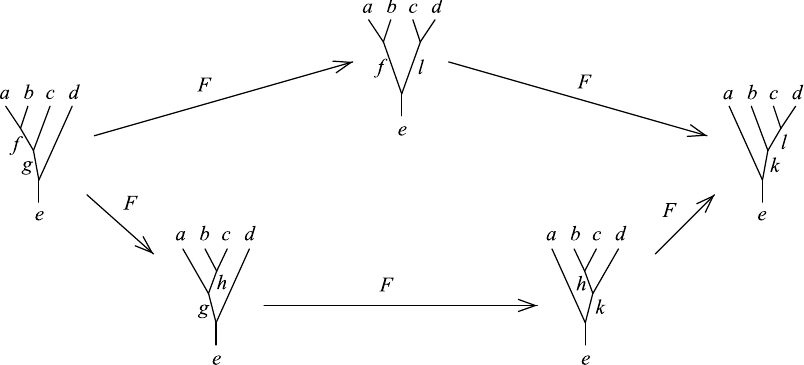}}} 
\caption{The pentagon equation.}
\label{pentagon}
\end{figure}

The group G has an action on simple objects denoted by $^g a_h\in \cat_{^g h}$ where $^g h = ghg^{-1}$. Each simple object $a_g$ has a unique conjugate $\bar{a}\in\cat_{\bar{g}}$ that can fuse together to give the vacuum, where $\bar{g}=g^{-1}$.  Flipping the direction of an edge is equivalent to conjugating the charge label
\begin{align}\label{edgeflip}
\vcenter{\hbox{
\includegraphics[scale=1]{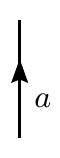}}}
= \
\vcenter{\hbox{
\includegraphics[scale=1]{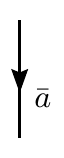}}}
   .
\end{align}
The element $[F^{a\bar{a}a}_a]^{1}_1=\frac{\varkappa_a}{d_a}$ consists of a quantum dimension which arises from popping a bubble
\begin{align}\label{quantdim}
d_a=d_{\bar{a}}
 =
 \vcenter{\hbox{
\includegraphics[scale=1]{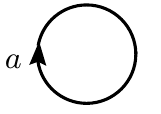}}}
\end{align}
and a Frobenius-Shur (FS) indicator which arises when a cup and cap are canceled
\begin{align}\label{fsind}
 \vcenter{\hbox{
\includegraphics[scale=1]{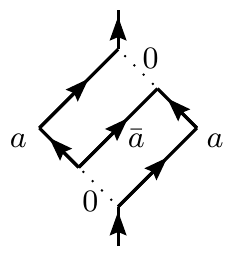}}}
 = \varkappa_a 
  \vcenter{\hbox{
\includegraphics[scale=1]{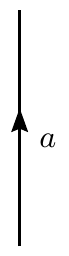}}}
   . 
\end{align}
Note the FS indicator can be fixed to $1$ via a gauge transformation unless $a=\bar{a}$ in which case $\varkappa_a=\pm 1$.

The total quantum dimension of $\cat$ is $\D^2=\sum\limits_{a\in\cat} d_a^2$ and similarly for each sector $\cat_g$, $\D_g^2=\sum\limits_{a\in\cat_g} d_a^2$. The nonempty sectors form a subgroup $H\leq G$ and satisfy $\D_h=\D_1$ for $h\in H$. Note all defects in a given sector are related by fusion with objects in $\cat_1$.

The physical defects appear at the end of branch cuts and can be dragged around by adiabatically deforming the Hamiltonian. This leads to braided worldlines of the defects attached to worldsheets of the branch cuts. We follow the convention of Ref.\cite{barkeshli2014symmetry} and depict the worldsheets going into the page. The worldsheet of a defect worldline acts on other defects which pass behind it.  
The G-crossed braiding is defined by several pieces of data, the $R$-symbol 
\begin{align}\label{Rsymb}
  \vcenter{\hbox{
\includegraphics[scale=1]{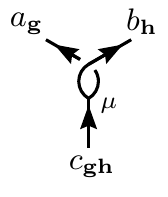}}}
 = \sum_\nu [R_{c_{gh}}^{a_g b_h}]^\nu_\mu 
  \vcenter{\hbox{
\includegraphics[scale=1]{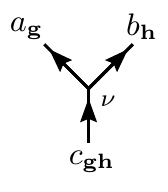}}}
\end{align}
and the $U$ and $\eta$ symbols, which arise due to symmetry actions as a defect is slid over or under a fusion vertex,
\begin{align}\label{usymb}
  \vcenter{\hbox{
\includegraphics[scale=.9]{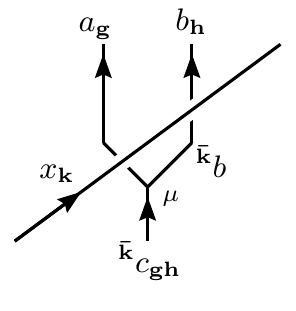}}}
\  &= \sum_\nu [U_k(a,b;c)]^\nu_\mu 
  \vcenter{\hbox{
\includegraphics[scale=.9]{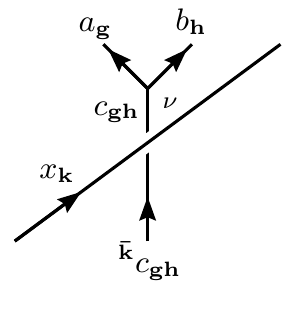}}}
\end{align}
\begin{align}
 \label{nsymb}
  \vcenter{\hbox{
\includegraphics[scale=.9]{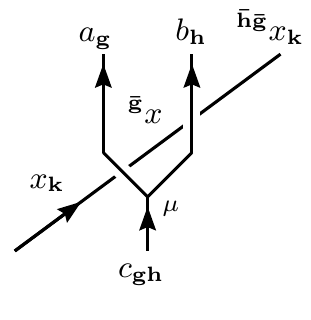}}}
\ &= \eta_x(g,h) 
  \vcenter{\hbox{
\includegraphics[scale=.9]{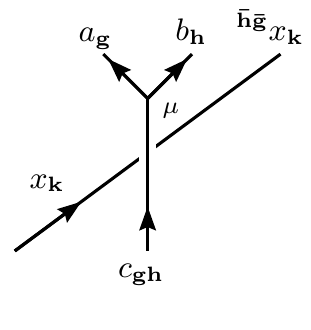}}}
 .
\end{align}
$U$ corresponds to the action of the global symmetry on topological degrees of freedom, while $\eta_x$ corresponds to the 2-cocycle of the projective representation carried by $x$.

For the data $F^{abc}_{d}, R^{ab}_c, U_g(a, b ; c), \eta_x(g,h)$ to define a consistent UGxBFC $\cat^\times_G$ the symbols must satisfy a number of conditions. The $F$-symbols must satisfy the pentagon equation in Fig.\ref{pentagon}. 
Equating the two different paths in Fig.\ref{forkcrossing} yields a constraint corresponding to the action of $(kl)\bar{l}\, \bar{k}$ being trivial (technically a natural isomorphism)
\begin{align}\label{Uetaconsist}
[\kappa_{k,l}(a,b;c)]_\mu^\nu &= \sum_{\alpha\beta} [U^{-1}_k(a,b;c)]^\alpha_\mu [U^{-1}_l({}^{\bar{k}}a,{}^{\bar{k}}b;{}^{\bar{k}}c)]^\beta_\alpha 
\nonumber \\
[U_{kl}(a,b;c)]_\beta^\nu &= \frac{\eta_a(k,l) \eta_b(k,l)}{\eta_{c}(k,l)} \delta_\mu^\nu .
\end{align}
 \begin{figure}[th]
\center
{{
 \includegraphics[width=0.6\linewidth]{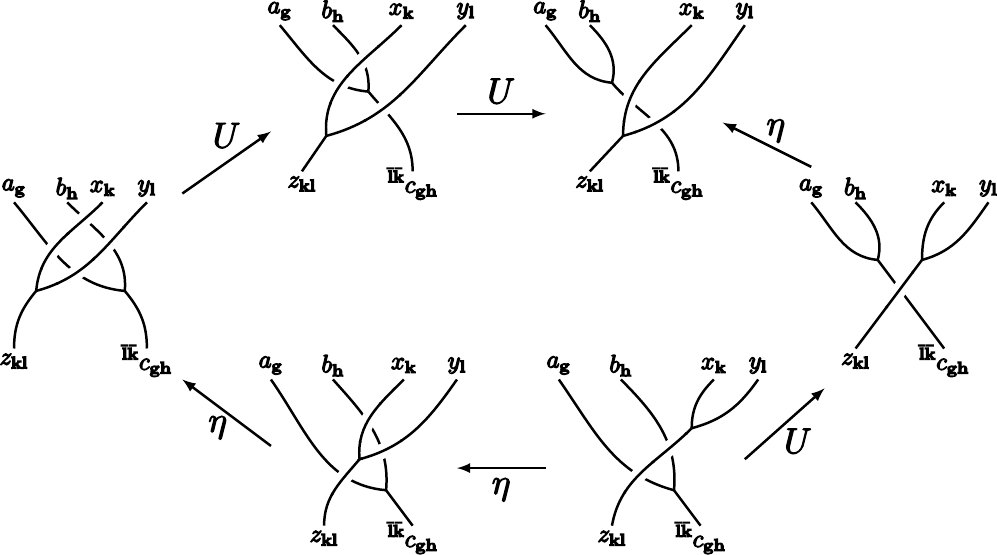}}} 
\caption{Consistency of the global symmetry action and projective phases.} 
\label{forkcrossing}
\end{figure}
Associativity of the group action $klm$ yields a further constraint on $\kappa_{k,l}$
$$  \kappa_{l,m}({}^{\bar{k}}a,{}^{\bar{k}}b;{}^{\bar{k}}c)  \kappa_{k,lm}(a,b;c) = \kappa_{k,l}(a,b;c) \kappa_{kl,m}(a,b;c) . $$
Consistency of fusion and $\eta$ leads to the equation 
$$  \eta_{{}^{\bar{g}}x}(h,k) \eta_x(gh,k ) = \eta_x(g,h) \eta_x(gh,k) $$
by equating the two paths in Fig.\ref{etaFconsistency}. This ensures the symmetry fractionalization is not anomalous and can be realized in a standalone $(2+1)$D system, corresponding to the vanishing of a $H^3(G,\mathcal{A})$ obstruction where $\mathcal{A}$ is the group of abelian anyons. 
 \begin{figure}[th]
\center
\begin{minipage}{.49\textwidth}
\center
\begin{align*}
  \vcenter{\hbox{
\includegraphics[scale=.9]{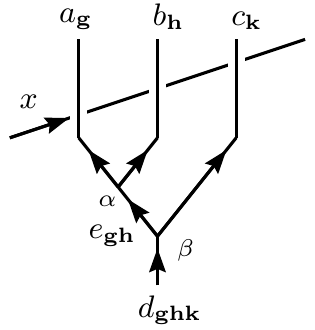}}}
\quad \rightarrow \quad
  \vcenter{\hbox{
\includegraphics[scale=.9]{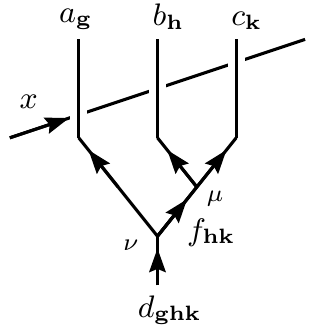}}}
\\
 \downarrow \hspace{4.3cm} \downarrow \hspace{1.32cm}
\\
  \vcenter{\hbox{
\includegraphics[scale=.9]{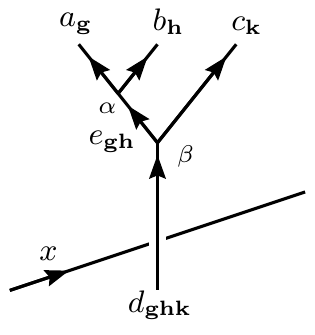}}}
\quad \rightarrow \quad
  \vcenter{\hbox{
\includegraphics[scale=.9]{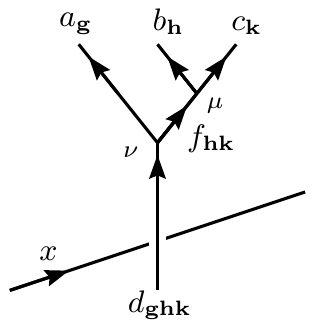}}}
\end{align*}
\end{minipage}
\begin{minipage}{.49\textwidth}
\center
\begin{align*}
  \vcenter{\hbox{
\includegraphics[scale=.9]{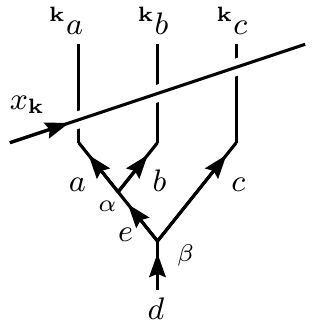}}}
\quad \rightarrow \quad
  \vcenter{\hbox{
\includegraphics[scale=.9]{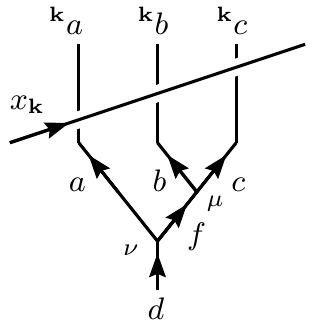}}}
\\
 \downarrow \hspace{4.3cm} \downarrow \hspace{1.32cm}
\\
  \vcenter{\hbox{
\includegraphics[scale=.9]{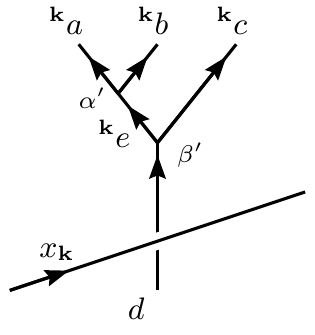}}}
\quad \rightarrow \quad
  \vcenter{\hbox{
\includegraphics[scale=.9]{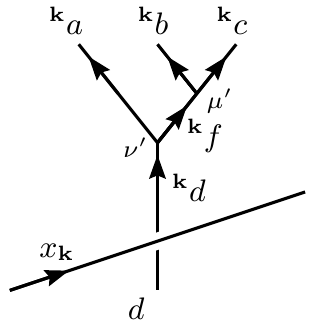}}}
\end{align*}
\end{minipage}

\caption{Consistency of fusion with: $\eta$ (left), and $U$ (right).}
\label{etaFconsistency}
\end{figure}
Similarly consistency of fusion and $U$ yields the equation 
\begin{align}\label{UFconsist}
\sum_{\alpha'\beta'\mu'\nu'} [U_k({}^{\bar{k}}a,{}^{\bar{k}}b;{}^{\bar{k}}e)]^{\alpha'}_\alpha [U_k({}^{\bar{k}}e,{}^{\bar{k}}c;{}^{\bar{k}}d)]^{\beta'}_\beta  [F^{{}^{\bar{k}}a{}^{\bar{k}}b {}^{\bar{k}}c}_{{}^{\bar{k}} d}]_{{}^{\bar{k}} e\alpha ` \beta'}^{{}^{\bar{k}}f \mu' \nu `}
 [U^{-1}_k({}^{\bar{k}}b,{}^{\bar{k}}c;{}^{\bar{k}}f)]^{\mu}_{\mu `}  [U^{-1}_k({}^{\bar{k}}a,{}^{\bar{k}}f;{}^{\bar{k}}d)]^{\nu}_{\nu `}
\nonumber \\
 = [F^{abc}_{d}]_{e\alpha\beta}^{f\mu\nu} 
\end{align}
which corresponds to a symmetry condition on $F$ under the group action.

The Yang-Baxter equation is no longer a strict equality in a UGxBFC and leads to a consistency equation between the dragging of a string over or under a crossing
\begin{align}\label{URconsist}
&\frac{\eta_{{}^{\bar{k}}a }({}^{\bar{k}} h,k)}{\eta_{{}^{\bar{k}} a}(k,h)} 
\sum_{\mu'\nu'}  [U_k({}^{{k}}b,{}^{{k}\bar{h}}a;{}^{{k}}c)]^{\mu'}_{\mu} [R_{{}^{k}c}^{{}^{k}a {}^{k}b}]^{\nu'}_{\mu'}
[U^{-1}_k({}^{{k}}a,{}^{{k}}b;{}^{{k}}c)]^{\nu}_{\nu `}
= [R_{c}^{ab}]^{\nu}_{\mu}
\end{align}
which is a symmetry condition on $R$ under the group action.
 \begin{figure}[th]
\center
\begin{align*}
  \vcenter{\hbox{
\includegraphics[scale=1]{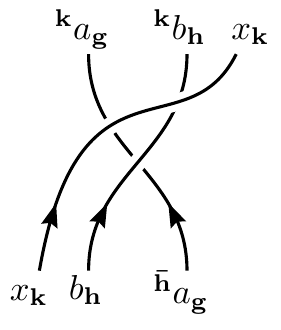}}}
\vcenter{\hbox{
\includegraphics[scale=1]{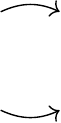}
}}
\hspace{.8mm}
  \vcenter{\hbox{
\includegraphics[scale=1]{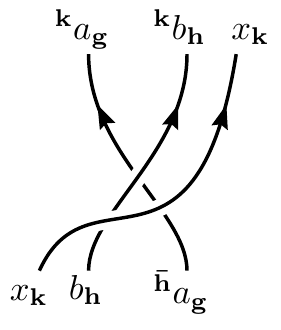}}}
\end{align*}
\caption{Yang-Baxter relation in a UGxBFC.}
\label{URconsistency}
\end{figure}

Finally there are consistency relations between $F$ and $R$ (also involving $U$ and $\eta$) known as the heptagon equations, generalizing the well known hexagon equations for a UBFC, one each for right and left handed braiding (see Fig.\ref{heptagon})  as follows
\begin{align}\label{heptagoneqns}
\sum_{\lambda\gamma}  [R_{e}^{ac}]^{\lambda}_{\alpha}  [F^{ac {}^{\bar{k}}b}_{d}]_{e\lambda\beta}^{m\gamma\nu}   [R_{m}^{bc}]^{\mu}_{\gamma}
&=\sum_{f\sigma\delta\theta\psi}  [F^{c {}^{\bar{k}} a {}^{\bar{k}} b}_{d}]_{e\alpha\beta}^{{}^{\bar{k}}f\delta\sigma}  [U_k(a,b;f)]^\theta_\delta [R_{d}^{fc}]^{\psi}_{\sigma} [F^{abc}_{d}]_{f\theta\psi}^{m\mu\nu} 
\\
\sum_{\lambda\gamma}  [(R_{e}^{ca})^{-1}]^{\lambda}_{\alpha}  [F^{a{}^{\bar{g}}cb}_{d}]_{e\lambda\beta}^{m\gamma\nu}   [(R_{m}^{{}^{\bar{g}}cb})^{-1}]^{\mu}_{\gamma}
&=\sum_{f\sigma\delta\psi}  [F^{cab}_{d}]_{e\alpha\beta}^{f\delta\sigma}  \eta_c(g,h)  [(R_{d}^{cf})^{-1}]^{\psi}_{\sigma} [F^{ab{}^{\bar{h}\bar{g}}c}_{d}]^{g\mu\nu}_{f\delta\psi} 
\end{align}
where the defect sectors are determined by $a_g,b_h,c_k$. 
 \begin{figure}[th]
\center
 \includegraphics[width=0.47\linewidth]{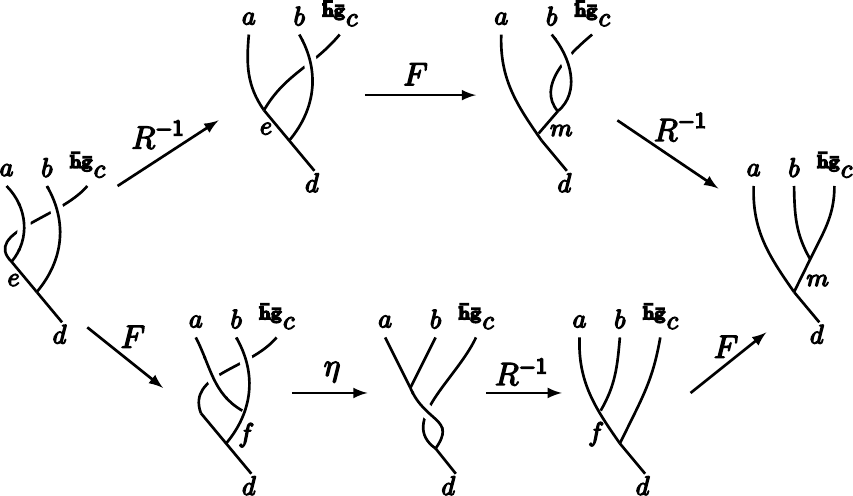} \hspace{.7cm}
 \includegraphics[width=0.47\linewidth]{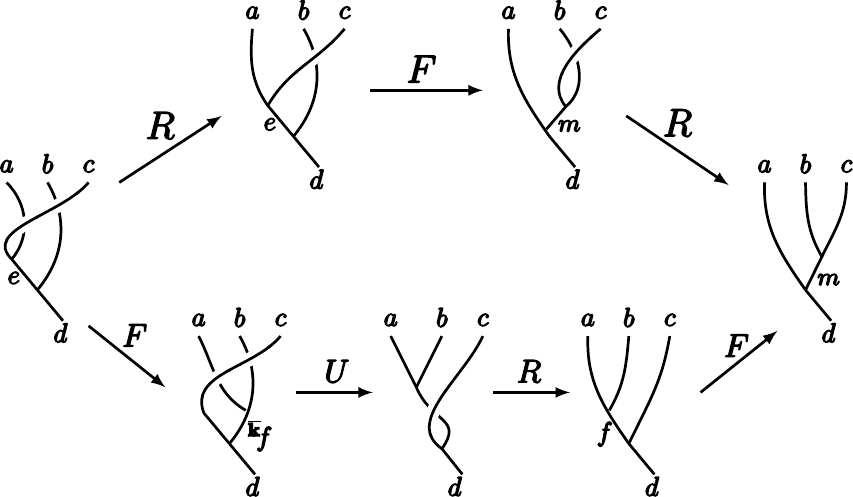}
\caption{The left and right handed heptagon equations of a UGxBFC, respectively.}
\label{heptagon}
\end{figure}
For a unitary GxBFC it is required that $ [(R_{c}^{ab})^{-1}]^{\nu}_{\mu} = ([R_{c}^{ab}]^{\mu}_{\nu})^*$. 
Note the trivial sector $\cat_1$ of a UGxBFC is itself a UBFC as the heptagon equations reduce to the hexagon equations in that case. 

The consistency equations for a UGxBFC guarantee it is not anomalous and can be realized in a stand-alone $(2+1)$D system. Not all group actions on UBFCs can be extended to a UGxBFC as some are anomalous. The anomalies lie in $H^4(G,U(1))$, which is related to weakening the pentagon equation Fig.\ref{pentagon}, and $H^3(G,\mathcal{A})$ which is related to weakening the fractonalization constraint Fig.\ref{etaFconsistency}.

\subsubsection{Example: $\Z_3$ Tambara-Yamagami Category}

A simple example of a UGxBFC known as the $\Z_3$ Tambara-Yamagami category can be constructed from a $\Z_3^{(1)}$ anyon theory $\{0,1,2\}$ with a $\Z_2$ symmetry that permutes $1$ and $2$. This theory is also known as $SU(3)_1$ which has conformal charge $c=2$ and a $\Z_2$ particle-hole symmetry. Physically, this theory describes the topological order of a sector of the $\nu=\frac{1}{3}$ Laughlin FQH state.

The UGxBFC consists of two sectors $\cat^\times_{\Z_2}=\cat_0\oplus \cat_1$.
The $\Z_3^{(1)}$ UBFC constitutes the $\cat_0$ sector and is defined by the fusion rules $N_{ab}^c=\delta_{a+b=c \ \text{mod}\, 3}$, trivial $F$ symbols and braiding $R^{ab}_{a+b}=e^{2\pi i ab/3}$.
The nontrivial sector contains only a single defect $\cat_1=\{ \sigma \}$. The fusion rules are thus
\begin{align*}
\sigma\times a = a \times \sigma = \sigma
\\
\sigma\times \sigma = \sum_{a\in\cat_0} a .
\end{align*}
The anyons in $\cat_0$ each have quantum dimension 1, hence $d_\sigma = \sqrt{3}$. 
The nontrivial $F$ symbols are then given by
\begin{align*}
[F^{a\sigma b}_{\sigma}]^{\sigma}_\sigma= [F^{\sigma a \sigma}_{b}]^{\sigma}_\sigma = \chi(a,b) 
\\
[F^{\sigma \sigma \sigma}_{\sigma}]^{b}_a = \frac{\varkappa_\sigma}{\sqrt{3}} {\chi(a,b)}^*
\end{align*}
where $\chi(a,b)=e^{2\pi i ab/3}$ is a symmetric bi-character.
The $G$-crossed braidings involving $\sigma$ are determined by
\begin{align*}
R^{\sigma a}_\sigma = U_1(\sigma,\sigma,a)(-1)^a e^{- \pi i a^2/3}, \quad  R^{a \sigma }_\sigma = (-1)^a e^{- \pi i a^2/3}
\\
R^{\sigma \sigma}_a = \gamma (-1)^a e^{\pi i a^2/3}, \quad \gamma^2 = \frac{\varkappa_\sigma}{\sqrt{3}} \sum_{a\in\cat_0} (-1)^a e^{-\pi i a^2/3}
\end{align*}
where $U_1(\sigma,\sigma,a)=\pm 1 $ and $\varkappa_\sigma=\pm 1$ are choices which yield slightly different UGxBFC extensions of $\cat_0$, note $\eta$ is trivial in all cases.

\subsubsection{State Sum from UGxBFC}

The data of a UGxBFC $\cat^\times_G$ can be used as input to generate a family of $(3+1)$D state sum TQFTs~\cite{shawnthesis} generalizing the Crane-Yetter-Walker-Wang model. 
The label set is as follows $L^{(1)}=G,\ L^{(2)}=\cat^\times_G,\ L^{(3)}=\text{Hom}(\cat^\times_G \otimes \cat^\times_G,\cat^\times_G \otimes \cat^\times_G)$, where we are abusing notation by using $\cat^\times_G$ to denote the set of equivalence classes of simple objects. 
That is each edge is labeled by a group element $g$, each triangle is labeled by a defect $a$ and each tetrahedron is labeled by a triple $(a,\mu,\nu)$ of a defect and two degeneracy labels. The only configurations that have nonzero contributions to the state sum must satisfy the following constraints between the labels on the different strata: the defect on a simplex $012$ must satisfy $a_{012}\in\cat_{(dg)_{012}}$ where $(dg)_{012}=\bar{g}_{02}g_{01}g_{12}$ and the defect labels on the faces and body of a tetrahedra $0123$ are subject to the constraints $N_{a_{\hat{1}}\, {}^{\bar{g}_{23}}a_{\hat{3}}}^{a_{0123}}\neq 0 \neq N_{a_{\hat{2}}a_{\hat{0}}}^{a_{0123}}$ (then $\mu,\nu$ correspond to degeneracy labels of these fusion spaces).

The $15j$-symbols are given by evaluating diagrams in the UGxBFC shown in Fig.\ref{15jsymbols}. The choice of diagram is determined by the configuration $s$ and the vertex ordering on a pentachoron. We use the compressed notation $F^{abc}_{d;ef}=[F^{abc}_d]_e^f$, explicit evaluation of the diagrams in Fig.\eqref{15jsymbols} yields
\begin{align}
\ftj^+_{s(01234)} =& \sum_{a,b} d_b F^{024,234, {}^{\bar{34}\cdot\bar{23}}012}_{b;0234,a}  \eta^{-1}_{012}(\bar{34},\bar{23}) 
R_{a}^{{}^{\bar{24}}012 , 234 } (F^{024, {}^{\bar{24}}012 , 234 }_{b; a, 0124})^{-1} 
F^{014,124,234}_{b;0124,1234}
(F^{014,134,{}^{\bar{34}}123}_{b;1234,0134})^{-1}
\nonumber \\
 &\ F^{034,{}^{\bar{34}}013,{}^{\bar{34}}123}_{b;0134,{}^{\bar{34}}0123 } 
U_{\bar{34}}(023,{}^{\bar{23}}012;0123)  U^{-1}_{\bar{34}}(013,123;0123)  
 (F^{034,{}^{\bar{34}}023,{}^{\bar{34}\cdot\bar{23}}012}_{b;{}^{\bar{34}}0123,0234} )^{-1}
\\
\ftj^-_{s(01234)} =&   \sum_{a,b} d_b (F^{024,234, {}^{\bar{34}\cdot\bar{23}}012}_{b;a,0234})^{-1}  \eta_{012}(\bar{34},\bar{23}) 
(R_{a}^{{}^{\bar{24}}012 , 234 })^{-1}  F^{024, {}^{\bar{24}}012 , 234 }_{b;0124,a }
(F^{014,124,234}_{b;1234,0124})^{-1}
 F^{014,134,{}^{\bar{34}}123}_{b;0134,1234}
\nonumber \\
&\ (F^{034,{}^{\bar{34}}013,{}^{\bar{34}}123}_{b;{}^{\bar{34}}0123,0134})^{-1} 
U^{-1}_{\bar{34}}(023,{}^{\bar{23}}012;0123)  U_{\bar{34}}(013,123;0123)  
F^{034,{}^{\bar{34}}023,{}^{\bar{34}\cdot\bar{23}}012}_{b;0234,{}^{\bar{34}}0123} 
\end{align}
where each label $\simp_i$ is to be read as $s(\simp_i)$, we have omitted the explicit writing of $s$ for brevity.   
 \begin{figure}[ht]
\center
 \includegraphics[width=0.25\linewidth]{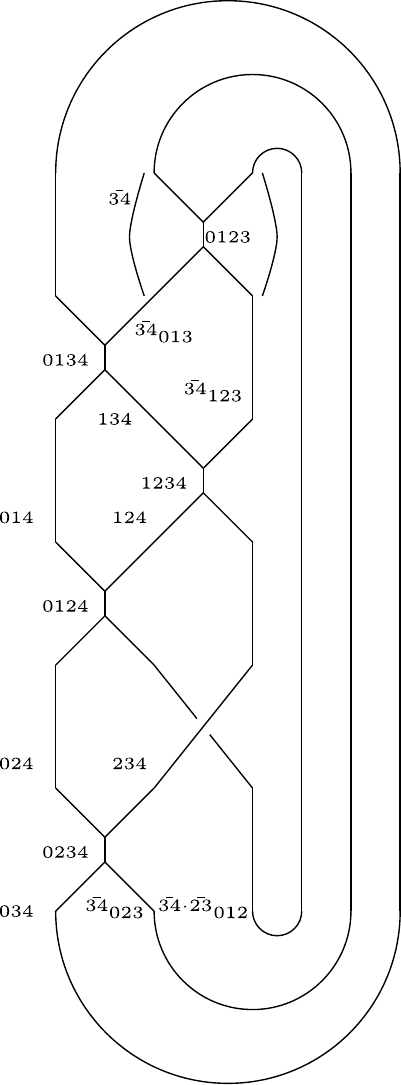}
\hspace{2.5cm}
  \includegraphics[width=0.25\linewidth]{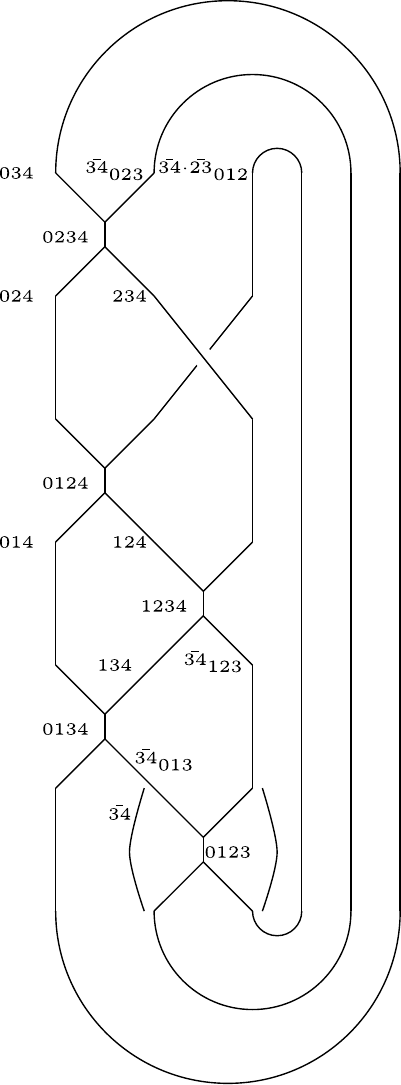} 
\caption{Diagrams in the UGxBFC that define the $15j$-symbols $\ftj^+_{s(01234)}$ (left), and $\ftj^-_{s(01234)}$ (right)~\cite{shawnthesis}. It is intended that $\simp_i$ should be read as $s(\simp_i)$.}
\label{15jsymbols}
\end{figure}

These $15j$-symbols, together with the quantum dimensions, define the state sum partition function 
$$ \tft{X}=\sum_{s} \prod_{\simp_{4}}  \ftj_{s(\simp_{4})}^{\orient{\simp_{4}}} \frac{ \prod\limits_{\simp_{2}}  d_{s(\simp_{2})}   \prod\limits_{\simp_{0}}  D^2/|G| }{ \prod\limits_{\simp_{3}}  d_{s(\simp_{3})} \prod\limits_{\simp_1} D^2 }  $$
where $D$ is the total quantum dimension.
It was shown in Ref.\cite{shawnthesis} that the partition function is a topological invariant, i.e. does not depend on the choice of vertex ordering or triangulation. The latter condition is guaranteed by the equalities $\tft{\subc}=\tft{\simp_5 \backslash \subc}$ for all 4-subcomplexes $\subc \subseteq \simp_5$.

This state sum captures many known constructions as special cases: 
\begin{itemize}
\item For $G=\{1\}$ the trivial group the UGxBFC $\cat_G^\times=\cat_1$ is a regular UBFC and the $15j$-symbols match the construction of Crane and Yetter~\cite{crane1993categorical,crane1997state}, hence $\tft{X}= \text{CY}_{\cat_1}(X)$ the Crane-Yetter partition function for $\cat_1$. This implies our lattice models include the Walker-Wang models~\cite{walker2012}.
\item Another special case constructed from a categorical group (or crossed module) yields Yetter's homotopy 2-type invariant~\cite{yetter1993tqft}. 
Note this inclusion implies that our lattice models capture those of Ref.\cite{bullivant2016topological}.
Categorical groups are in 1-1 correspondence with crossed modules, we follow Ref.\cite{shawnthesis} and use the latter to build a UGxBFC. A crossed module is specified by the data $(G,H,\rho,a)$ for $G,H$ finite groups, $\rho:H\rightarrow G$ a group morphism and $a:G\times H\rightarrow H$ a group action of $G$ on $H$ subject to the conditions $\rho(a_g (h))={}^{g}\rho(h)$ and $a_{\rho(h')} (h)= {}^{h'} h$. 
A UGxBFC $\cat(G,H,\rho,a)=\bigoplus\limits_{g\in G}\cat_g$ is constructed from the data as follows: the simple elements are $h\in H$, the grading is given by $h\in \cat_{\rho{h}}$, fusion is given by multiplication in $H$, the $G$-action is given by $a$, the braiding and $F$ symbols are trivial. The $15j$-symbols are simply delta conditions  on the configuration being admissible (i.e. face and tetrahedra constraints satisfied) and the partition function satisfies $\tft{X}=\text{Y}(X)$ for the Yetter invariant~\cite{yetter1993tqft} constructed from the categorical group corresponding to $(G,H,\rho,a)$.
In the special case that $H=\{1\}$ is trivial, $\rho$ \& $a$ are also trivial, then the triangle constraints become the flatness condition $dg=0$ and the partition function recovers the untwisted Dijkgraaf-Witten theory for $G$, $\tft{X}=\text{DW}_G(X)$~\cite{dijkgraaf1990topological}.
\item The case where the only nontrivial sector is $\cat_1$, a UBFC, the triangle constraints imply the flatness condition $dg=0$. If in addition the group action is trivial the group and anyon degrees of freedom decouple and the partition function factors into a product of CY and untwisted DW theory, $\tft{X}=\text{DW}_G(X) \text{CY}_{\cat_1}(X)$.
\end{itemize}

For the trivially graded case it is possible to introduce additional cocycle data to produce variants of the UGxBFC: 
\begin{itemize}
\item Since the state sum only involves flat $G$-connections the 4 group variables $g_{i,i+1}$ fully specify the $G$ configuration on a pentachoron. One may modify the $15j$-symbol by multiplication with a 5-cocycle phase factor $[\omega] \in H^4(G,U(1))$ to produce $\hat \ftj^\pm_{s(\simp_4)} = \ftj^\pm_{s(\simp_4)} \omega^{\pm 1}(g_{01},g_{12},g_{23},g_{34})$ which will give rise to a topologically invariant state sum. 
If in addition the group action is trivial the resulting partition function is given by a product of CY and twisted DW theories $\tft{X}=\text{DW}^\omega_G(X) \text{CY}_{\cat_1}(X)$.
\item In the case that $\cat_1=H$ an abelian group, with trivial $F$ and $R$ symbols, and a group action $a:G\times H \rightarrow H$ the tetrahedra constraint reads $(d^ah)_{0123}=a_{g_{23}}(h_{\hat 1})+h_{\hat 3} - h_{\hat 0} - h_{\hat 2}=0$. 
One may introduce a twisted 3-cocycle $[\beta]\in H^{3}_a(G,H)$ modifying the flatness condition to $(d^ah)_{0123}=\beta(g_{01},g_{12},g_{23})$. The $15j$-symbols are then delta conditions on the flatness of a 2-group connection defined by the data $\G=(G,H,a,\beta)$, following Ref.\cite{kapustin2013higher}. Furthermore one may introduce a multiplicative cocycle $[\omega]\in H^4(B\G,U(1))$ to produce a new $15j$-symbol $\hat \ftj^\pm_{s(\simp_4)} = \ftj^\pm_{s(\simp_4)} \omega^{\pm 1}(s(\simp_4)\,)$. The partition function then recovers the twisted 2-group gauge theory $\tft{X}={\text{2-DW}}^\omega_\G(X)$. 
\end{itemize}

No rigorous connection has been established between the aforementioned $H^3\ \&\ H^4$ cocycles and the $H^3(G,\A)\ \&\  H^4(G,U(1))$ anomaly classes of an SET.  In these cases the SET theory $\cat_1$ and group action cannot be extended to a UGxBFC. However we conjecture it will remain possible to construct an extension of the UGxBFC with a single sector whose $15j$-symbol has an intrinsic $H^3\ \&\ H^4$ anomaly. We defer the details of this to future work~\cite{anomalousugxbfc}. Note the possibility of adding an arbitrary $H^3\ \&\ H^4$ as discussed above suggest the intrinsic anomalies should be thought of as torsors. 
Furthermore we speculate that it should be possible to construct a unitary fusion $2$-category generalizing the UGxBFC that describes extension of an anomalous SET to nontrivial defect sectors, and this construction may yield a state sum with $15j$-symbols generalizing those of the UGxBFC.

It is not yet known how strong the UGxBFC state sum invariant is. Considering the special cases it contains it is clearly sensitive to homotopy 2-type and also the second Stiefel–-Whitney class (as the anyons can be fermionic). It is unclear if the theory is able to detect any smooth structure, while it is known from general considerations that it cannot be sensitive to all smooth structure~\cite{freedman2005universal}.

\subsection{Review of Kashaev's TQFT}

Kashaev's family of state sum TQFTs~\cite{kashaev2014simple,kashaev2015realizations} are indexed by a natural number $N\in\N$, they are specified by a tensor
\begin{align}\label{Qtens}
Q=N^{-\frac{1}{2}} \sum_{k,l,m\in\Z_N} \omega^{km} \ket{k}\bra{k+m}\otimes\ket{l}\bra{l+m}\otimes\ket{m}
\end{align}
where $\omega\in \U(1)$ is a primitive $N$th root of unity. 
\begin{figure}[ht]
\center
{{
 \includegraphics[height=0.23\linewidth]{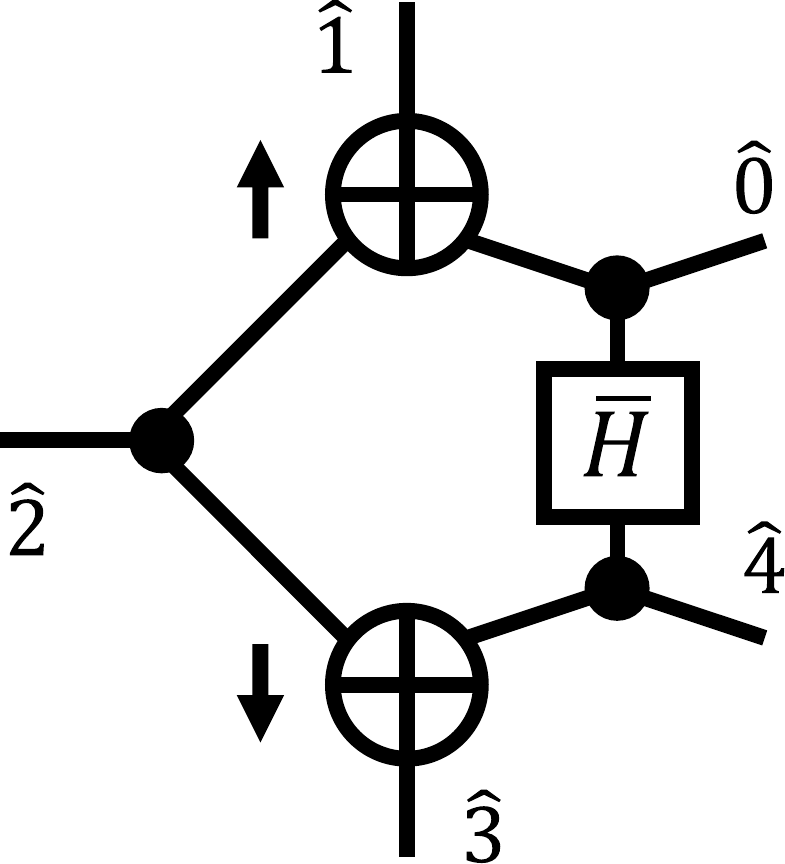} \hspace{2.5cm} \includegraphics[height=0.23\linewidth]{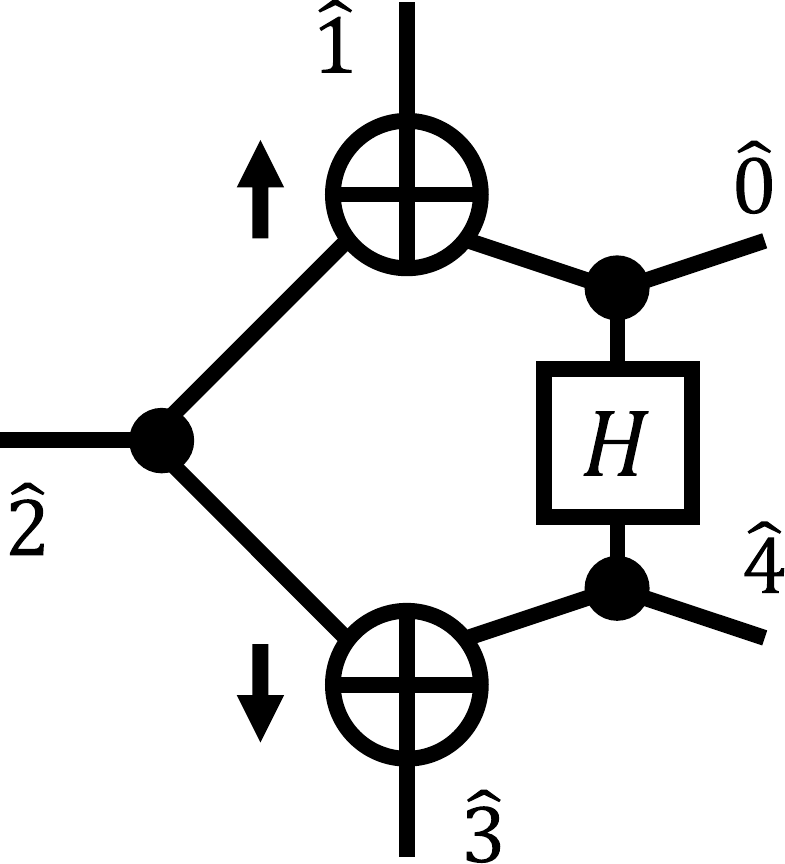}}} 
\caption{Kashaev's Q tensor (left) and its conjugate (right).}
\label{Qpic}
\end{figure}
A tensor $Q$ ($Q^\dagger$) is assigned to each 4-simplex in the triangulation where the orientation induced by the branching structure matches (is opposite to) the ambient orientation of the triangulated manifold.
An $N$ level qudit is associated to each 3-simplex facet of the 4-simplex tensor, they are written in the order given by taking the dual of the vertex order inherited from the branching structure. 
The full partition function on a triangulated 4-manifold $(\man,\triang)$ is given by the evaluation of the tensor network times the normalization factors $N^{\frac{3}{2}\chi(\man)}$ and $N^{-|\triang_0|}$, where $\chi$ is the Euler characteristic and $\triang_0$ is the number of vertices in the triangulation (with those on a boundary counted as half).

Kashaev has shown in Ref.\cite{kashaev2014simple} that the $Q$ tensors satisfy the Pachner move invariance conditions, together with the Hermitian property of the tensors (i.e. parity reversal induces complex conjugation) this implies the construction outlined in Ref.\cite{higherdto} gives rise to a local commuting projector Hamiltonian. 
The dimension of the (unfrustrated) zero energy eigenspace of the Hamiltonian on a spatial manifold $\nan$ is given by $\tft{\nan \times {S}^1}$. In particular the Hamiltonian is frustration free iff $\tft{\nan \times {S}^1}$ is nonzero. Partition functions have been calculated for a number of manifolds by Kashaev~\cite{kashaev2014simple} and for $S^{1}\times T^3$ by the authors. These results are summarized in Table.\ref{kashaevpf} and show that the TQFT is stable (i.e. the Hamiltonian does not exhibit spontaneous symmetry breaking). 
\begin{table}[ht]
\begin{center}
\begin{tabular}{| c | c | c | c | }
\hline 
$\man$ & $\chi(\man)$ & $\sigma(\man)$ & $Z_N[\man]$
\\ \hline
$S^4$ & 2 & 0 & 1
\\ \hline
$S^2\times S^2$ & 4 & 1 & $\frac{3+(-1)^N}{2}$
\\ \hline
$\C P^2$ & 3 & 1 & 
$\frac{1}{\sqrt{N}} \sum\limits_{k=1}^{N} \omega^{k^2}$
\\ \hline
$S^3\times S^1$ & 0 & 0 & 1
\\ \hline
$S^2\times S^1 \times S^1$ & 0 & 0 & $\frac{3+({-}1)^N}{2}$
\\ \hline
$S^1\times S^1 \times S^1 \times S^1 $ & 0 & 0 & \textcolor{red}{$(\frac{3+(-1)^N}{2})^3$}
\\ \hline
\end{tabular}	
\end{center}
\caption{Partition functions of Kashaev's TQFT}
\label{kashaevpf}
\end{table}
The final element of the table (highlighted in red) is the result of a new calculation and yields the ground state degeneracy on the 3 torus for all $N$. 
Furthermore we have 
$$\left|\tft{\C P^2}\right|^2=1+({-}1)^{\frac{N}{2}}\frac{1+({-}1)^N}{2}. $$ 
Hence the data computed for the Kashaev theory is consistent with a modular CYWW model (an invertible TQFT) for $N$ odd, and a premodular CYWW with transparent subcategory: $\Z_2$ with trivial braiding (topological order equivalent to toric code) for $N=0 \mod 4$, and $sVec$ for $N=2 \mod 4$ (as $Z[\C P^2]=0$ the partition function can be seen to depend on spin structure in this case).

We conjecture that the general construction of Kashaev~\cite{kashaev2015realizations} is dual to the Crane-Yetter TQFT, in a similar fashion to the duality between Kuperberg's $(2+1)$-manifold invariants~\cite{kuperberg1991involutory} and the Turaev-Viro TQFT~\cite{turaev1992state}.

\section{Lattice Model for Kashaev's TQFT}
\label{kashaev}

In this section we apply the framework developed in Ref.\cite{higherdto} to produce a translation invariant local commuting projector Hamiltonian for the theory on a particular 3-dimensional lattice. 

\subsection{The Hamiltonian}

With the $Q$ tensor from Eq.\eqref{Qtens} one can explicitly construct a local commuting projector Hamiltonian of the form in Eq.\eqref{tftham} on any 3-manifold equipped with a triangulation and branching structure $(\nan,\triang )$ by following the recipe outlined in Ref.\cite{higherdto}. 
For concreteness we work with the body centered cubic (BCC) triangulation of $T^3$ or $\R^3$ which is dual to a tessellation by regular 4-permutohedra (also known as truncated octahedra). The branching structure is given for $\R^3$ by the rule that all edges not orthogonal to the $\hat{z}$ axis are oriented along the $+\hat{z}$ direction, while those in an $xy$-plane point along the $+\hat{x}$ or $+\hat{y}$ direction (note these edges are always parallel to one of these axes). This also induces a branching structure on the triangulation of $T^3$. Note this branching structure preserves the full translational symmetry along each of the spatial axes in addition to a translation symmetry generated by $(\frac{1}{2},\frac{1}{2},\frac{1}{2})$.
\begin{figure}[ht]
\center
{{
 \includegraphics[height=0.25\linewidth]{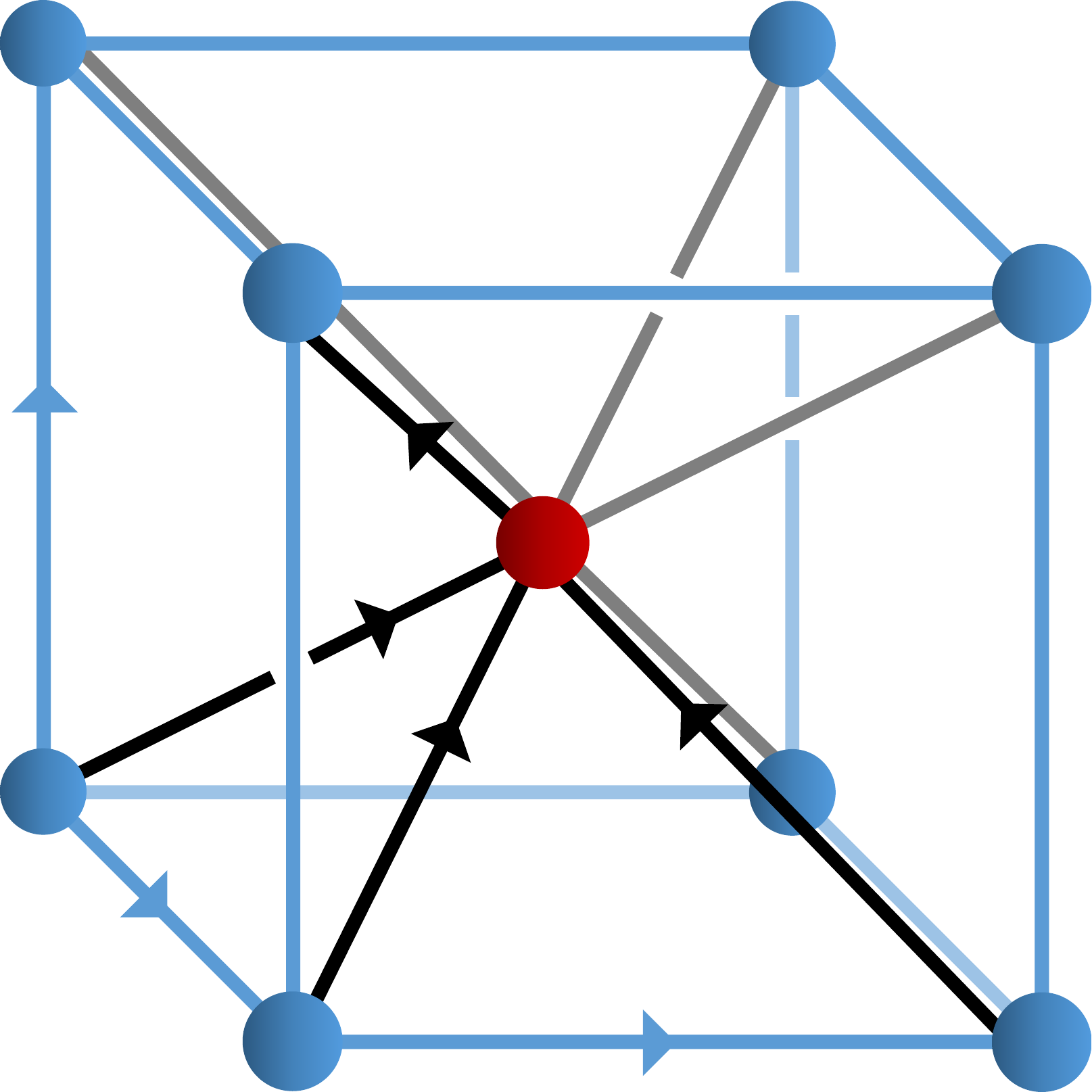}}} 
\caption{Triangulation and branching structure within a unit cell.}
\label{branch}
\end{figure}
Working on the permutohedron cellulation dual to the triangulation the model is defined on a Hilbert space consisting of a qudit degree of freedom for each vertex $\hilb = \bigotimes_{v\in \cellu} \C^N_v$. The Hamiltonian is a sum of identical terms each acting on 24 qudits in the boundary of a different permutohedron. 
To explicitly evaluate the Hamiltonian produced by the recipe of Ref.\cite{higherdto} we specify a numbering of the vertices on the boundary of a permutohedron depicted in Fig.\ref{Hordering}.
\begin{figure}[ht]
\center
{{
 \includegraphics[height=0.35\linewidth]{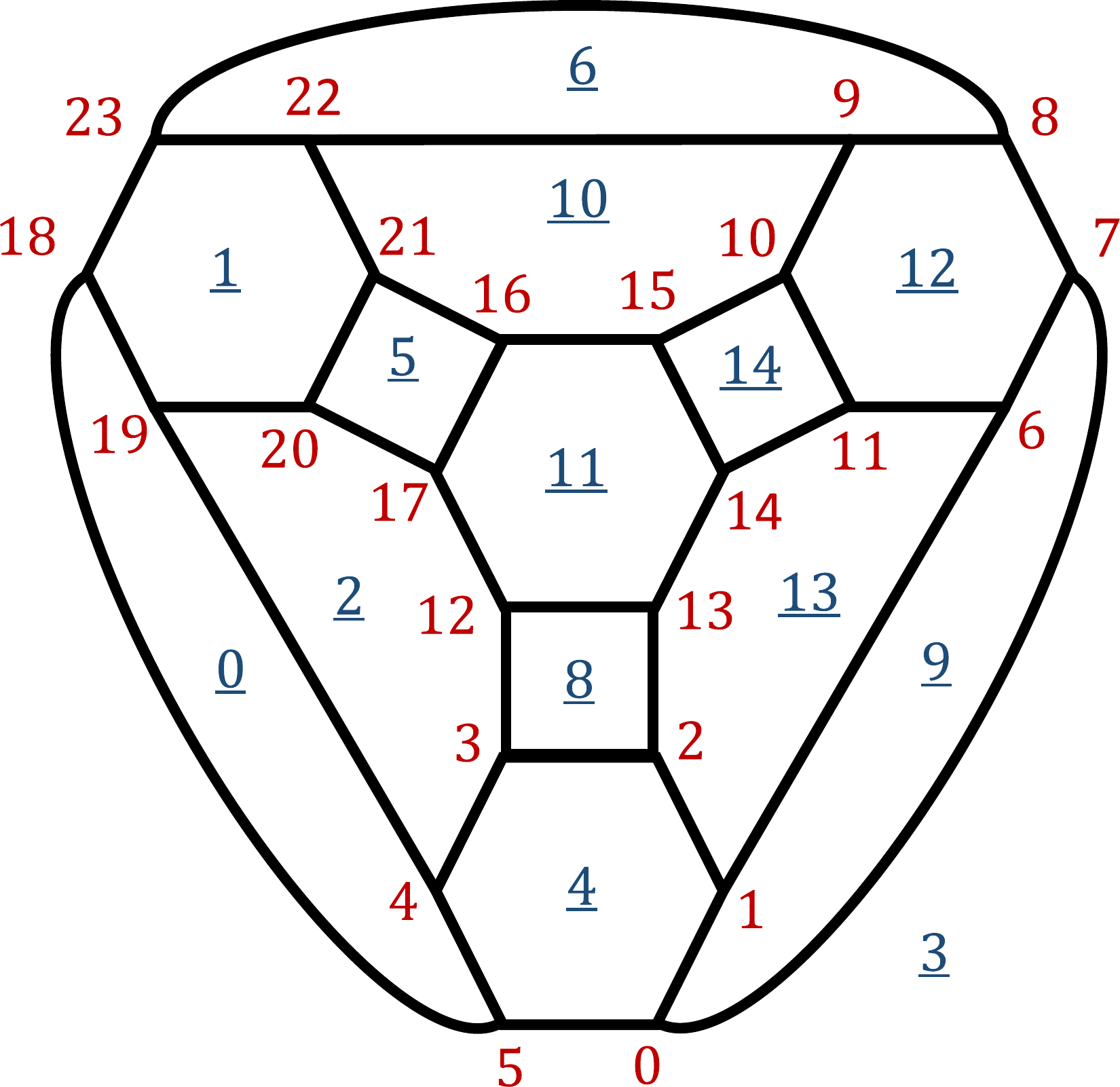} }} 
\caption{Ordering of the vertices (red) and dual vertices (underlined blue) on the boundary of a permutahedron flattened onto the plane.}
\label{Hordering}
\end{figure}
 The Hamiltonian is given by 
 \begin{align}\label{kham}
\bra{ \vect{j}} \openone{-}h_v \ket{\vect{i}} =& 
\frac{\omega^{\vect{i}_Z \cdot \vect{j}_X- \vect{i}_X \cdot \vect{j}_Z} }{N^{12}}
\delta_{i_8-i_9+i_{22}+i_{23}} \delta_{i_{10}+i_{11}+i_{14}+i_{15}} 
\delta_{i_{16}+i_{17}-i_{20}+i_{21}}  \delta_{i_0-i_1+i_2+i_3+i_4-i_5} \delta_{-j_0+j_1+j_6+j_7}
\nonumber \\
&\ \delta_{j_2+j_3+j_{12}-j_{13}} 
\delta_{j_4+j_5+j_{18} + j_{19}} \delta_{j_9-j_{10}+j_{15}+j_{16}+j_{21}-j_{22}}
\delta_{i_1-i_2+i_6-i_{11}+i_{13}-i_{14}-j_1-j_2+j_{11}+j_{14}} 
\nonumber \\
&\ \delta_{i_3+i_4+i_{12}+i_{17}-i_{19}-i_{20}+j_3+j_4-j_{17}+j_{20}} 
\delta_{i_6+i_7-i_9+i_{10}+j_6+j_7-j_8+j_9-j_{10}+j_{11}} 
\nonumber \\
&\ \delta_{-i_{12}+i_{13}+i_{15}+i_{16}+j_{12}-j_{13}+j_{14}+j_{15}+j_{16}+j_{17}} 
\delta_{-i_{18}+i_{19}+i_{21}+i_{22}+j_{18}+j_{19}-j_{20}+j_{21}-j_{22}-j_{23}}
 \end{align}
where the $i_n$ and $j_n$ labels are are in the either the $X$ or $Z$ basis as shown in Table.\ref{xzbasis}
\begin{table}[h]
\begin{center}
\begin{tabular}{c | c  c  c  c  c  c  c  c  c  c  c  c  c  }
n & 0 & 1 & 2 & 3 & 4 & 5 & 6 & 7 & 8 & 9 & 10 & 11 & 12 
\\ \hline
$i_{n}$ & X & Z & Z & X & Z & Z & X & Z & X & Z & X & X & Z
\\ 
$j_n$ & Z & X & X & Z & X & X & Z & X & Z & X & Z & Z & X 
\end{tabular}
\\	
~\hspace{.75cm}
\begin{tabular}{ c c  c  c  c  c  c   c  c  c c }
13 & 14 & 15 & 16 & 17 & 18 & 19 & 20 & 21 & 22 & 23 
\\ \hline
 X & X & X & Z  & X & Z & Z & Z & X & X & Z
\\ 
 Z & Z & Z & X & Z & X & X & X & Z & Z & X
\end{tabular}	
\end{center}
\caption{Basis choices for $\vect{i},\vect{j}$.}
\label{xzbasis}
\end{table}
and by $i$ in the X basis we mean $\ket{\hat i }:=N^{-\frac{1}{2}}\sum\limits_{k=0}^{N-1}\omega^{{-}i\cdot k }\ket{k}$. Also the notation $\vect{i}_Z$ indicates the subset of $i$ labels in the $Z$ basis and similarly for $j$ and $X$.

A matrix element of the Hamiltonian for fixed $\vect{i}$ and $\vect{j}$ as above gives rise to a tensor network multiplied by some nonzero weight. The tensor network is composed of delta tensors and $X$ matrices and computes a delta condition on the flatness of the configuration shown in Fig.~\ref{pflatness}.
\begin{figure}[ht]
\center
{{
  \includegraphics[height=0.4\linewidth]{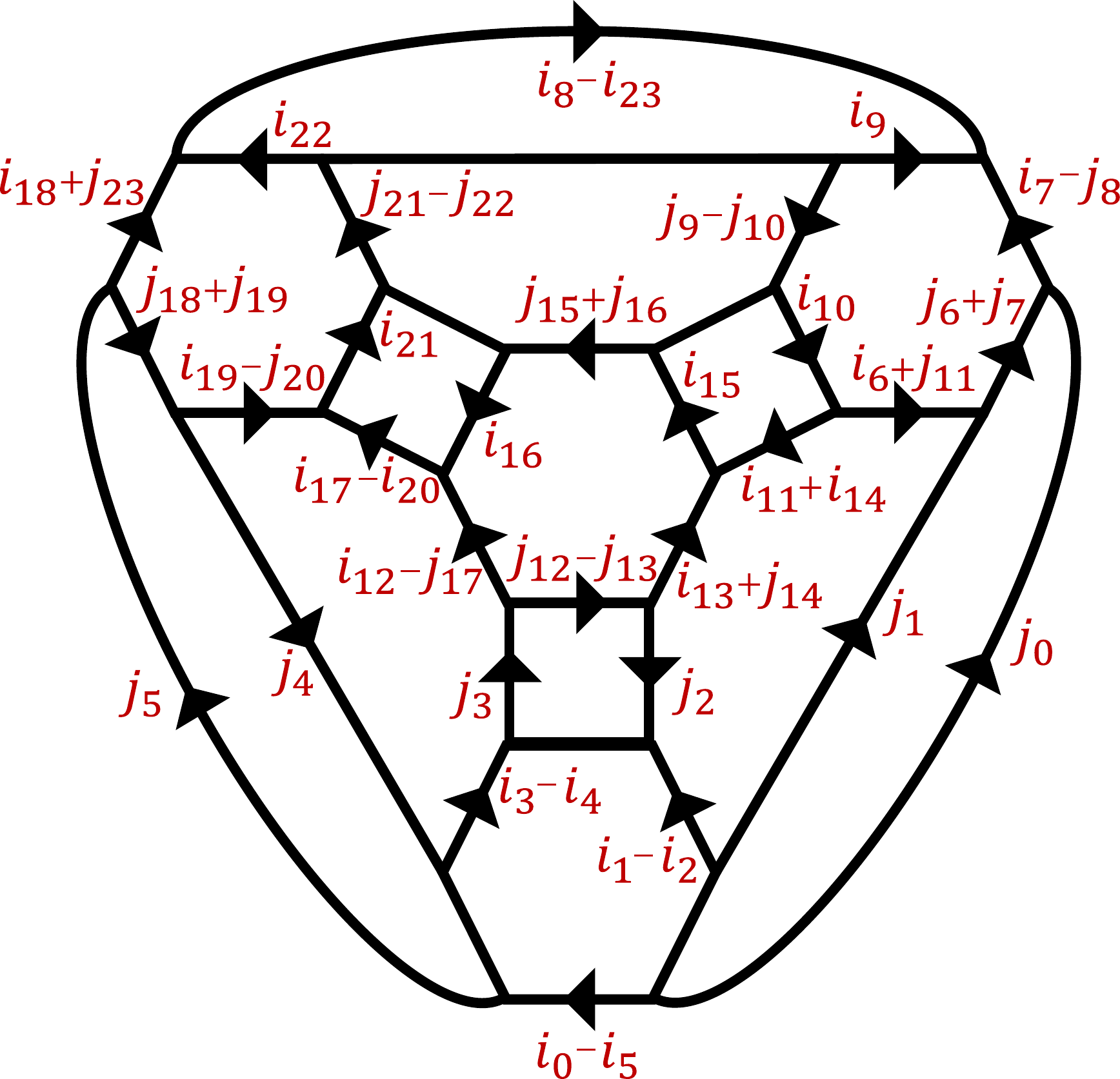}}} 
\caption{Configuration induced by fixing the input/output of a Hamiltonian term.}
\label{pflatness}
\end{figure}
Evaluating the delta flatness condition yields the Hamiltonian term in Eq.\eqref{kham}.

\subsection{Degeneracy, Statistics, and the Ground State Wave Function}

Thus far there is little known about the relation of the Kashaev TQFTs to other more established families of models. We conjecture that the Kashaev TQFTs and their Hamiltonian realizations are equivalent to Crane-Yetter-Walker-Wang  (CYWW) models for $\Z_N$ with a particular choice of braiding. In the case $N$ is odd it is a nondegenerate braiding on $\Z_N$ hence the CYWW model is an invertible TQFT and its partition functions depend only on Euler characteristic and signature. In the case of $N$ even the relevant braiding on $\Z_N$ is degenerate and hence the CYWW model is premodular. 

The partition functions calculated for Kashaev's TQFT support our conjecture as they are consistent with an invertible modular CYWW model for $N$ odd and consistent with a premodular CYWW in the even case, the simplest example being $\Z_2$ which we expect to be the CYWW model based on $sVec$.
More generally we conjecture the even case is equivalent to a CYWW model based on a premodular category with transparent subcategory: $\Z_2$ (with trivial braiding) for $N=0 \mod 4$, and $sVec$ for $N=2\mod 4$. 

 The partition function $\tft{S^1\times \nan}$ equals the dimension of the ground space $\vft{\nan}$ (note in this case normalization by Euler characteristic and signature are irrelevant as both are 0). The values of $\tft{S^1\times S^3}$ in the table indicate that Kashaev's TQFT is stable i.e. does not spontaneously break any symmetry. 

In accordance with our conjecture we expect the topological excitations of the Kashaev model to match those of CYWW. In particular for $N$ odd there are no deconfined particle like excitations in the bulk while there may be interesting loop like excitations. For $N$ even there is a species of point like fermionic excitations in the bulk as well as loop like excitations. Explicitly comparing the 3 loop braiding statistics of the loop excitations in Kashaev and CYWW is an interesting problem which we leave for future work. 

There is a PEPS representation of a ground state wave function for all Kashaev TQFTs which is obtained by following the procedure of Ref.\cite{higherdto}. Due to the Pachner move symmetry of the tensors used to construct this PEPS it will have a closed surface PEPO topological symmetry on the virtual level~\cite{higherdto,schuch2010peps,buerschaper2014twisted,burak2014characterizing,williamson2014matrix}. It should be possible to construct the excitations from this PEPO by following a higher dimensional generalization of the procedure laid out in Ref.\cite{bultinck2015anyons} for $(2+1)$D. Note the procedure of Ref.\cite{higherdto} also yields a MERA representation of the ground state wave functions.

\subsection{Example: $\Z_2$ Case}

The explicit tensor for the $N=2$ Kashaev TQFT is given by 
\begin{align}\label{Qtens2}
Q=\frac{1}{\sqrt{2}} \sum_{k,l,m\in\Z_2} (-1)^{km} \ket{k}\bra{k+m}\otimes\ket{l}\bra{l+m}\otimes\ket{m}
\end{align}
\begin{figure}[ht]
\center
{{
 \includegraphics[height=0.23\linewidth]{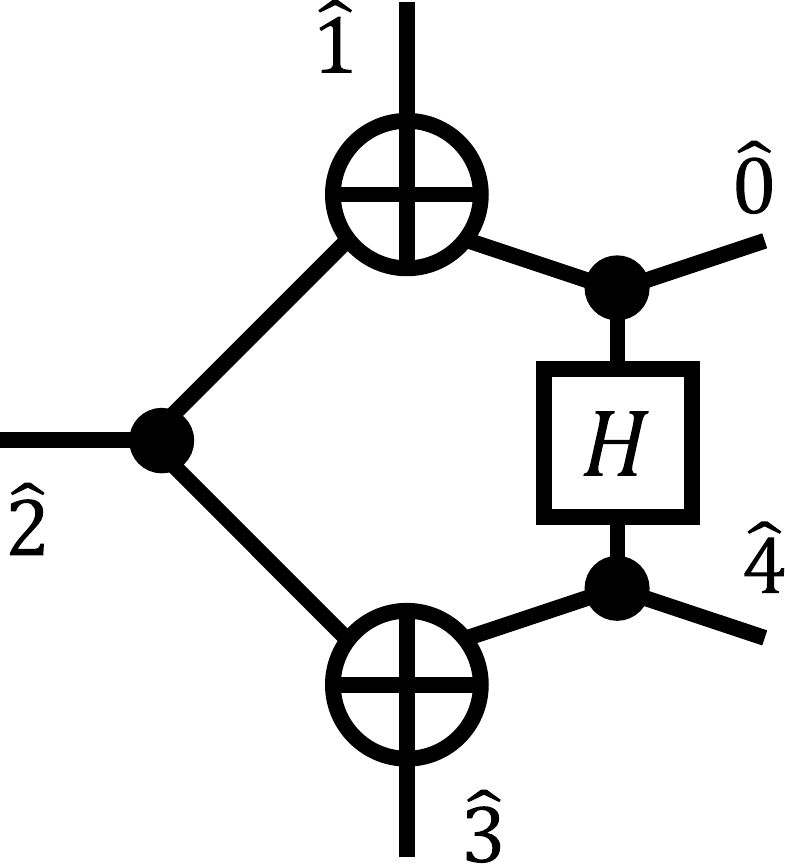}}} 
\caption{Kashaev's Q tensor for $N=2$.}
\label{Qpic2}
\end{figure}
We conjecture this $N=2$ Kashaev model is described by the CYWW model for the premodular category $sVec$ consisting of a vacuum and fermion particle. Hence we expect the partition functions to depend on the possible spin structures of a manifold. This is supported by the observations that $\tft{\CP^2}=0$ which corresponds to $\CP^2$ not admitting a spin structure, and we proceed to show $\tft{T^4}=8$ corresponding to 8 spin structures on the space manifold $T^3$.
\begin{lemma}\label{lem1}
For the commuting, projector, real-space blocking RG fixed point Hamiltonian term $h_v$ we have {\normalfont $\dim \vft{T^3}=\tft{T^4}= \tr{\pi \tft{v'*\st{v}}\,}  = \tr{\pi (\openone - h_v)}$} for the permutation $\pi$ given below.
\end{lemma}
\begin{proof}
We consider the translation invariant BCC triangulation of $T^3$ (or $\R^3$) dual to a tiling by regular 4-permutohedra. For concreteness we fix the branching structure to be that of Fig.\ref{branch} and the ordering of tetrahedra in $\st{v}$ to be that of Fig.\ref{Hordering}. However note any branching structure that is consistent with periodic boundary conditions may be used, and the ordering is totally arbitrary. Considering the Hamiltonian term of Eq.\ref{kham} we have $\openone-h_v=\tft{v'*\st{v}}$ and note this is a tensor network on the triangulation of a 4d hypercube. Conceptually it is clear that taking closed boundary conditions yields the partition function of the 4d torus. The gluing map that corresponds to closing the boundary conditions is specified by the permutation
\begin{align}
\begin{tabular}{ l  c l l l l l l l l }
$\pi$: & 0 & $\rightarrow$ & 13 & $\rightarrow$ & 20 & $\rightarrow$ & 9 & $\rightarrow$ & 0
\\
& 1 & $\rightarrow$ & 18 &  $\rightarrow$ & 17 & $\rightarrow$ & 10 & $\rightarrow$ & 1
\\
& 2 & $\rightarrow$ & 19 & $\rightarrow$ & 8 & $\rightarrow$ & 15 & $\rightarrow$ & 2
\\ 
& 3 & $\rightarrow$ & 6 & $\rightarrow$ & 23 & $\rightarrow$ & 16 & $\rightarrow$ & 3
\\ 
& 4 & $\rightarrow$ & 7 & $\rightarrow$ & 14 & $\rightarrow$ & 21 & $\rightarrow$ & 4
\\
& 5 & $\rightarrow$ & 12 & $\rightarrow$ & 11 & $\rightarrow$ & 22 & $\rightarrow$ & 5
\end{tabular}
\end{align}
abusing notation slightly we also use $\pi$ to denote the linear operator $\sum_{\{ i_n \}} \ket{\{i_{\pi(n)}\}}\bra{\{i_n\}}$. Then we have $\tft{T^4}=\tr{\pi \tft{v'*\st{v}}\,}  = \tr{\pi (\openone - h_v)}$. 
\end{proof}
We furthermore conjecture that a similar relation holds in all dimensions, following from the basic facts that the regular $(n+1)$-permutohedron tiles $n$ dimensional euclidean space (or the $n$ dimensional torus) and that the join of its dual triangulation of the $n$-sphere with a line (including its two endpoints) is a triangulation of the $(n+1)$ hypercube. By taking appropriate periodic boundary conditions, specified by a generalization of the permutation $\pi$ we find a similar relation as in $(3+1)$D.
\begin{proposition}
For Kashaev's model at $N=2$~\emph{\cite{kashaev2014simple} }
 $\dim \vft{T^3}=\tft{T^4}=8$
\end{proposition}
\begin{proof}
We make use of Lemma~\ref{lem1} and calculate $ \tr{\pi \tft{v'*\st{v}}\,} = \tft{T^4}$ using the tetrahedron labeling in Fig.\ref{Hordering} and the branching structure in Fig.\ref{T4branch}.
\begin{figure}[ht]
\center
{{
 \includegraphics[height=0.35\linewidth]{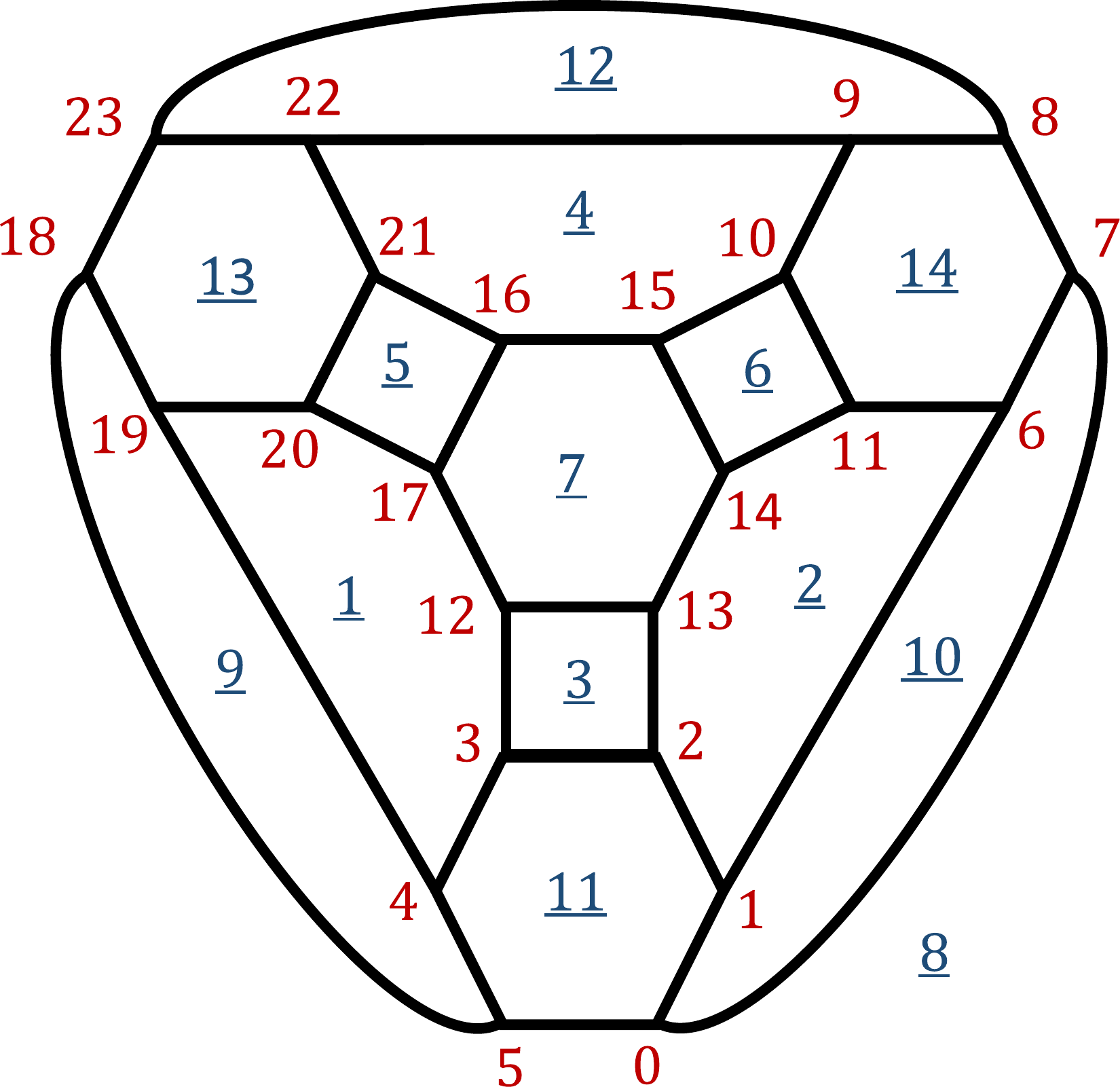}}} 
\caption{The branching structure used to calculate $Z[T^4]$.}
\label{T4branch}
\end{figure}
The quantity we want to compute is a sum over tensor networks on the 4-permutohedron graph
\begin{align}
\tr{\pi \tft{v'*\st{v}}\,} =\sum_{\{i_n\}}\bra{\{i_n\}} \tft{v'*\st{v}} \ket{\{i_{\pi(n)}\}} . 
\end{align}
Each tensor has 5 indices, 2 exterior indices corresponding to $i_n,i_{\pi(n)}$ and 3 internal indices of the tensor network on the permutohedron. With the branching structure we have chosen the fixed exterior labels correspond to the 0 and 4 indices of each individual tensor, fixing these labels we find the following 3 index tensor
\begin{align}\label{reduced}
\vcenter{\hbox{
\includegraphics[width=0.18\linewidth]{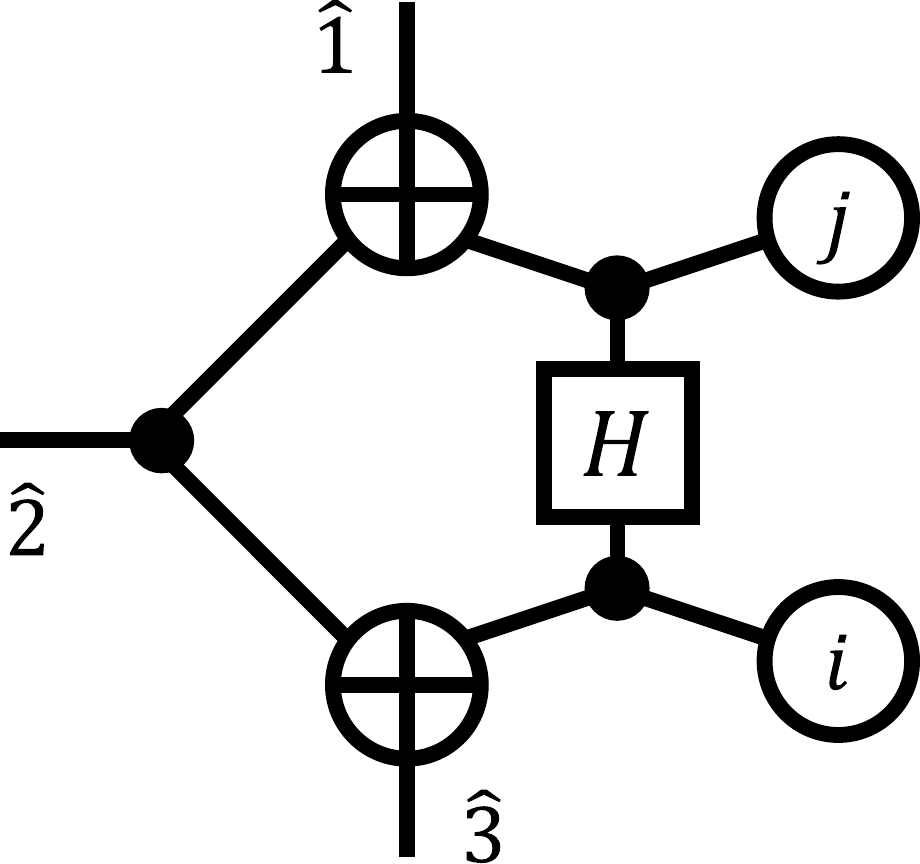}}}
\ = \frac{({-}1)^{ij}}{\sqrt{2}}\ \vcenter{\hbox{
\includegraphics[width=0.065\linewidth]{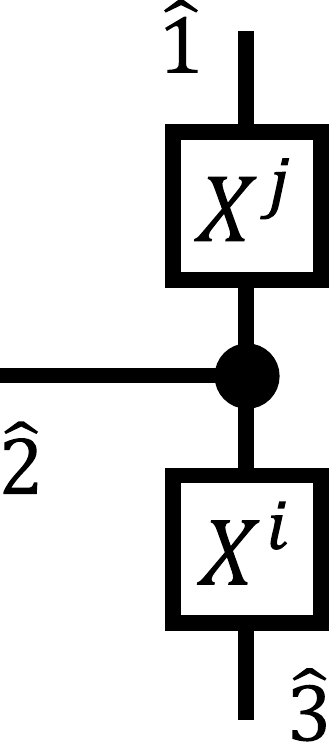}}}
\end{align}
which consists of a delta tensor along with $X$ matrices and a sign that is determined by the fixed external indices. 
From this equation it follows that the tensor network on the 4-permutohedron graph specifies a map $\sigma:\Z_2^{24}\rightarrow \Z_2^{36}$ from the fixed external indices to a $\Z_2$ configuration on the edges of the permutohedron, corresponding to the exponent of the $X$ matrix on each edge. 
Contracting the tensor network yields a nonzero result only for those connections which are flat, i.e. the sum around each plaquette is 0. 
To describe this precisely we denote the map from a $\Z_2$ edge configuration to the induced flux through each plaquette by $f$ then we have 
\begin{align}\label{eq9}
\bra{\{i_n\}} \tft{v'*\st{v}} \ket{\{i_{\pi(n)}\}} = \frac{(-1)^{\vect{i}\cdot\pi(\vect{i})}}{2^{12}} \delta(\sigma\circ f(\vect{i})=0)
\end{align}
where $\vect{i}\cdot\pi(\vect{i})=\sum_n i_n\cdot i_{\pi(n)}$ and the normalization factor comes from a product of $\frac{1}{\sqrt{2}}$ for each of the 24 tensors (See Eq.\ref{reduced}), a factor $\frac{1}{2}$ from the normalization $N^{-|\triang_0|}$, and a factor 2 from the contraction of delta tensors and $X$ matrices on the permutohedron graph. 
We proceed to show that all flat configurations contribute with a $+1$ sign and hence the problem is to count the number of them. First we use the relation $\sum_k\delta_{i,j,k}\delta_{k,l,m}=\delta_{i,j,l,m}$ to remove 12 edges of the permutohedron tensor network by contracting them, precisely those on which an $X$ never occurs, see Fig.\ref{reducedtn}. 
\begin{figure}[ht]
\center
{{
 \includegraphics[height=0.4\linewidth]{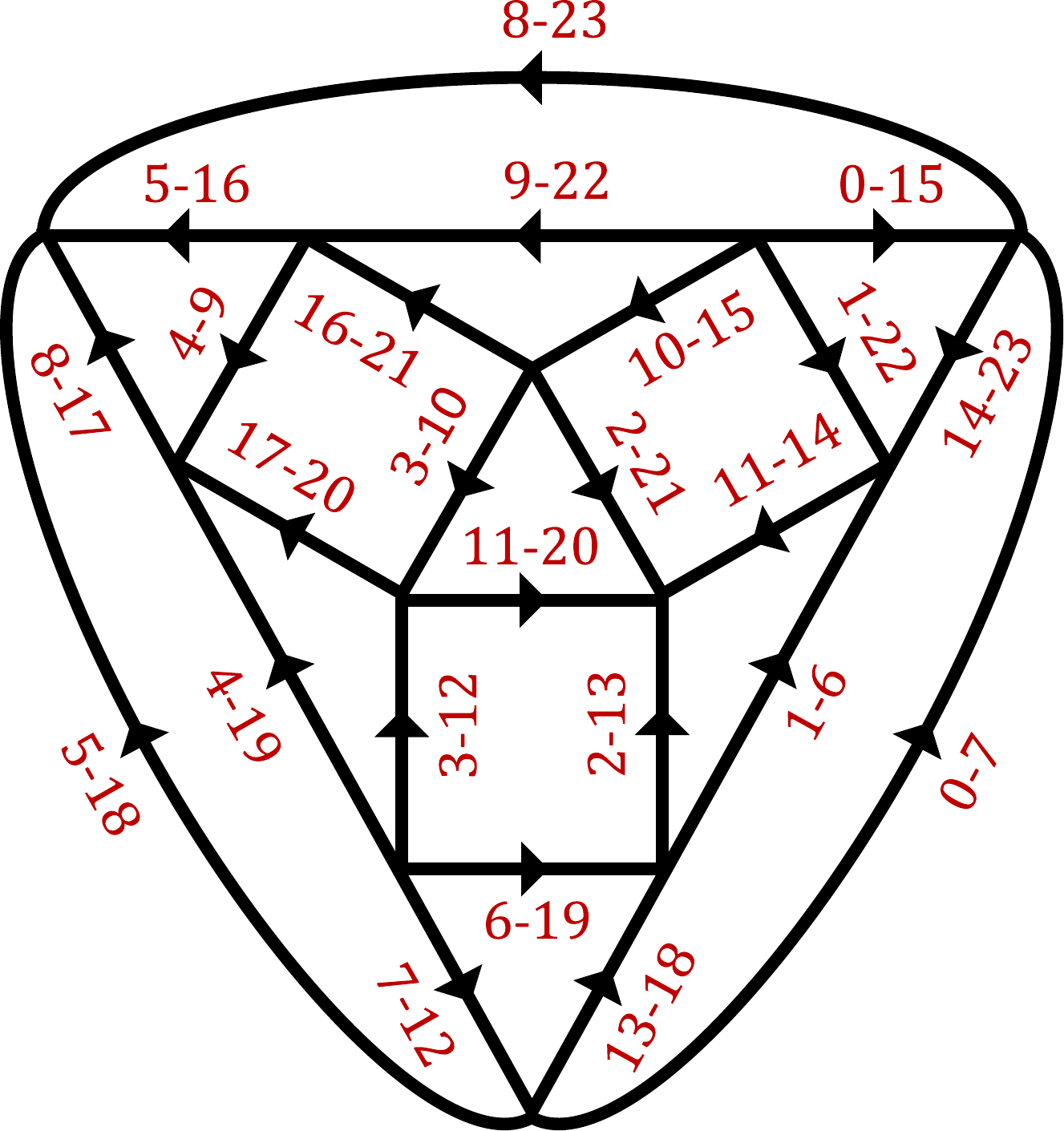}}} 
\caption{Configuration induced by the map $\sigma$. Value $n{-}m$ to be read as $i_n{-}i_m$.}
\label{reducedtn}
\end{figure}
Hence we can consider reduced $\Z_2$ maps $\sigma:\Z_2^{24}\rightarrow \Z_2^{24}$ and $f:\Z_2^{24}\rightarrow \Z_2^{13}$ by noting there are 24 edges and 14 faces (one of which is redundant as its value equals the sum of all other faces). 
Since we are considering flat configurations on a cellulation of $S^2$ there is only a single homology class and furthermore an injective map $d:\Z_2^{11}\rightarrow \Z_2^{24}$ from $\Z_2$ values on the 12 vertices (with one vertex value fixed) to a $\Z_2$ configuration on edges such that $\im\, d = \ker f$. A full list of generators for $\ker\, \sigma \circ f$ is
\begin{align}\label{generators}
\begin{tabular}{ l | r r r r r r r r r r r r r r  r }
$i$ & 0 & 1 & 2 & 3& 4& 5& 6& 7& 8& 9& 10& 11& 12& 13& 14
\\ \hline
 & 7 & 0 & 5 & 12 & 6 & 1 & 15 & 9 & 4 & 3 & 2 & 0 & 1 & 2 & 8 
\\
 & 18 & 23  & 8  & 19 & 13 & 14 & 22 & 16 & 17 & 20 & 11 &7 & 6 & 13 & 17
\\
$g^i_{(\cdot)}=1$ &  &  &  & & & & & & & & & 12 & 19 & 18 & 20
\\
&  &  &  & & & & & & & & & 3 & 4 & 5 & 11
\\
&  &  &  & & & & & & & & & 10 & 9 & 16 & 14
\\
&  &  &  & & & & & & & & & 15 & 22 & 21 &  23
\end{tabular}
\end{align}
one can verify that these are independent. Generators 11-14 correspond to the generators of $\ker\, \sigma$, while 0-10 lie in $(\ker \sigma)^\perp$ and their image under $\sigma$ generates $ \ker f$. Hence $2^{15}$ flat configurations contribute to the sum, and furthermore we have $\sum_n g^i_n\cdot g^j_{\pi(n)}=\sum_n g^j_n\cdot g^i_{\pi(n)}$. 
Any flat configuration $\vect{x}\in \ker \sigma\circ f$ is of the form $\vect{x}=\sum_i x_i  \vect{g}^i$ and the corresponding phase factor is 1 as $\vect{x} \cdot \pi(\vect{x})=0$ hence all such configurations contribute with a positive sign.
\end{proof}

\subsection{Back to the general case}

Our analysis of $Z[T^4]$ for $N=2$ largely carries over to the case of general $N$, the main modifications required involve keeping track of orientations and complex conjugations. 
The calculation proceeds as above up to Eq.\eqref{reduced} at which point we find
\begin{align}\label{reduced2}
\vcenter{\hbox{
\includegraphics[width=0.18\linewidth]{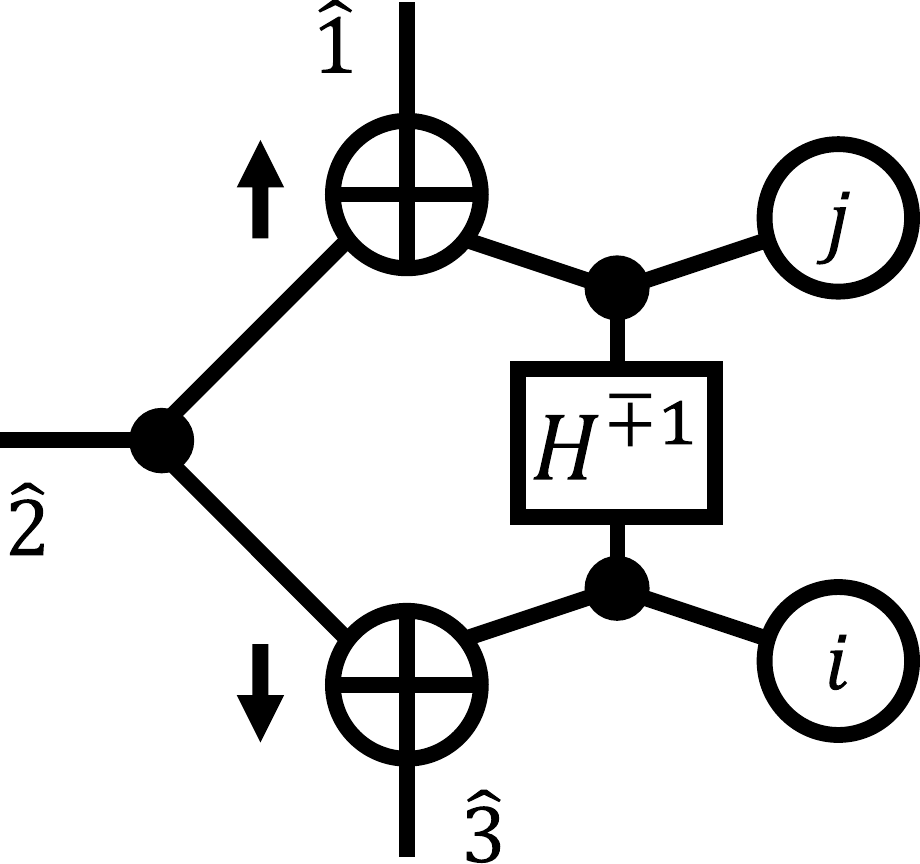}}}
\ =  \frac{\omega^{\pm ij}}{\sqrt{N}} \ 
\vcenter{\hbox{\includegraphics[width=0.066\linewidth]{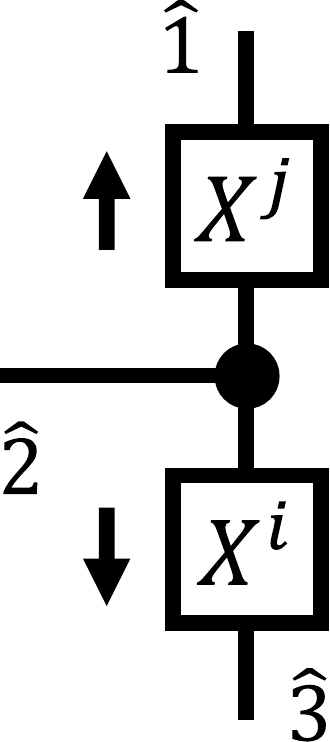}}}
\end{align}
where $\pm$ depends on the orientation of the simplex to which the tensor is associated. 
This leads to new maps $\sigma:\Z_N^{24}\rightarrow \Z_N^{24}$ and $f:\Z_N^{24}\rightarrow \Z_N^{13}$ which can be understood in terms of flat $\Z_N$ connections on a cellulation of $S^2$, see Fig.\ref{reducedtn}. 
In place of Eq.\eqref{eq9} we have
\begin{align}
\bra{\{i_n\}} \tft{v'*\st{v}} \ket{\{i_{\pi(n)}\}} = \frac{\omega^{\vect{i}\cdot\pi(\vect{i})}}{N^{12}} \delta(\sigma\circ f(\vect{i})=0)
\end{align}
where the dot product has been altered as follows $\vect{i}\cdot\pi(\vect{i})=\sum_n ({-}1)^n i_n\cdot i_{\pi(n)}$ and the normalization arises in the same way as above (with $N$ in place of 2).
Again we have a boundary map $d:\Z_N^{11}\rightarrow Z_N^{24}$ that satisfies $\im d = \ker f$ and the generators of $\sigma\circ f$ are the same as those given in Eq.\eqref{generators}, however now each has order $N$. 
Now $N^{15}$ flat configurations contribute to the sum, the new complication being that they may take on different phase values. We still have the identities $\vect{g}^i\cdot\pi(\vect{g}^i)=0$ and $\vect{g}^i\cdot\pi(\vect{g}^j)=\vect{g}^j\cdot\pi(\vect{g}^i)$ however they no longer guarantee the phase factors are trivial for $N>2$. 
An arbitrary element $\vect{x}\in \ker \sigma\circ f$ is of the form $\vect{x}=\sum\limits_i x_i  \vect{g}^i$ and we have 
\begin{align*}
\vect{x} \cdot \pi(\vect{x})&=\sum\limits_{i<j} 2x_i x_j \vect{g}^i \cdot \pi(\vect{g}^j)
 \\
&= 2 (x_0(x_8-x_5)+(x_1-x_9)(x_4-x_7)+(x_2-x_{10})(x_6-x_3))
.
\end{align*}
Hence the overall summation becomes
\begin{align}
Z[T^4]
&=\left( \sum_{i,j\in\Z_N} \frac{\omega^{2ij}}{N} 
\right)^3
\\
&=\left( \sum_{i\in\Z_N} \delta(2i=0 \text{ mod } N)
\right)^3
\end{align}
which takes the value $1$ for $N$ odd and $8$ for $N$ even.

\section{Lattice Model for UG{\lowercase{x}}BFC TQFT}
\label{gcrossed}

In this section we introduce tensor network and graphical calculus constructions of a class of commuting projector Hamiltonians for topological models based on the UGxBFC state sum TQFT.

Before discussing constructions of a Hamiltonian let us describe the Hilbert space on a general triangulation or equivalently the Poincar\'e dual simple polyhedra. 
Given a UGxBFC $\mathcal{C}_G^{\times}$ and $L$ the label set $L^{(1)}=G,\,L^{(2)}=\cat_G^\times,\, L^{(3)}=\cat_G^\times$.  
Suppose $Y$ is an oriented spatial $3$-manifold with a vertex ordered $\Delta$-complex triangulation $\complex$. $\Gamma_\complex$ denotes the dual simple polyhedron.  
We take a resolution of each 4-valent vertex in $\Gamma_\complex$ into a pair of trivalent vertices where the edges dual to the faces $\hat 0, \hat 2$ and $\hat 1, \hat 3$ meet, denote the resulting polyhedron as $\Gamma_{\complex}'$.
 Let $V,E,F$ denote the sets of vertices, edges, and faces of $\Gamma_{\complex}'$ respectively.  A configuration on $\Gamma_{\complex}'$ is a labeling of each edge by a defect label $a\in\cat_G^\times$, each face by a group element $g\in G$, and each vertex by a basis element in $\oplus_{(a,b,c)\in (\cat_G^\times)^3}{\mathbb{C}^{N_{ab}^c}}$.  Hence the total local Hilbert space is
\begin{align}\label{hilbsp}
 \mathcal{H}(Y,\complex)=\bigotimes_E{\C[\cat_G^\times]}\bigotimes_V( \bigoplus_{ \substack{(a,b,c) \\ \in  (\cat_G^\times)^3}}{\C^{N_{ab}^c}})\bigotimes_F \C[G].
\end{align}

\subsection{Tensor Network Approach}
\label{gcrossedtnham}

The recipe outlined in Ref.\cite{higherdto}  constructs a local commuting projector Hamiltonian for the UGxBFC state sum on any triangulation from a set of $15j$-symbols. 
The $15j$-symbols for the UGxBFC $\ftj^\pm_{s(\simp_4)}$ are shown in Fig.\eqref{15jsymbols} they take as input a configuration $s=(g,a):(\complex^{(1)},\complex^{(2)}\cup\complex^{(3)})\rightarrow (G,\cat_G^\times)$ of group elements $g_e\in G$ on edges, and defects $a_\simp,a_{\simp_3} \in \cat_G^\times$ on triangles and tetrahedra of a 4-simplex $\simp_4$ and return a value in $\C$. 
$\ftj^\pm_{s(\simp_4)}$ only take nonzero values on admissible configurations, those satisfying the triangle constraints $a_\simp \in\cat_{dg_\simp}$, where $dg_{012}=\bar{g}_{02}g_{01}g_{12}$, and the tetrahedra constraints $N_{a_{\hat{1}}\, {}^{\bar{g}_{23}}a_{\hat{3}}}^{a_{0123}}\neq 0 \neq N_{a_{\hat{2}}a_{\hat{0}}}^{a_{0123}}$.

The Hamiltonian takes the form $H=\sum\limits_v h_v$ where the term at vertex $v$ is given by
\begin{align}\label{k2ham}
h_v=\openone- \frac{(D^2)^{-|(\st{v})_1|}}{|G|} \sum_{\gamma,\alpha} \frac{\prod\limits_{\simp_2\in\subc}d_{\alpha(\simp_2)}}{\prod\limits_{\simp_3\in\subc}d_{\alpha(\simp_3)}} B_v^{\gamma, \alpha}
\end{align}
where $\subc = \interior(v'*\st{v})$ can be thought of as a small piece of spacetime,  with $v'$ an auxiliary copy of vertex $v$ at the next time step. The elements $\gamma\in G,\, \alpha: \subc^{(2)}\cup\subc^{(3)} \rightarrow \cat_G^\times$ label the timelike edge, triangles and tetrahedra in $\subc$.

The individual summands are given by
\begin{align}
&B_v^{\gamma,\alpha}=\sum_{\config,\config'} \prod_{\substack{\simp_i \in \lk{v}  \\ i>0}} \delta_{\config(\simp_i),\config'(\simp_i)}
\frac{\prod\limits_{\simp_2\in\st{v}}\sqrt{d_{a(\simp_2)}d_{a'(\simp_2)}}}{\prod\limits_{\simp_3\in\st{v}}\sqrt{d_{a(\simp_3)}d_{a'(\simp_3)}}}
 \prod_{\simp_4\in \subc}\ftj^{\orient{\simp_4}}_{s(\simp_4)}
 \bigotimes_{e,\simp,\simp_3\in \st{v}} \ket{g_e',a_\simp',a_{\simp_3}'} \bra {g_e,a_\simp,a_{\simp_3}}
\end{align}
where $\config=(g,a),\, \config'=(g',a')$ denote configurations on the triangulated spatial slices ${ \cl (\st{v}),}\, {v'*\lk{v}}$ and $s$ denotes the full spacetime configuration $\{S,S',(\gamma,\alpha)\}$ on $\cl(\subc)$. 
Note the variables in $\lk{v}$ are fixed control qudits for the operator $B_v^{\gamma,\alpha}$, while the variables in $\st{v}$ fluctuate.

The Hamiltonian built in this way contains only vertex terms. These terms also enforce the flatness (admissibility) triangle and tetrahedra constraints on basis states with nonzero ground space overlap. Contrast this with the more conventional way of writing fixed point Hamiltonians as a sum of separate vertex fluctuation and plaquette flatness terms. 

On the BCC triangulation with the branching structure chosen in Fig.\ref{branch} the Hamiltonian is a translationally invariant sum with a single type of term. 
Note the case described explicitly above was assuming no multiplicity in the fusion of $\cat_G^\times$, to include possible multiplicities one simply includes the corresponding fusion multiplicity labels together with the defect label on each tetrahedron.

\subsection{Graphical Calculus Approach}

Another approach to constructing the Hamiltonian closer to that of Ref.\cite{walker2012} is to use the graphical calculus of the UGxBFC to define the local terms. 
First we consider two different cellulations of the $3$-torus $T^3$: the cellulation used in Ref.\cite{walker2012}, and the simple polyhderon of the permutohedron cellulation dual to the BCC triangulation in Fig.\ref{branch}.  Both celluations have the full translational symmetry which keep the Hamiltonians relatively simple.
We then explain the general construction on a simple polyhedra.

\subsubsection{The Hamiltonian on the $3$-Torus: resolved cubic lattice}

We first focus on the simple case where $G$ is abelian and all group elements and defects are self inverse, this removes the need to keep track of edge orientations. Note it is simple to generalize to the non self inverse case by keeping track of edge orientations, however extending to nonabelian $G$ requires nontrivial work as the cellulation is not a simple polyhedra. 

The cellulation $\Gamma$ is given by the following resolution of the cubic lattice into trivalent vertices. 
\begin{align*}
{
\vcenter{\hbox{
\includegraphics[scale=1]{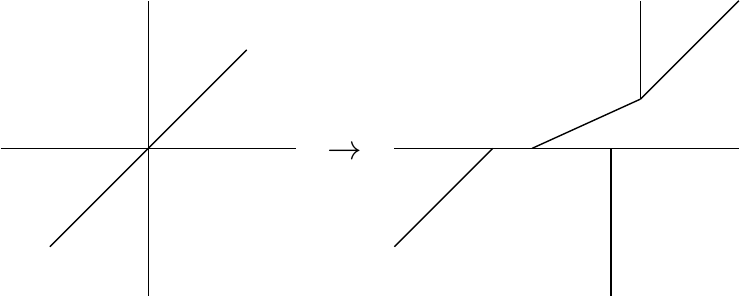}
}}
} 
\end{align*}
Group degrees of freedom live on the plaquettes of the cubic lattice, and defect degrees of freedom live on the edges of the resolved lattice. 
Hence the Hilbert space is given in Eq.\eqref{hilbsp}.

The Hamiltonian is given by 
\begin{align}\label{unorientedham}
H_{\Gamma}=  -\sum_{v\in V} A_v - \sum_{e \in E} A_e  -\sum_{f\in F} \sum_{ \substack{g\in G, \\ a\in \cat_g}} \frac{d_a}{D^2} B_f^{g,a}  
 - \sum_{c\in C} \sum_{g\in G} B_c^g
\end{align}
where $V,E,F,C$ are the vertices, edges, faces and 3-cells of $\Gamma$. 
The $A_v$ term enforces the constraint that each triple of defects $a,b,c$ meeting at a vertex is an admissible fusion $N_{ab}^c\neq 0$. 
\begin{align*}
A_v \left( \vcenter{\hbox{
\includegraphics[scale=1]{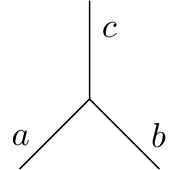}
}} \right) = \delta_{ab}^c
\vcenter{\hbox{
\includegraphics[scale=1]{figs/tikzfig3}
}}
\end{align*}
$A_e$ enforces the constraint that each defect lies in the sector given by the boundary of the group configuration on the adjacent faces i.e. $a_e\in \cat_{(\partial g)_e}$. 
\begin{align*}
A_e\left ( \vcenter{\hbox{
\includegraphics[scale=1]{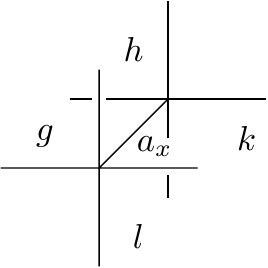}
}} \right)
= \delta_{x,ghkl} \vcenter{\hbox{
\includegraphics[scale=1]{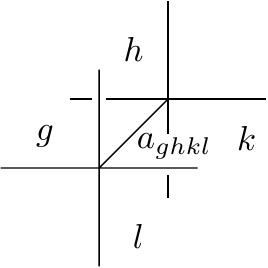}
}}
\end{align*}
$B^g_c$ fluctuates the group configuration adjacent to $c$ in the conventional way $h_f g$, for $f\in\partial c$. 
\begin{align*}
B^g_c \left( \vcenter{\hbox{
\includegraphics[scale=1]{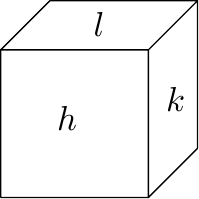}
}} \right)
=
\vcenter{\hbox{
\includegraphics[scale=1]{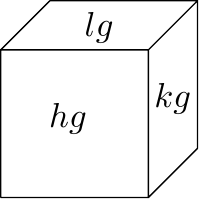}
}}
\end{align*}
The $B_f^{g,a}$ term fluctuates the group and defect configuration adjacent to $f$ by fusing a loop of defect $a$ in to the defects on $\partial f$ and simultaneously multiplying $h_f g$. This term is reminiscent of the plaquette term in the Walker-Wang model and matches it exactly in the case $G=\{1\}$. 

We proceed to calculate the effect of the $B_f^{g,a}$ term using the diagrammatic calculus of the UGxBFC, we use the compressed notation $F^{abc}_{d;ef}=[F^{abc}_d]_e^f$. 
Consider an initial configuration $\sigma_{\text{I}}$ depicted on the left; first the edges crossing $f$ are moved aside
\begin{align*}
{
\vcenter{\hbox{
\includegraphics[scale=1]{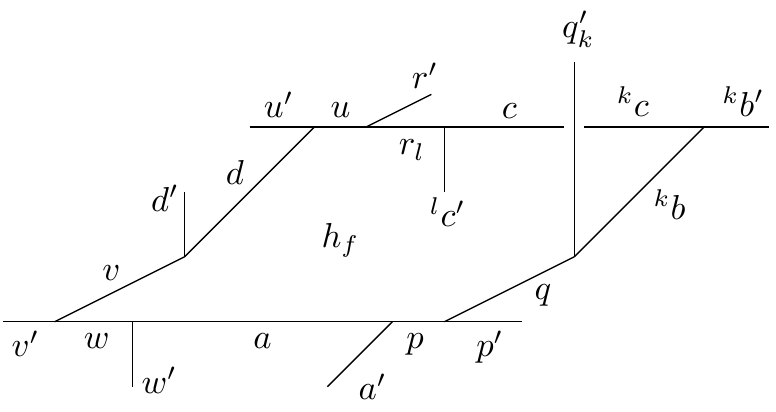}
}}
}
\rightarrow
\vcenter{\hbox{
\includegraphics[scale=1]{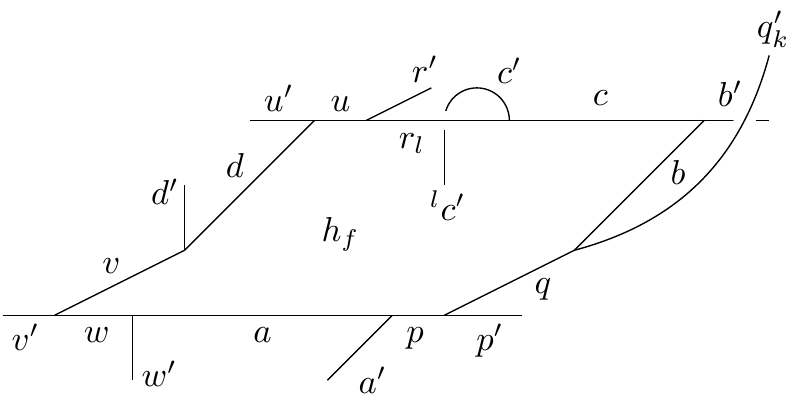}
}}
\end{align*}
 unlike the Walker-Wang model this invokes a factor of $U_{k}(b,c;b')$ on top of the braiding symbols $\overline{R_q^{q' {}^{k}b}} R_c^{{}^{l}c' r}$.
Now acting with $B_f^{g,a}$ introduces a loop of defect $a$ onto plaquette $f$:
\begin{align*}
B_f^{g,a} \left(
\vcenter{\hbox{
\includegraphics[scale=.9]{figs/tikzfig7}
}}
\right)
= 
\vcenter{\hbox{
\includegraphics[scale=.9]{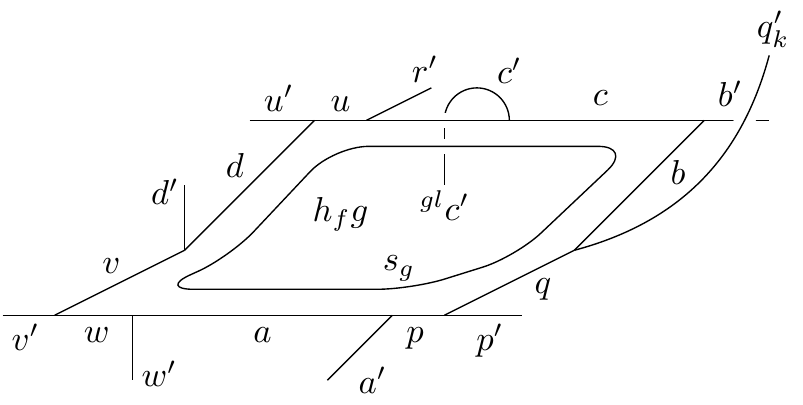}
}}
\end{align*}
which we proceed to fuse in to $\partial f$:
\begin{align*}
\vcenter{\hbox{
\includegraphics[scale=1]{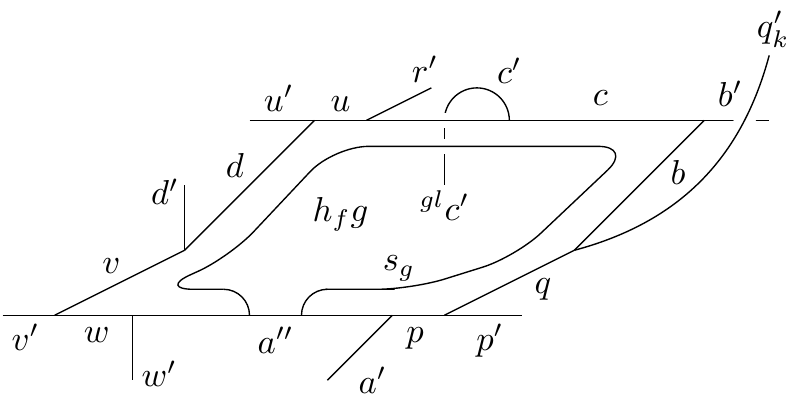}
}}
\rightarrow
\vcenter{\hbox{
\includegraphics[scale=1]{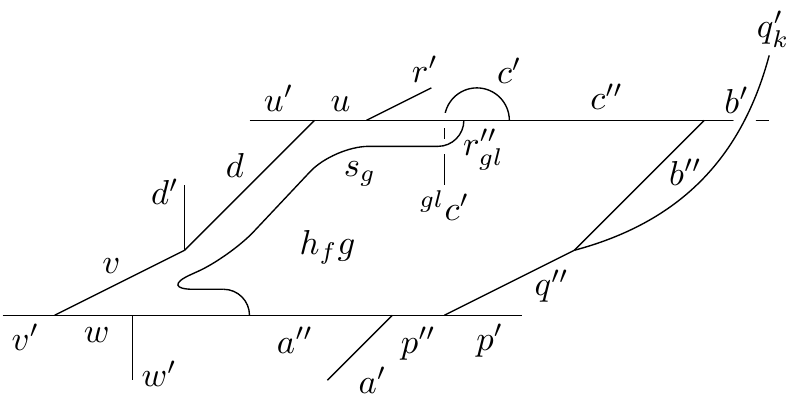}
}}
\end{align*}
this induces a factor $F^{a''sp}_{a';ap''}F^{p''sq}_{p';pq''}F^{q''sb}_{q';qb''}F^{b''sc}_{b';bc''}F^{c''sr}_{c';cr''}$.
The next step induces a factor $\eta_{c'}(l,g)$
\begin{align*}
\vcenter{\hbox{
\includegraphics[scale=1]{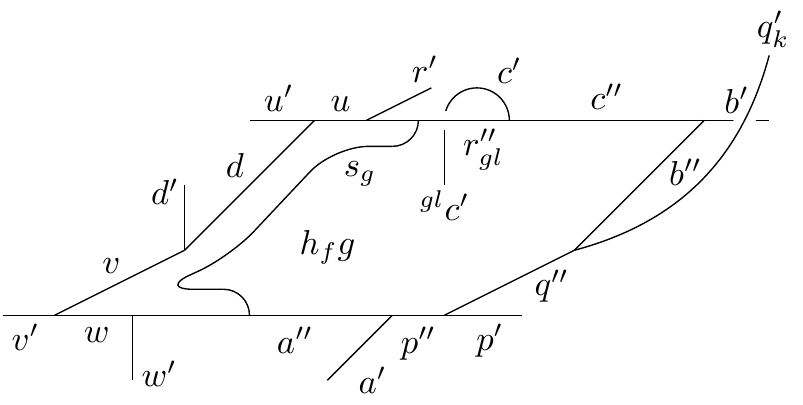}
}}
.
\end{align*}
Then five more $F$-moves leads to 
\begin{align*}
\vcenter{\hbox{
\includegraphics[scale=1]{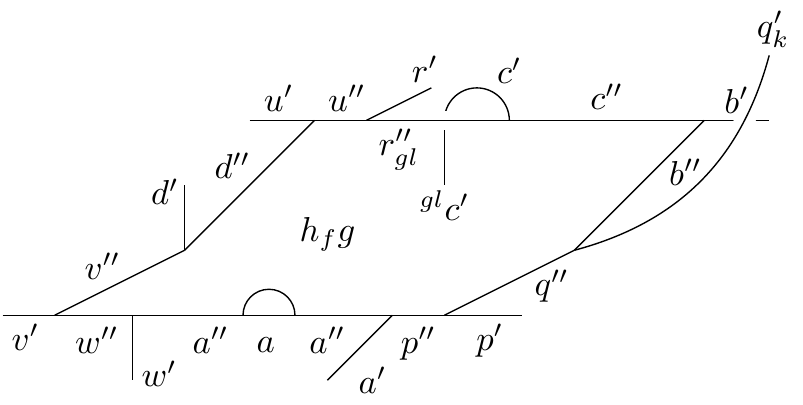}
}}
\end{align*}
along with a factor $F^{r''su}_{r';ru''}F^{u''sd}_{u';ud''}F^{d''sv}_{d';dv''}F^{v''sw}_{v';vw''}F^{w''sa}_{w';wa''}$. 
Finally restoring the lattice to its original position we find the final configuration $\sigma_\text{F}$ shown below
\begin{align*}
\vcenter{\hbox{
\includegraphics[scale=1]{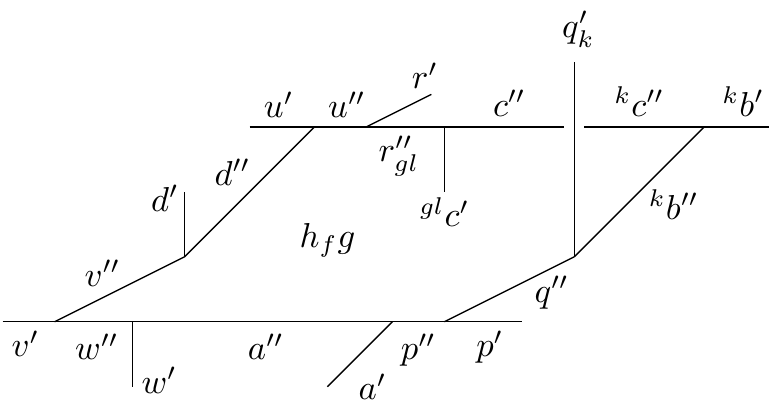}
}}
\end{align*}
along with a factor  $U^{-1}_{k}(b'',c'';b') R_{q''}^{q' {}^{k}b''} \overline{R_{c''}^{{}^{gl}c' r''}}$ note $r''\in\cat_{gl}$.
Hence the full plaquette term is
\begin{align}
\bra{\sigma_\text{F}} B_f^{g,a} \ket{\sigma_\text{I}} =& U_{k}(b,c;b') U^{-1}_{k}(b'',c'';b') \eta_{c'}(l,g) 
 \overline{R_q^{q' {}^{k}b}} R_c^{{}^{l}c' r} R_{q''}^{q' {}^{k}b''} \overline{R_{c''}^{{}^{gl}c' r''}} 
 \nonumber  \\
 &F^{a''sp}_{a';ap''}F^{p''sq}_{p';pq''}F^{q''sb}_{q';qb''}F^{b''sc}_{b';bc''}
F^{c''sr}_{c';cr''}
 F^{r''su}_{r';ru''}F^{u''sd}_{u';ud''}F^{d''sv}_{d';dv''}F^{v''sw}_{v';vw''}F^{w''sa}_{w';wa''} 
 \end{align}
 which differs noticeably from WW in the appearance of the factors $U_{k}(b,c;b') U^{-1}_{k}(b'',c'';b') \eta_{c'}(l,g) $.

\subsubsection{The Hamiltonian on Simple Polyhedra}

We now turn to the general case of an arbitrary finite $G$ and a UGxBFC $\cat_G^\times$, this requires a cellulation $\Gamma_\complex$ dual to a triangulation $\complex$ with branching structure and keeping track of edge orientations.
 $\Gamma_\complex'$ is then the cellulation where each 4-valent vertex $v$ has been resolved into a pair of trivalent vertices $v^+,\, v^-$, as in Eq.\eqref{hilbsp}. The orientations of dual edges in $\Gamma_\complex $ are specified as follows: for vertices dual to positively oriented tetrahedra, the dual $\hat 0,\hat 2$ edges point out and $\hat 1, \hat 3 $ point in, and vice versa for vertices dual to negatively oriented tetrahedra. The extra edges introduced in $\Gamma_\complex'$ point from the $\hat 0 ,\hat 2$ vertex to the $\hat 1, \hat 3$ vertex in a resolved vertex dual to a positively oriented tetrahedra, and vice versa for negative. 

The Hamiltonian is similar to that in Eq.\eqref{unorientedham} 
\begin{align}\label{orientedham}
 H_{\Gamma}=  -\sum_{v\in V} (A_{v^+}+ A_{v^-}) - \sum_{e \in E} A_e 
 -\sum_{f\in F} \sum_{ \substack{g\in G, \\ a\in \cat_g}} \frac{d_a}{D^2} B_f^{g,a}   - \sum_{c\in C} \sum_{g\in G} B_c^g 
\end{align}
where $V,E,F,C$ are the vertices, edges, faces, and 3-cells of $\Gamma_\complex$.

Writing $0123$ for the tetrahedron dual to $v$ then the $A_{v^+}$ term enforces the admissibility of the fusion $N_{a_{\hat 2}a_{\hat 0}}^{a_{0123}} \neq 0$, while the $A_{v^-}$ term enforces a twisted fusion constraint $N_{a_{\hat 1} {}^{\bar{g}_{23}} a_{\hat 3}}^{a_{0123}}\neq 0$. Both $A_{v^+}$ and $A_{v^-}$ terms project onto the subspace spanned by locally admissible configurations such as the following (depicted on a triangulation and its dual cellulation)
\begin{align*}
\vcenter{\hbox{
\includegraphics[scale=1]{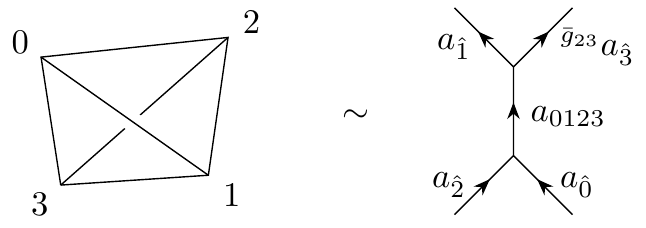}
}}
.
\end{align*}
The operator $A_e$ enforces the constraint that each defect lies in the sector specified by the difference of the adjacent group variables, for an edge $e$ dual to the triangle $012$ the constraint reads $a_{012}\in\cat_{\hat{g}_{02}g_{01}g_{12}}$. Hence the  $A_e$ term projects onto the subspace spanned by admissible configurations such as: 
\begin{align*}
\vcenter{\hbox{
\includegraphics[scale=1]{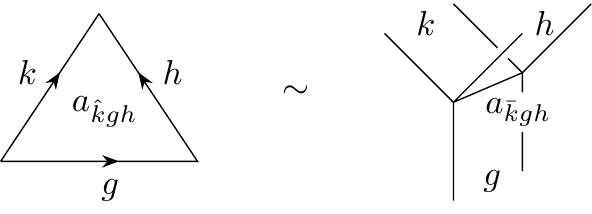}
}}
.
\end{align*}
The term $B_c^g$ fluctuates the group configuration on faces $f\in \partial c$. To be more specific we work on the dual triangulation, then each group variable on an edge $e^-$ pointing towards $v_c$ (the vertex dual to $c$) transforms as $ h_{e^-} \bar g$, while a group variable on an edge $e^+$ leaving $v_c$ transforms as $g h_{e^+} $. Hence the operator $B_c^g$ maps configurations as follows
\begin{align*}
\vcenter{\hbox{
\includegraphics[scale=1]{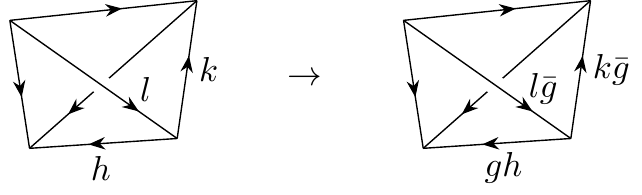}
}}
 .
\end{align*}
Additionally the defects on the five edges associated to the dual of a tetrahedra having $v_c$ as its highest ordered vertex (all edges pointing in) are acted upon by $h$ as follows
\begin{align*}
\vcenter{\hbox{
\includegraphics[scale=1]{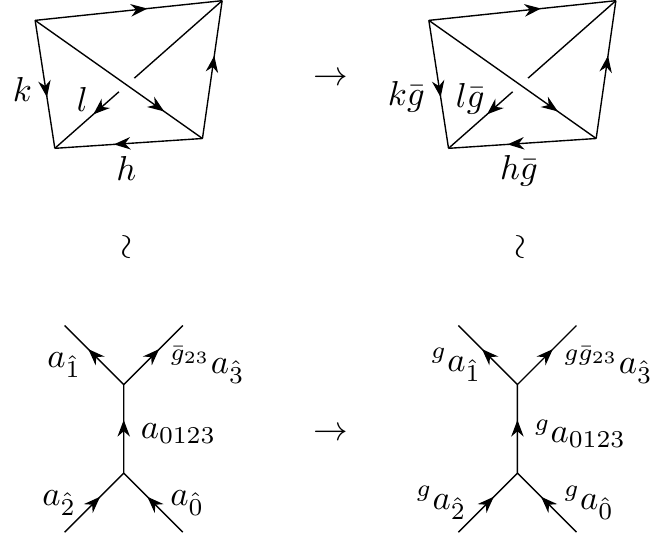}
}}
\end{align*}
when $B_c^g$ is applied. \\
The plaquette term $B_f^{g,a}$ right multiplies the group variable on face $f$, resulting in $h_f \bar g$, and fuses a loop of defect $a\in\cat_g$, oriented along the boundary of $f$, into the defects on edges $e\in \partial f$. The numerical amplitudes of $B_f^{g,a}$ are calculated according to the diagrammatic rules of the input UGxBFC.

The ground space of the Hamiltonian is supported on a subspace of states satisfying the vertex and edge constrains which is spanned by consistent diagrams from the UGxBFC. 
Naively these states and the Hamiltonian seem to depend on the choice of projection to the $2$D plane of the picture, up to a local unitary gauge equivalence due to $U$ and $\eta$. We note that such an apparent dependence does not appear when following the tensor network approach to produce a Hamiltonian which was described in Sec.\ref{gcrossedtnham}.

\subsubsection{The Hamiltonian on the $3$-Torus: BCC lattice}

To explicitly construct an important special case of the Hamiltonian on simple polyhedra we pick $\Gamma_\complex$ to be the regular cellulation of $T^3$ by permutohedra which is dual to the BCC triangulation $\complex$. We use the branching structure on $\complex$ obtained form Fig.\ref{branch} via translations.
The Hamiltonian is given in Eq.\eqref{orientedham}. We proceed to explicitly calculate the matrix elements of the plaquette term $B_f^{g,a}$ for the top face of a permutohedron by deriving the effect of fusing a defect loop $a$ onto the lattice 
\begin{align*}
\vcenter{\hbox{
\includegraphics[scale=1]{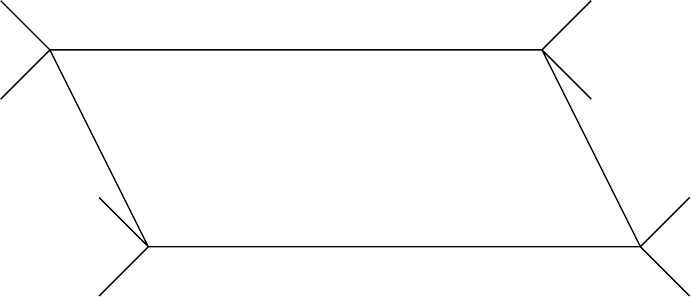}
}}
\end{align*}
the other terms are calculated similarly. 
First the 4-valent vertices are resolved using our choice of branching structure. On the resolved lattice we have some initial configuration $\sigma_\text{I}$ which is depicted on the left. Next the edges crossing $f$ are moved aside:
\begin{align*}
\vcenter{\hbox{
\includegraphics[scale=1]{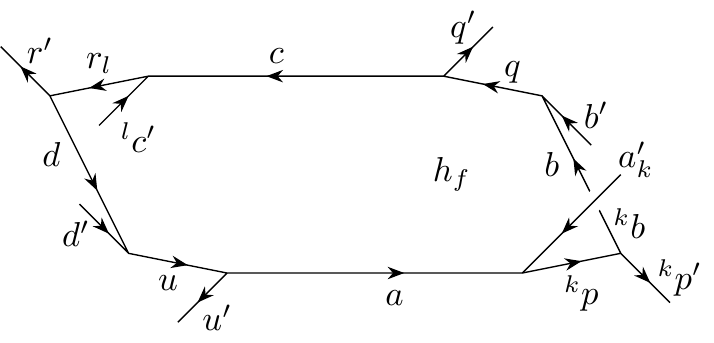}
}}
\rightarrow
\vcenter{\hbox{
\includegraphics[scale=1]{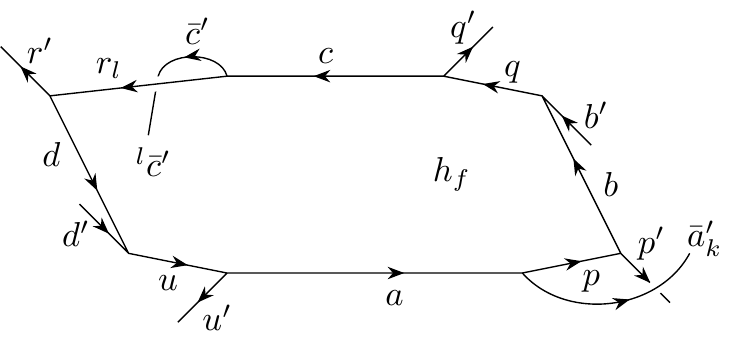}
}}
\end{align*}
this introduces a factor $U_{k^{-1}}^{-1}(\bar b,p;p') \overline{R_a^{p \bar{a}'}} R_c^{{}^{l}\bar{c}' r}$.
The plaquette term $B_f^{g,s}$ introduces a loop of defect $s_g$ onto the face $f$: 
\begin{align*}
\vcenter{\hbox{
\includegraphics[scale=1]{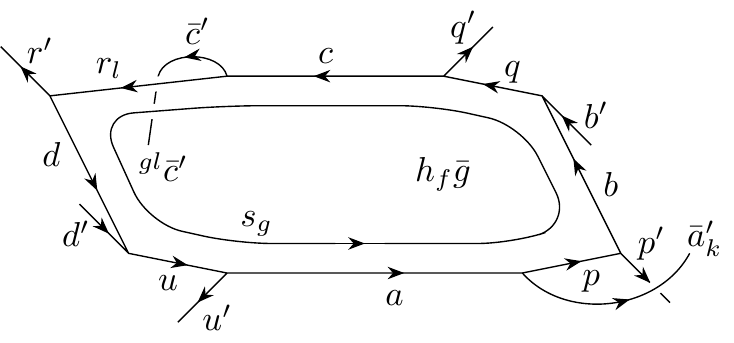}
}}
\end{align*}
 this is then fused into the lattice: 
 \begin{align*}
\vcenter{\hbox{
\includegraphics[scale=1]{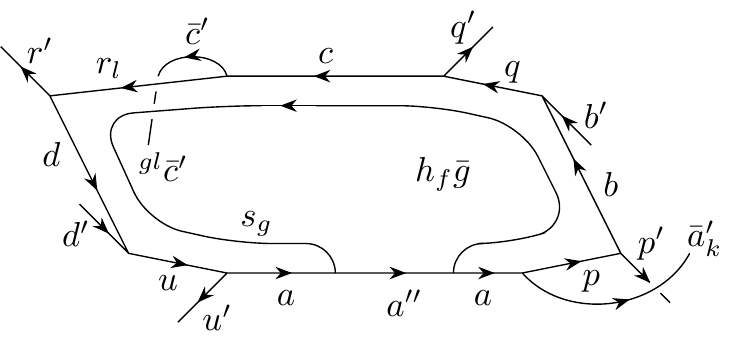}
}}
\rightarrow
\vcenter{\hbox{
\includegraphics[scale=1]{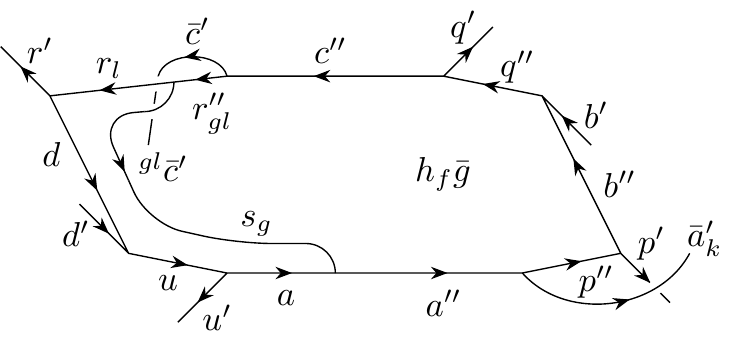}
}}
\end{align*} 
resulting in a factor $F^{\bar a''sp}_{a';\bar ap''}F^{\bar p''sb}_{\bar p';\bar pq''} F^{\bar b'' s q}_{b';\bar b q''} F^{\bar q'' s c}_{\bar q' ; \bar q c''} F^{\bar c'' s r}_{c';\bar c r''}$. Next we slide the $c'$ line under a vertex 
\begin{align*}
\vcenter{\hbox{
\includegraphics[scale=1]{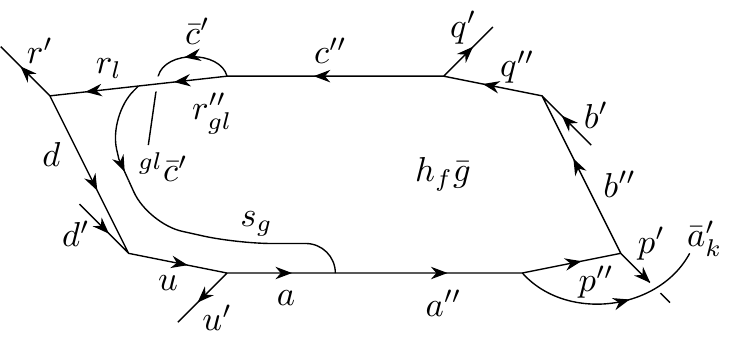}
}}
\end{align*}
yielding a factor $\eta_{{}^{gl} \bar c'}(g,l)$. 
Making three additional $F$-moves $F^{\bar r'' s d}_{\bar r'; \bar r d''} F^{\bar d'' s u}_{d'; \bar d u''} F^{\bar u'' s a}_{\bar u'; \bar u a''}$ leads to 
\begin{align*}
 \vcenter{\hbox{
\includegraphics[scale=1]{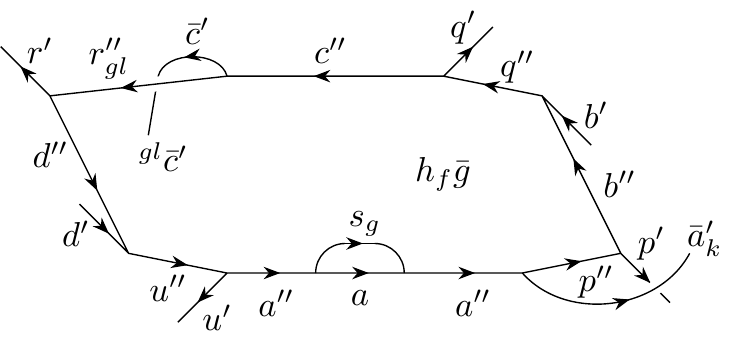}
}}
.
\end{align*}
Restoring the lattice to the initial position 
\begin{align*}
 \vcenter{\hbox{
\includegraphics[scale=1]{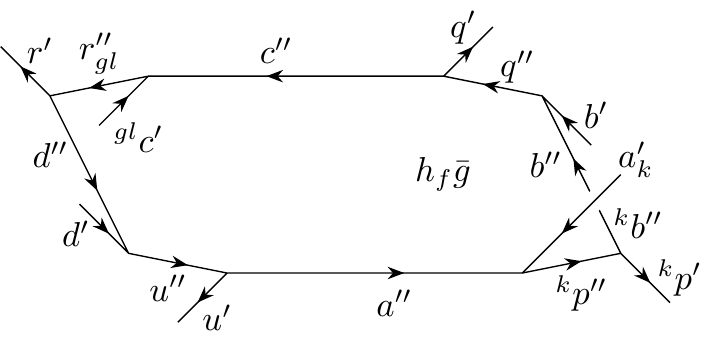}
}}
\end{align*}
yields the final configuration $\sigma_\text{F}$ and a phase $U_{k^{-1}}^{-1}(\bar b'',p'';p') R_{a''}^{p'' \bar{a}'} \overline{R_{c''}^{{}^{gl}\bar{c}' r''}}$. 
Hence the full matrix element of the plaquette term is given by:
\begin{align}
\bra{\sigma_\text{F}} B_f^{g,a} \ket{\sigma_\text{I}} =&
 U_{k^{-1}}^{-1}(\bar b,p;p') U_{k^{-1}}^{-1}(\bar b'',p'';p') \eta_{{}^{gl} \bar c'}(g,l)
  \overline{R_a^{p \bar{a}'}} R_c^{{}^{l}\bar{c}' r} R_{a''}^{p'' \bar{a}'} \overline{R_{c''}^{{}^{gl}\bar{c}' r''}}
\nonumber  \\
  &F^{\bar a''sp}_{a';\bar ap''}F^{\bar p''sb}_{\bar p';\bar pq''} F^{\bar b'' s q}_{b';\bar b q''} F^{\bar q'' s c}_{\bar q' ; \bar q c''} F^{\bar c'' s r}_{c';\bar c r''} 
  F^{\bar r'' s d}_{\bar r'; \bar r d''} F^{\bar d'' s u}_{d'; \bar d u''} F^{\bar u'' s a}_{\bar u'; \bar u a''} .
\end{align}

\subsection{Degeneracy, Statistics, and the Ground State Wave Function}

The GSWF of the Hamiltonian in Eq.\eqref{k2ham} admits a simple PEPS and MERA following the approach of Ref.\cite{higherdto}. 
For a triangulated space manifold $(\nan,\triang)$ the tensor network is given by 
\begin{align}
\tft{\triang*v_0}
\end{align}
where $v_0$ is an auxiliary vertex below the other vertices in the ordering. 
To normalize this state properly we use the convention that any weight associated to a simplex at the boundary $w(\simp_i)$ is included in the state after taking a square root $\sqrt{w(\simp_i)}$. This ensures that upon gluing along such a boundary the full weight is recovered. In particular
\begin{align}
\braket{\tft{\triang*v'_0}  \,|\, \tft{\triang*v_0}} = \tft{\{v_0',v_0\}* Y}
\end{align}
assuming the TQFT is unitary (Hermitian).
Note in our models these weights are always positive real numbers and the positive square root is chosen, if these weights are negative (or complex) such a convention is not straightforward (for example this occurs in Ref.\cite{freedman2015double}).
Similarly for the MERA consider a triangulated identity bordism on some space manifold $(\nan\times I, \triang')$ such that the triangulation at the space manifold $(\nan,0)$ reduces to the physical lattice $\triang$ and we pick a minimal triangulation $\triang''$ of $(\nan,1)$ at the `top' of the MERA corresponding to the ultra IR scale. Then upon fixing a vector containing the fully coarse grained topological information $\ket{t}$ the MERA is given by 
\begin{align}
\tft{\triang'}\ket{t}
\end{align}
The PEPS built this way has a virtual PEPO symmetry, to extract the physical superselection sectors corresponding to point and loop like excitations we expect a generalization of Ocneanu's tube algebra can be constructed directly in the tensor network as has been achieved in $(2+1)$D~\cite{bultinck2015anyons}.

There is an alternate approach to a tensor network description of the GSWF working directly with the diagrammatic representation of the anyons. For CYWW this yields roughly a $2.5$D tensor network representation that is presented with a particular choice of projection down to a plane but transforms trivially under changing this choice~\cite{anyondiag}. This approach encounters complications for the UGxBFC model due to the nontrivial action of anyon wordlines upon configurations behind them~\cite{barkeshli2014symmetry}. Hence it appears this approach may only produce a gswf that transforms with a local unitary upon changing the plane of projection. Note a similar complication may occur for a figure 8 worldine in CYWW, however this can be corrected with a careful labeling of the writhe of each anyon worldline.

In Section~\ref{background} we discussed how special cases of the UGxBFC recover an uncoupled DW and CY theory, for trivial grading and group action, or a $2$-group gauge theory, for a categorical group input. We also pointed out that one can add $H^3$ and $H^4$ cocycles to the data of the ungraded case to realize a general $2$-group gauge theory with cohomology twist. This suggests an interpretation of the model (at least in the ungraded case) as a theory of anyons coupled to a $2$-group gauge theory.

\section{Discussion}
\label{disc}

\begin{table*}[th]
\begin{center}
\scriptsize
\begin{tabular}{| >{\centering\arraybackslash}m{3.6cm} | >{\centering\arraybackslash}m{2.9cm} | >{\centering\arraybackslash}m{1.8cm} | >{\centering\arraybackslash}m{1.4cm} | >{\centering\arraybackslash}m{2.2cm} | >{\centering\arraybackslash}m{2.4cm} | }
\hline 
  State Sum & Hamiltonian & Input data &  ST Dimension D & Sensitivity & Physical excitations 
\\ \hline\hline
Trivial/Invertible theory & Trivial paramagnet & - & all D & Classical local invariants: Euler characteristic, signature, \dots & Local excitations, no nontrivial superselection sectors
\\\hline
GHZ TQFT (Includes all $2$D TQFTs based on Frobenius Algebras~\cite{fukuma1994lattice}) & Symmetry breaking & $n\in\N$ & all D & $\pi_0$ & Domain wall excitations
\\ \hline
n-group Dijkgraaf-Witten gauge theory~\cite{dijkgraaf1990topological,kapustin2013higher} or Yetter homotopy n-type~\cite{yetter1993tqft} (Includes Birmingham-Rakowski model~\cite{birmingham1995state} and Mackaay's group examples~\cite{mackaay2000finite}) & Higher group lattice gauge theory (includes twisted quantum doubles~\cite{kitaev2003fault,hu2013twisted,wan2015twisted}, generalized toric codes, 2-group gauge theory \cite{bullivant2016topological} and Yoshida's models~\cite{yoshida2015topological}) & n-group $\G$ \& D-cocycle $\alpha\in H^{\text{D}}(B\G,U(1))$ & all D$\geq$n & n-homotopy type (or n-Postnikov system)
& Gauge charges, fluxes etc..
\\ \hline
Turaev-Viro~\cite{turaev1992state} (dropping semisimplicity assumption gives Kuperberg \& Barrett-Westbury invariants~\cite{kuperberg1991involutory,barrett1996invariants,barrett1995equality}.) & Levin-Wen string-net model~\cite{levin2005string} & UFC $\cat$~\cite{etingof2015tensor} & $(2+1)$ D & PL homeomorphism 
& $Z(\cat)$ anyon theory
 \\ \hline
Crane-Yetter~\cite{crane1993categorical,crane1997state} (captures unitary Broda, Petit, Barenz-Barett dichromatic state sums~\cite{broda1993surgical,roberts1996refined,petit2008dichromatic,barenz2016dichromatic} via chainmail construction~\cite{roberts1995skein})  & Walker-Wang model~\cite{walker2012} & UBFC $\cat$~\cite{etingof2015tensor} & $(3+1)$ D & $\pi_1,\,w_2$
& Bosons and fermions and loop excitations (only for nonmodular $\cat$)
 \\ \hline
 Crane-Frenkel~\cite{crane1994four} \& Carter-Kauffman-Saito~\cite{carter1999structures}  & ? & Hopf category and cocycle  & $(3+1)$ D & homotopy ?
& ?
 \\ \hline
Kashaev TQFT~\cite{kashaev2014simple,kashaev2015realizations} & Kashaev model & $\Z_n$ & $(3+1)$ D & $\pi_1,\, w_2$  & Fermions (bosons) and loop excitations for $N=2\, (0) \mod 4$ (trivial for $N$ odd)
  \\ \hline
UGxBFC~\cite{shawnthesis} (includes Mackaay's spherical $2$-category models~\cite{mackaay1999spherical}) & UGxBFC Hamiltonian & UGxBFC $\cat$~\cite{turaev2000homotopy,kirillov2004g,etingof2009fusion,barkeshli2014symmetry} & $(3+1)$ D & $w_2$, homotopy 3-type ? 
& Bosons, fermions and loop excitations
  \\ \hline
Conjectural $n$-category TQFT (semi simplicity condition corresponds to having a single object) & $(n-1)$-membrane net Hamiltonian & Unitary $n$-category $\cat$ & all D$=(n+1)$  & PL homeomorphism (except for D$=4$)
& higher categorical center $Z(\cat)$
    \\ \hline
\end{tabular}	
\end{center}
\end{table*}

In this section we aim to place the new models into the broader context of previously constructed state sum TQFTs~\cite{crane1993categorical,crane1997state,kashaev2014simple,kashaev2015realizations,carter1999structures,crane1994four,yetter1993tqft, mackaay2000finite,mackaay1999spherical,kapustin2014topological,kapustin2013higher}. In doing so we sketch the general framework for state sum TQFTs and explain how the UGxBFC model fits into this. We also describe the relation of the UGxBFC model to other classes of $(3+1)$D state sums and their boundary physics.

It is conventional wisdom that an $n$-category describes a local or fully extended TQFT restricted to the disc~\cite{baez1995higher,lurie2009classification,walker1991witten}. This correspondence is materialized by the general prescription to construct an $(n+1)$D state sum model from an $n$-category~\cite{higherdto}. The recipe dictates that the i-simplices of a triangulation are labeled by i-morphisms of the $n$-category along with an (n-1)-associator for each (n+1)-simplex, a tensor satisfying the Pachner move equations. To make contact with familiar examples, first in $(1+1)$D, one can view the morphisms of a linear category with a single object as an associative algebra. Decorating the edges of a triangulated surface with these morphisms and assigning the structure coefficients to each triangle recovers the familiar Frobenius algebra TQFTs~\cite{fukuma1994lattice}. 
For $(2+1)$D consider a $2$-category with a single object, the 1-morphisms and 2-morphisms can be identified with the objects and morphisms of a fusion category respectively. Using these to label the edges and faces of a triangulated 3-manifold and assigning F-symbol associators to each tetrahedron recovers the Turaev-Viro TQFTs (Levin Wen string nets)~\cite{levin2005string,turaev1992state}. 
In $(3+1)$D consider a $3$-category with a single object and single 1-morphism, the 2- and 3- morphisms can be identified with the objects and morphisms of a braided fusion category. Using these to label the faces and tetrahedra of a triangulated 4-manifold and assigning a 15j-symbol to each 4-simplex recovers the Crane-Yetter TQFTs (Walker-Wang models)~\cite{crane1997state,crane1993categorical,walker2012}. 

These examples display the general pattern that adding structure to an $n$-category is often equivalent to shifting all the morphisms up a level while introducing a single object. From this point of view it is natural that to resolve the UV anomaly that prevents a $(2+1)$D (commuting projector Hamiltonian) lattice realization of a chiral anyon theory one should consider the boundary of a $(3+1)$D theory. This is precisely what the WW model achieves. It also suggests that to realize the most general $(3+1)$D topological orders (with excitations described by a unitary braided fusion $2$-category) with commuting projector Hamiltonians on the lattice one must similarly consider boundary theories of $(4+1)$D state sums.

\begin{figure*}[tb]
\center
{{
 \includegraphics[width=0.6\linewidth]{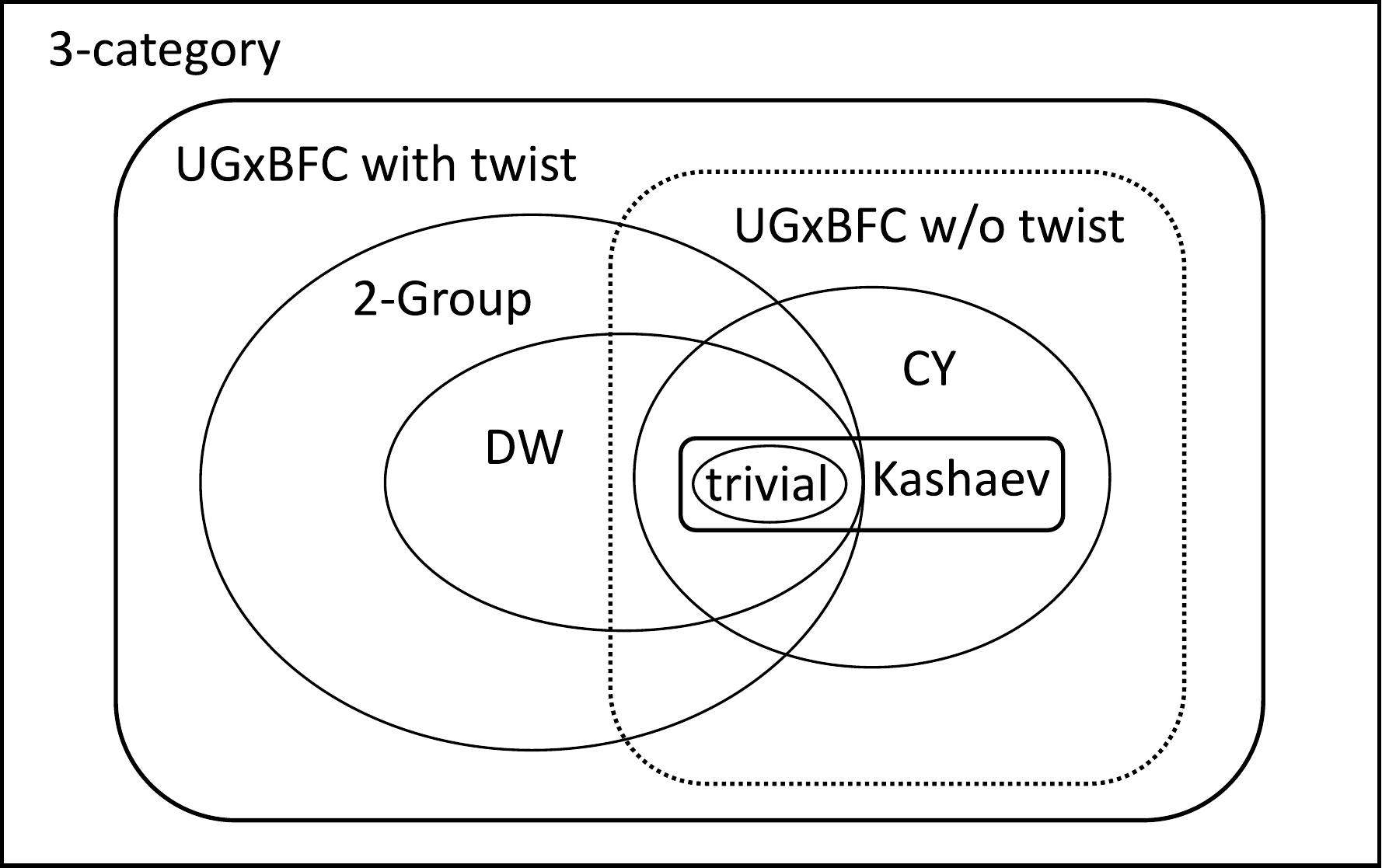}}} 
\caption{Conjectured relations of unitary $(3+1)$D state sum TQFTs (``twist'' refers to the addition of various cocycle functions).}
\label{4dstatesums}
\end{figure*}

From the examples above the UGxBFC models appear to be the natural generalization of TV to $(3+1)$D as they correspond to a $3$-category with a single object (which can be thought of as a $(2+1)$D topological phase) the 1-morphisms are identified with group elements (invertible gapped boundaries of the $(2+1)$D phase), 2-morphisms between the 1-morphisms $g_1$ and $g_2$ correspond to defects in the sector $\cat_{g_1^{-1}g_2}$ and 3-morphisms are the regular morphisms in the UGxBFC. 
Following the recipe, edges are labeled by group elements, triangles by defects and tetrahedra by UGxBFC morphisms while 4-simplices come with a 15j-symbol.

Notice in all the examples thus far we have considered only n-categories with a single object, loosening this requirement seems to correspond to dropping the assumption of semi simplicity (or possibly simplicity of the unit object)~\cite{kuperberg1991involutory,chang2015enriching}. We conjecture all non chiral topological phases of finite spin models can be realized by a state sum construction from an $n$-category with a single object. One way to extend the UGxBFC model might be to include multiple objects in the $4$-category, this naturally corresponds to considering boundaries between different topological phases rather than the same phase (this corresponds to a tricategory of bimodules which is a Gray category~\cite{carqueville20163}). 

Throughout the paper we have considered anomaly free SETs described by a UGxBFC (the state sums for these are rigorously constructed in Ref.\cite{shawnthesis}), such SETs can be realized purely in $(2+1)$D systems with ultra-local symmetry actions. From the perspective of using $(3+1)$D models to realize anomalous boundary phases it is natural to consider extensions of the model to UGxBFCs with non vanishing $H^3$ and $H^4$ anomalies. 
In the case of a trivial grading (all nontrivial defect sectors empty) the labeling of the triangulation defines a flat 2-group connection. Using the language of Ref.\cite{kapustin2013higher} the 2 group specified by $\Pi_1=G,\,\Pi_2=\A$ (the abelian anyons) with a group action $\rho$ inherited form the UGxBFC and trivial 3-cocycle. It is possible to augment this construction by adding in a 3-cocycle $\beta \in H^3_\rho(G,\A)$, which alters the flatness condition to that of a different 2-group $\G$. Furthermore one may add in a 4-cocycle $H^4(B\G,U(1))$. Hence these trivially graded models can be understood as a theory of bosons or fermions coupled to a 2-group gauge field. This generalizes the picture of CYWW models as a theory of bosons or fermions coupled to a discrete group gauge field. We leave the details of this to future work~\cite{anomalousugxbfc}.

From this perspective the ungraded UGxBFC models have a deequivariantized or ungauged counterpart given by a 2-group $\G$-SPT with ultra-local symmetry action. The boundaries of these models can support all anomalous $(2+1)$D SET phases as the bulk serves to resolve the chiral, $H^3$ and $H^4$ anomalies~\cite{thorngren2015higher}. Upon gauging or equivariantizing the $\G$ symmetry of these SET models one recovers the ungraded UGxBFC model. 
These possible additions suggest the intrinsic $H^3$ and $H^4$ classes of an anomalous UGxBFC should be treated as a torsor, as they can be shifted by an arbitrary choice in the ungraded case, although it is unclear how this carries over to the general graded case.

Thus far we have explained how the 2-group and CYWW models are captured as subcases of the UGxBFC construction, furthermore we believe the Kashaev TQFT is equivalent to a subset of the CYWW model and hence is also captured. We have outlined what is conjectured to be the most general construction of a $(3+1)$D state sum in terms of a $3$-category. We made the case that restricting to models that have a single object is expected to capture all topological orders with a commuting projector Hamiltonian that admit a TQFT description. These topological orders are also known as gapped quantum liquids~\cite{zeng2015gapped} and they do not include models such as Haah's cubic code~\cite{haah2011local,yoshida2013exotic} even though it admits a generalized gauge theory description~\cite{PhysRevB,vijay2016fracton}. 
The UGxBFC model captures a very general case of the single object $3$-category state sum construction, and most importantly the construction comes with a wealth of examples originating from SET phases in $(2+1)$D~\cite{barkeshli2014symmetry,hung2013quantized,mesaros2013classification,teo2015theory,bombin2010topological,turaev2000homotopy,kirillov2004g,etingof2009fusion,tarantino2015symmetry}.

To assess whether the UGxBFC model truly goes beyond the preexisting constructions one would ideally construct the irreducible excitations and compare their full set of physically accessible topological invariants. In general the construction of the excitations should correspond to taking the Drinfeld double or 2-categorical center of the input treated as a unitary fusion $2$-category~\cite{muger2003subfactors,muger2003subfactors2}, this itself is not well understood. The resulting invariants are also not fully understood, although very interesting progress has been made~\cite{moradi2015universal,wang2016quantum} particularly on the 3-loop braiding statistics~\cite{jiang2014generalized,wang2014braiding,wang2015topological,wang2015non}. In principle these invariants should uniquely specify the unitary braided fusion $2$-category describing the physical excitations, however it has not even been rigorously shown that the commonly used $S$ and $T$ matrices are in 1-1 correspondence with UBFCs in the $(2+1)$D setting. 

We may resort to comparing the boundary physics of the proposed UGxBFC models to previous constructions, but as we have seen the relevant boundaries can be understood as coming from a CYWW model coupled to a 2-group gauge field~\cite{thorngren2015higher}. 

Another avenue is to focus on the closed manifold partition functions of the theory. This approach, for example, allows one to differentiate the Turaev-Viro models from Dijkgraaf-Witten in $(2+1)$D as the former is sensitive to PL homeomorphism, while the latter depends only on homotopy.
However in $(3+1)$D the situation is complicated by the fact that it is fundamentally impossible for a unitary TQFT to detect all inequivalent smooth structures on homotopic or s-cobordant manifolds~\cite{freedman2005universal} (this is a consequence of the existence of $3$D boundary diffeomorphisms that do not extend into the $4$D bulk). 
Here we should note that the equality (equivalence) of all partition functions is not known to be a sufficient condition for two theories to be equivalent. 
That being said it has been suggested that the UGxBFC state sum depends on the homotopy 3-type of a manifold~\cite{shawnthesis}, it can also be seen to depend on some Stiefel-Whitney classes of a manifold as it includes the CYWW model which can involve fermions that are sensitive to a choice of spin structure. These dependencies are consistent with the interpretation of the UGxBFC model as bosons or fermions coupled to a higher group gauge theory. It is currently unclear if the UGxBFC models with nontrivial grading give rise to more general invariants, we plan to study this in future work~\cite{provingugxbfc}. 

Finally let us clarify that for the UGxBFC models with non empty defect sectors, the TQFT constructed from the $15j$-symbol does not depend on extra structure or decoration of the cobordism category (beyond possibly an orientation). That is to say the theory is not an SET involving physical defects of some global symmetry. However it may be possible that the boundary theory can be thought of as an SET with a certain configuration of defects specified as a boundary condition.
\\ \\ \\
\emph{Acknowledgments -} The authors acknowledge Meng Cheng, Shawn Cui, Xie Chen, Sujeet Shukla, Jeongwan Haah, Maissam Barkeshli and Parsa Bonderson for helpful discussions and comments. We particularly thank Shawn Cui for sharing a draft of his unpublished thesis. 
DW acknowledges support from the Austrian Marshall Plan foundation. 
ZW was partially supported by NSF grants DMS-1108736 and DMS-1411212.

\bibliography{Hamiltonian_Realizations_of_(3+1)-TQFTs}

\pagebreak
\appendix

\section{Elementary combinatoric topology}
\label{comb topology}

In this appendix we introduce some basic notions from the field of combinatorial topology that are used throughout the paper. We recommend Ref.\cite{hatcher606algebraic} for further reading.

\begin{figure*}[ht]
\center
{{
 \includegraphics[width=0.25\linewidth]{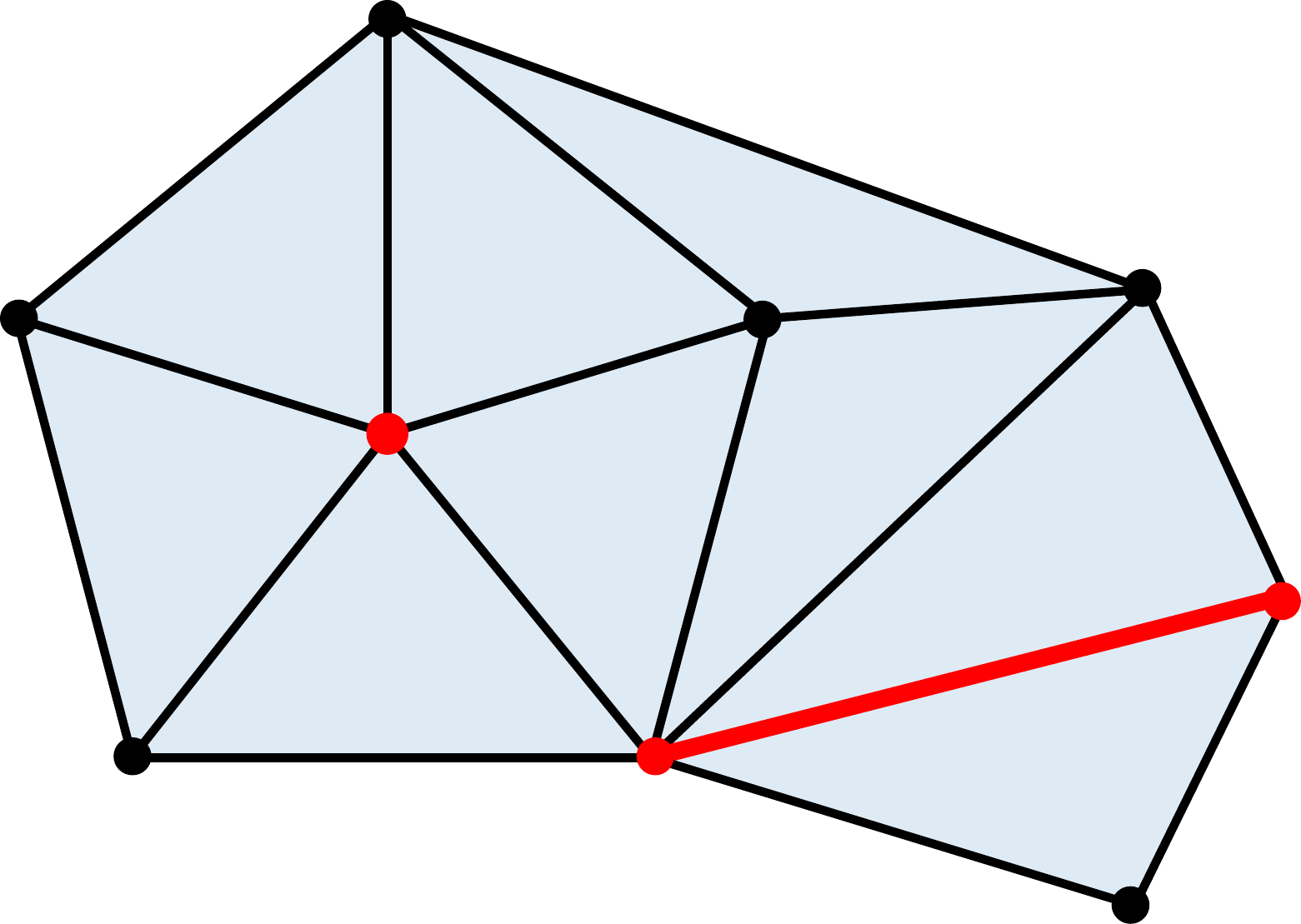}}}
 \quad 
  \includegraphics[width=0.25\linewidth]{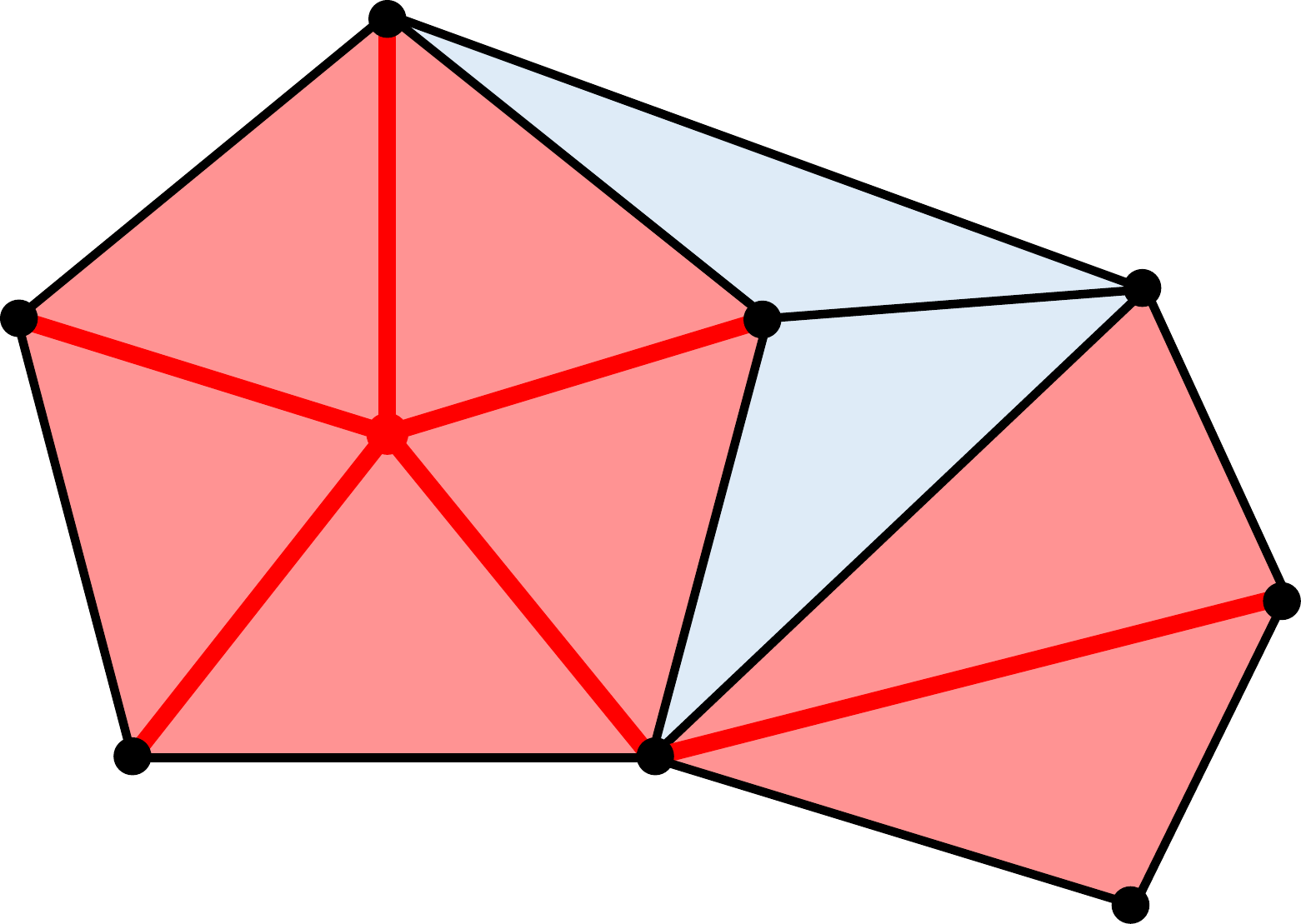}
  \quad
   \includegraphics[width=0.25\linewidth]{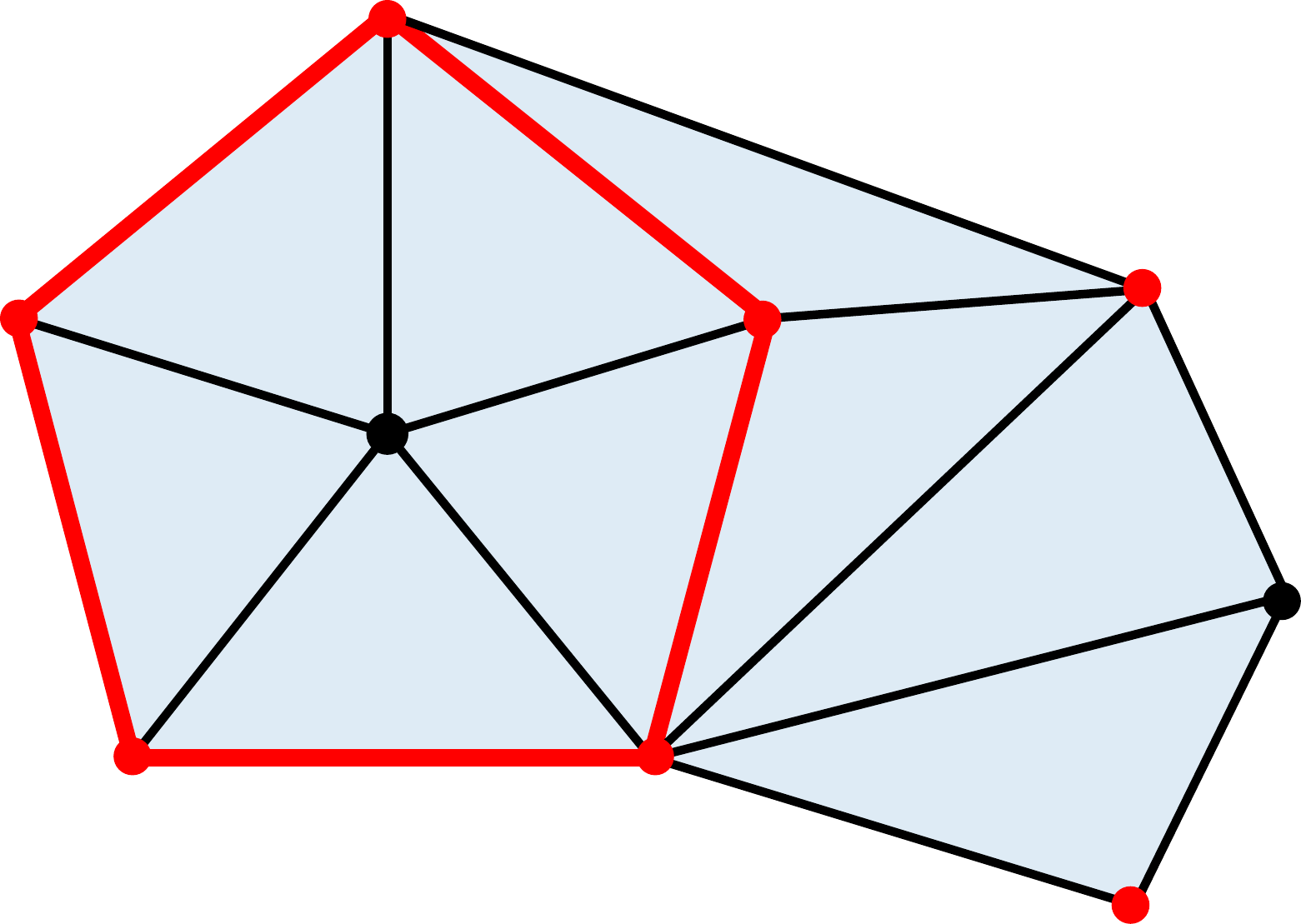}
\caption{Red indicates the closure (left), star (middle) and link (right) of a vertex and an edge.}
\label{pltop}
\end{figure*}
\begin{itemize}
\item
An {\bf n-simplex} $\simp_{n}=[v_0\dots v_n]$ is the convex hull of a set of (n+1) points $v_0,\dots,v_n\in\R^{m\geq n}$ such that the vectors $v_1-v_0,\dots,v_n-v_0$ are linearly independent. The orientation of a simplex can be defined as $\orient{\simp_{n}}:=\sgn(\det (v_1-v_0,\dots,v_n-v_0))$ and satisfies $\orient{[v_0\dots v_n]}=\sgn(\pi)\orient{[v_{\pi(0)}\dots v_{\pi(n)}]}$ for a permutation $\pi$.
We focus on the combinatorial aspects of simplices, notice the convex hull of any subset of vertices $[v_{i_0} \dots v_{i_j}]$,  for $0 \leq j \leq n$, is a $j$-subsimplex of $\simp_n$. Hence $\simp_n$ contains $\binom {n+1} {j+1}\ j$-subsimplices.
\item
A {\bf simplicial complex} $\complex$ is a union of simplices loosely defined as a subset of the power set of $(N+1)$ points $P\{v_0,\dots,v_N\}$ such that $\simp\in \complex \implies P(\simp)\subseteq \complex$. We exclusively deal with homogeneous simplicial complexes of some dimension D, which can be thought of as a union of D-simplices, in which all $\complex$-simplices, for $\dim \complex<$D, appear as a subsimplex of a D-simplex. The usual definition of simplical complex requires the intersection of any pair of simplices $\simp_p\cap\simp_q$ to be a subsimplex of $\simp_p$ and $\simp_q$. We also use a weaker notion referred to as a $\simp$-complex in which the intersection of a pair of simplices may consist of multiple subsimplices.
\item
The {\bf underlying space} of a simplicial complex $\complex$ is given by the union of all its simplices (treated as a topological space) denoted by $|\complex|$.
\item
The {\bf k-skeleton} of a simplicial complex $\complex$, denoted $\complex_k$ is the union of j-subsimplices $\simp_j\in \complex$ with $j\leq k$.
\item
We often assume the vertices of a simplicial complex have been ordered, this induces an orientation on the edges of the 1-skeleton from lesser to greater adjacent vertex. This orientation is a {\bf branching structure} since the edges on the boundary of a triangle never form a similarly oriented cycle. In fact our arguments only require a branching structure which is a local condition slightly weaker than a global ordering, although we will sometimes assume a global ordering for convenience.
\item
The {\bf boundary} of a D-simplicial complex $\complex$ is a $(\text{D}-1)$-simplicial complex $\partial \complex$ consisting of all $\simp_{\text{D}-1}\in \complex$ that are the subsimplex of a single D-simplex within $\complex$. Note $\partial\circ \partial = 0$.
\item
The {\bf closure} of a collection of simplices $\subc \subseteq \complex$ is given by $\cl_\subc $ the minimal subcomplex of $\complex$ containing $\subc$. 
\item 
The {\bf interior} of a subcomplex $\subc \subseteq \complex$ is given by $\interior \subc:= \cl_\subc \wo \subc$. 
\item 
The {\bf star} of a subcomplex $\subc \subseteq \complex$ is given by $\st{\subc}$ the union of simplices in $\complex$ which have a subsimplex contained in $\subc$. 
\item 
The {\bf link} of a subcomplex $\subc \subseteq \complex$ is given by $\lk{\subc}:=\cl_{\st{\subc}}\wo \st{\cl_\subc}$ 
\item 
The {\bf join} of two simplices $\simp_{n}=[v_0\dots v_n],\simp_{m}=[v_{n+1},\dots v_{n+m+1}]$ is the simplex $\simp_{n}*\simp_{m}=[v_0\dots v_{n+m+1}]\simeq\simp_{n+m+1}$.
The join of two simplicial complexes $\complex,\subc$ is given by $\complex*\subc$ the union of all $\simp_{i}*\simp_{j},\ \forall \simp_{i}\in\complex,\forall\simp_{j}\in\subc$ (note this includes joins with the empty simplex $\emptyset*\simp_{j}=\simp_{j}$). 
Note the join is associative and commutative (possibly up to orientation reversal). There is a simple relation for any simplex $\simp_{i}\in\complex$ given by $\st{\simp_{i}}=\simp_{i}*\lk{\simp_{i}}$.
\item
The {\bf cone} of a simplicial complex $\complex$ is its join with a point $v$ given by $v*\complex$.
\item
The {\bf suspension} of a simplicial complex $\complex$ is its join with two points $v,v'$ given by $\{v,v' \}*\complex$.
\begin{figure*}[t]
\center
\begin{equation*}
 \vcenter{\hbox{\includegraphics[height=0.12\linewidth]{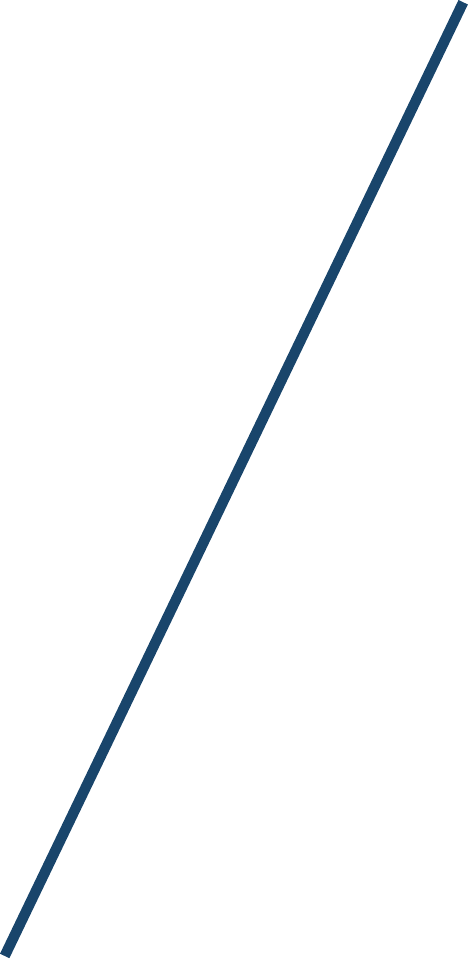}}} *\ \vcenter{\hbox{\includegraphics[width=0.09\linewidth]{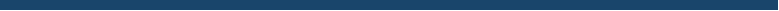}}} =  \vcenter{\hbox{\includegraphics[width=0.14\linewidth]{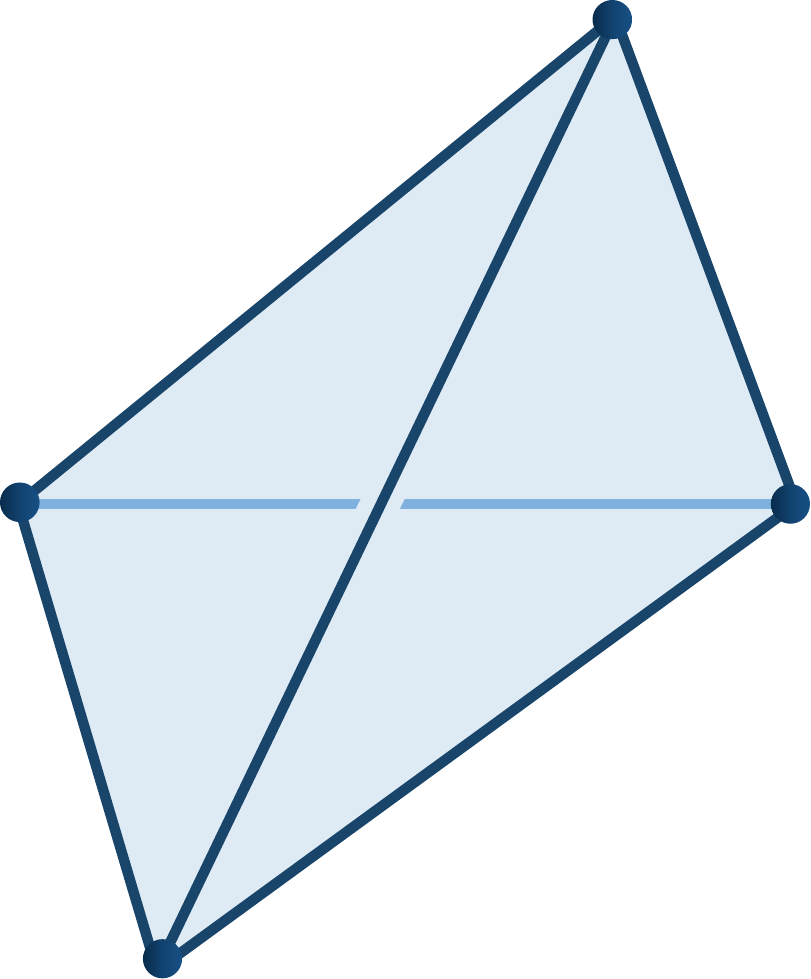}}}\ , 
 \quad \quad
  \vcenter{\hbox{\includegraphics[height=0.008\linewidth]{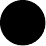}}} \ *\ \vcenter{\hbox{\includegraphics[height=0.08\linewidth]{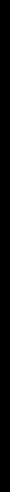}}} =  \vcenter{\hbox{\includegraphics[height=0.12\linewidth]{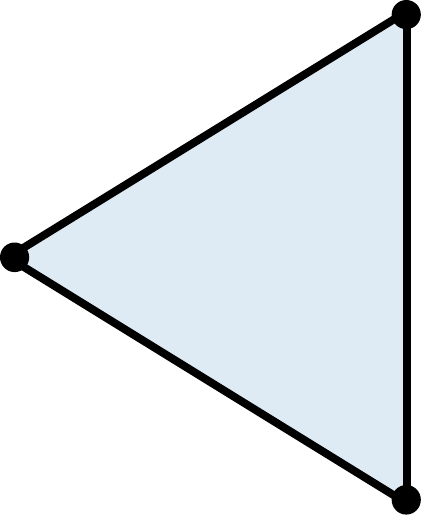}}}\ ,  
  \quad  \quad
\{  \vcenter{\hbox{\includegraphics[height=0.008\linewidth]{figs/fig22}}} , \vcenter{\hbox{\includegraphics[height=0.008\linewidth]{figs/fig22}}} \} \ *\ \vcenter{\hbox{\includegraphics[height=0.08\linewidth]{figs/fig23}}} =  \vcenter{\hbox{\includegraphics[height=0.12\linewidth]{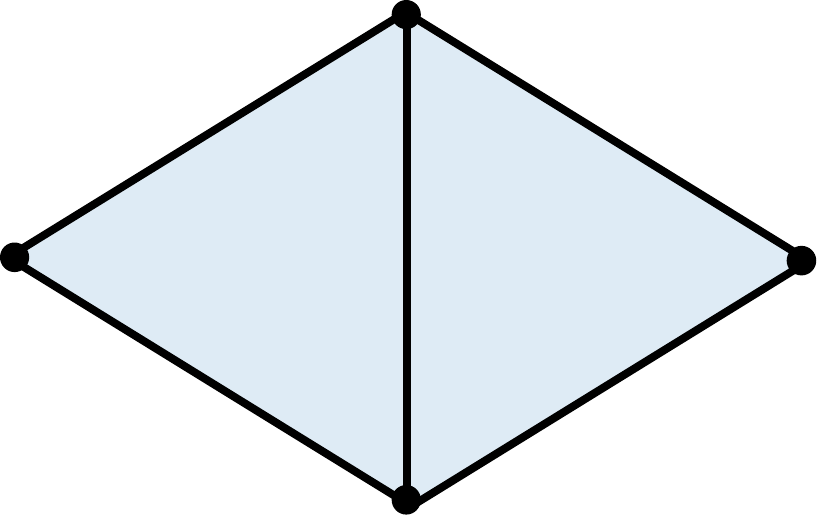}}} 
 \end{equation*}
\caption{The join of two edges (left), cone of an edge (middle) and suspension of an edge (right).}
\label{cone}
\end{figure*}
\item 
A {\bf bistellar flip} ({\bf Pachner move}) on any k-simplex $\simp_{k}\in\complex$ is constructed from an auxiliary (n-k)-simplex $\simp_{n-k}\notin\complex$ by taking $(\complex \wo\st{\simp{_k}})\cup_{\lk{\simp_{k}}}(\partial\simp_{k}*\simp_{n-k})$ with the identification $\lk{\simp_{k}} \simeq \partial\simp_{k}*\partial \simp_{n-k}$.
\begin{figure*}[th]
\center
\begin{equation*}
{{
 \vcenter{\hbox{\includegraphics[width=0.2\linewidth]{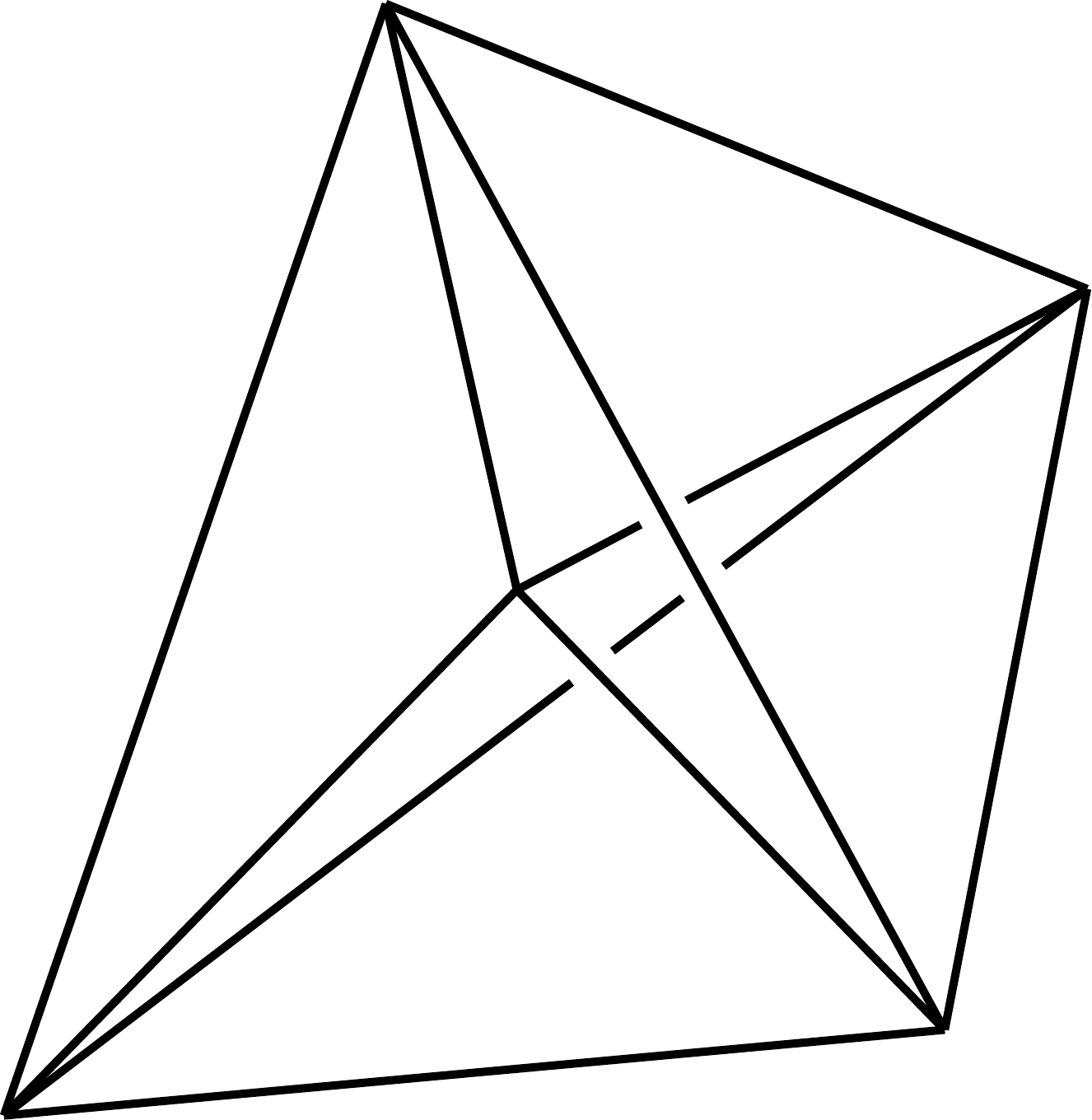}}}}} \quad \leftrightarrow  \vcenter{\hbox{\includegraphics[width=0.2\linewidth]{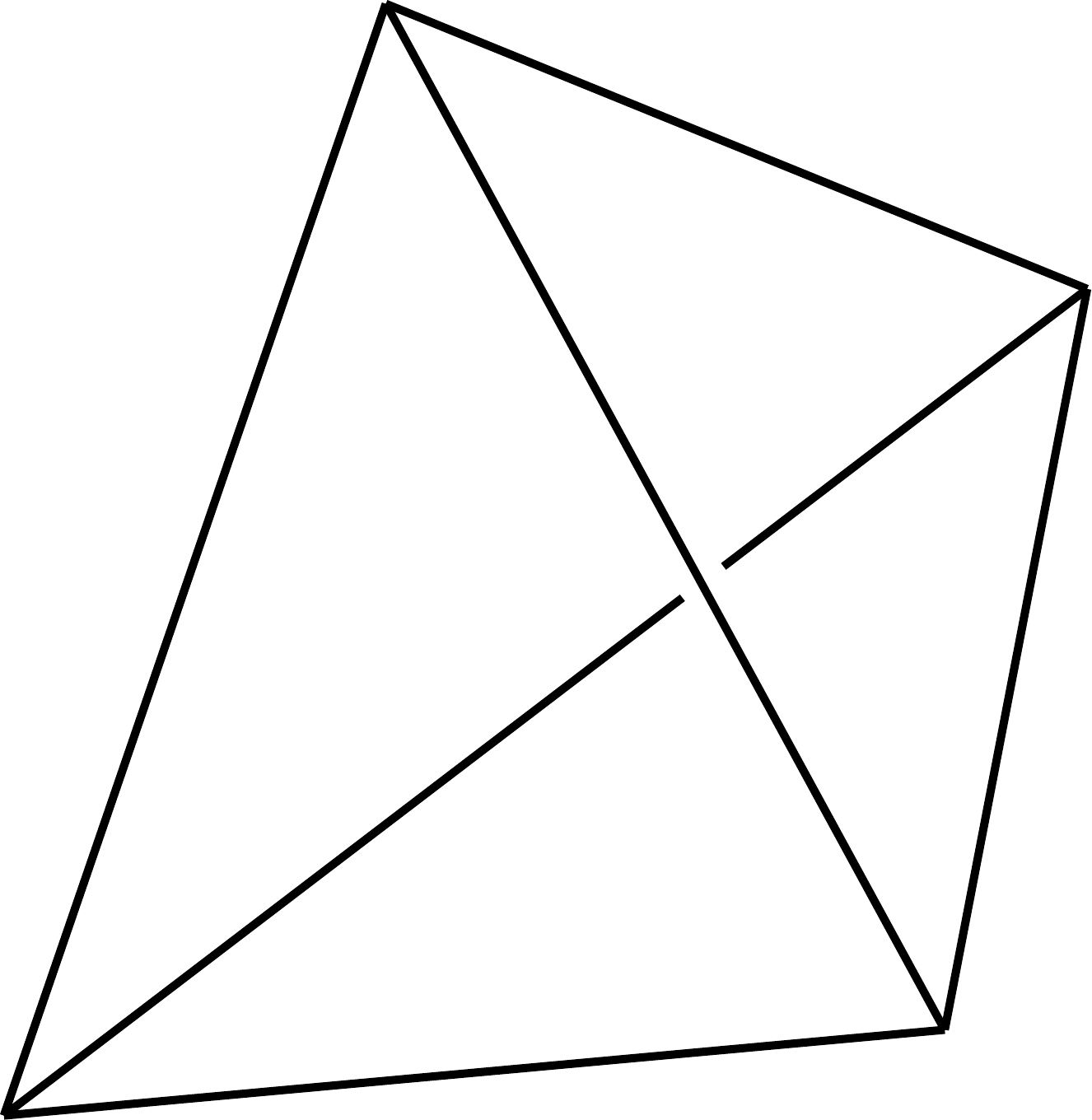}}}
 \quad\quad\quad\quad
 \vcenter{\hbox{\includegraphics[width=0.17\linewidth]{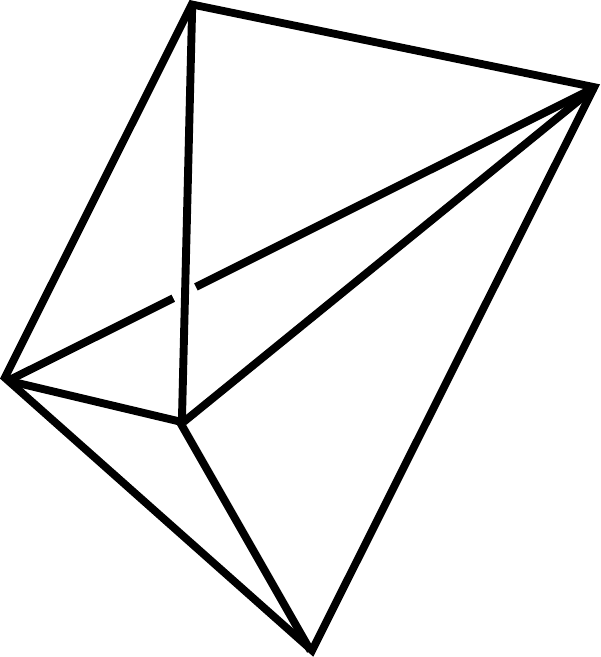}}}  \leftrightarrow \quad  \vcenter{\hbox{\includegraphics[width=0.17\linewidth]{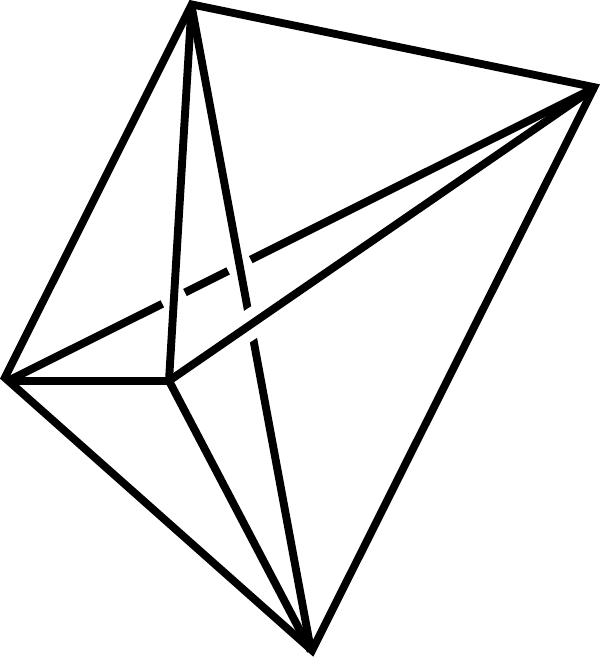}}}
\end{equation*}
\caption{The bistellar (Pachner) moves in $3$D.}
\label{pmoves}
\end{figure*}
\item 
The {\bf Poincar\'e dual} of an n-dimensional simplicial complex is an n-dimensional simple polytope with an (n-k)-cell for each k-simplex of the simplicial complex.
To construct this dual one can start by associating a vertex $v^i$ to each simplex $\simp_{n}^i\in\complex$ then iteratively adding a $j$-cell $\comp{j}$ for $j=1,\dots,n$ for each (n-j)-simplex $\simp_{n-j}\in\complex$.
For $\simp_{n-j}^i\in\complex$ we add a $j$-cell $\comp{j}^i$ with $k_l$ (j-l) faces, where $k_l$ is the number of (n-j+l)-simplices that intersect $\simp{(n-j)}$ in $\complex$.
Each (j-l) face of $\comp{j}^i$ is glued to the (j-l)-cell that is dual to the
corresponding (n-j+l)-simplex intersecting $\simp_{n-j}$.
\item
A piecewise linear (PL) manifold is a topological space equipped with an atlas of coordinate charts such that the transition functions between charts are piecewise linear. Similarly a smooth manifold is a topological space with an atlas of coordinate charts such that the transition functions are smooth.
({\bf Top, PL, Smooth}) is the category of (topological, PL, smooth) manifolds and (continuous, PL continuous, differentiable) maps between them, a (homeomorphism, PL homeomorphism, diffeomorphism) between spaces defines an equivalence. 
Not all topological manifolds admit a PL structure, an example of minimal dimensionality being Freedman's $E_8$ 4-manifold, and those which do may admit infinitely many inequivalent PL structures, lowest dimensional examples are exotic $\R^4$'s due to Freedman, Donaldson and Taubes. Note the existence of exotic 4-spheres is unknown and would provide a counter example to the $4$D smooth Poincar\'e conjecture.
 Similarly not all PL manifolds admit a smoothing, examples of minimal dimension 8 were discovered by Ells, Kuiper and Tamura, and those which do may admit multiple inequivalent smooth structures, minimal dimensional examples given by Milnor's exotic 7-spheres.
In general we have Smooth$\subseteq$PL$\subseteq$Top, while for D$\leq 6$ Smooth$\simeq$PL, for D$=7$ the inclusion Smooth$\subseteq$PL is surjective but not injective and for D$\geq 8$ the inclusion is neither injective nor surjective. For D$\leq 3$ PL$\simeq$Top while for D$\geq 4$ the inclusion is neither injective nor surjective.
\item
A {\bf triangulation} of a topological manifold $\man$ is a simplicial complex $\complex$ together with a homeomorphism $\phi:|\complex|\rightarrow \man$.
\item 
A {\bf PL triangulation} of a topological manifold $\man$ is a simplicial complex $\complex$ together with a homeomorphism $\phi:|\complex|\rightarrow \man$, satisfying the extra constraint that the link of any vertex $\lk{v},\,v\in\complex$ is homeomorphic to a PL (n-1)-sphere (not merely a homotopy sphere).
For D$\leq 4$ all triangulations are PL, while for D$\geq$5 one can construct a non PL triangulation by taking two (or more) suspensions of a triangulated Poincar\'e sphere. According to the discussion above all smooth manifolds admit unique PL triangulations, while topological manifolds admit unique PL triangulations for D$\leq 3$ and may admit anywhere from 0 up to an infinite family of inequivalent PL triangulations for D$\geq 4$. 
For D$\geq 5$ it was shown by Manolescu that there are manifolds that do not admit PL triangulations but do admit the weaker notion of triangulation.
\end{itemize}

\end{document}